\documentclass[10pt,journal,onecolumn,a4paper]{IEEEtran}
\usepackage{amsfonts,amsmath,amssymb,amstext,amsthm,cite,color,dsfont,enumerate,hyperref,makeidx,verbatim,xfrac,xr,url}
\usepackage{pict2e}
\usepackage[mathscr]{euscript}
\usepackage[geometry]{ifsym}
\hypersetup{
pdftitle		={The Sphere Packing Bound For Memoryless Channels},
pdfauthor		={Bar{\i}s Nakiboglu},}
\title{The Sphere Packing Bound \\ For Memoryless Channels}
\author{Bar\i\c{s} Nakibo\u{g}lu\\ \small{\href{mailto:bnakib@metu.edu.tr}{bnakib@metu.edu.tr}}
	\thanks{This paper was presented in part at 
		the 2017 IEEE International Symposium on Information Theory
		\cite{nakiboglu17}.}
}
\theoremstyle{plain}
\newtheorem{lemma}{Lemma} 
\newtheorem{theorem}{Theorem}

\newtheorem*{conjecture*}{Conjecture}
\newtheorem{corollary}{Corollary}
\theoremstyle{definition}
\newtheorem{definition}{Definition} 
\newtheorem{assumption}{Assumption}
\newtheorem{example}{Example}
\newtheorem{remark}{Remark}
\newtheorem*{remark*}{Remark}
\definecolor{mygray}{gray}{0.4}
\newcommand{\set} [1]			{{\mathscr{{#1}}}}
\newcommand{\alg}[1]			{{\mathcal{{#1}}}}
\newcommand{\rndv}[1]      {{\mathsf{{#1}}}}
\newcommand{\oper}[1]      {{\mathtt{{#1}}}}
\newcommand{\msr}[1]       {{\it    {{#1}}}}
\newcommand{\cnst}[1]      {{\mathit{{#1}}}}

\newcommand{\cntt}[1]      {{\widetilde{{\mathit{{#1}}}}}}

\newcommand{\integers}[1]	{{\mathbb{Z}}_{^{{#1}}}}

\newcommand{\reals}[1]		{{\mathbb{R}}_{^{{#1}}}}
\newcommand{\bigo} [1]     {{\cnst{O}\left({{#1}}\right)}}
\newcommand{\smallo}[1]    {{\cnst{o}\left({{#1}}\right)}}
\newcommand{\bigtheta}[1]  {{\cnst{\Theta}\left({{#1}}\right)}}

\newcommand{\inte}[1]      {{\mathtt{int}{{#1}}}}
\newcommand{\clos}[1]      {{\mathtt{cl}{{#1}}}}
\newcommand{\conv}[1]      {{\mathtt{ch}{{#1}}}}

\newcommand{\dif}[1]       {{\mathrm{d}{#1}}}  
\newcommand{\der}[2]        {\tfrac{\dif{#1}}{\dif{#2}}}  
\newcommand{\pder}[2]       {\tfrac{\partial{#1}}{\partial{#2}}}  
\newcommand{\supp}[1]       {\mathtt{supp}({{#1}})}       

\newcommand{\DEF}[0]			{{\!\!~\triangleq\!~}}  
\newcommand{\mtimes}[0]			{{\circledast}}

\newcommand{\AC}[0]            {{\prec}}

\newcommand{\abs}[1]           {{\left\lvert{{#1}}\right\lvert}}
\newcommand{\abp}[1]           {{\left\lvert{{#1}}\right\lvert^{+}}}
\newcommand{\lon}[1]           {{{\left\lVert{{#1}}\right\lVert}}} 

\newcommand{\IND}[1]           {{\mathds{1}_{\{#1\}}}}    
\newcommand{\ind}[0]           {{\imath}}

\newcommand{\knd}[0]           {{\kappa}}
\newcommand{\tin}[0]           {{\cnst{t}}}
\newcommand{\blx}[0]           {{\cnst{n}}}
\newcommand{\tlx}[0]           {{\cnst{T}}}

\newcommand{\domtr}[1]         {{\set{Q}}_{{#1}}}

\newcommand{\pint}[0]          {{\cnst{\zeta}}}
\newcommand{\PXS}[2]         {{\bf P}_{{#1}}\!\left[{#2}\right]}
\newcommand{\EXS}[2]         {{\bf E}_{{#1}}\!\left[{#2}\right]}

\newcommand{\EX}[1]          {\EXS{\!}{{#1}}}                      

\newcommand{\fX}[0]          {{\cnst{f}}}

\newcommand{\gX}[0]          {{\cnst{g}}}   
\newcommand{\GX}[0]          {{\cnst{G}}}   
\newcommand{\hX}[0]          {{\cnst{h}}}
\newcommand{\fXS}[0]         {{\set{F}}}

\newcommand{\cm}[1]         {{{\cnst{g}}_{{#1}}}}

\newcommand{\RD}[3]				
{{\cnst{D}}_{{#1}}            \!\left(\left.            \! {#2}\right\Vert {#3}                  \right)}
\newcommand{\CRD}[4]			
{{\cnst{D}}_{{#1}}            \!\left(\left.\!\left.    \! {#2}\right\Vert {#3} \right\vert{{#4}}\right)}

\newcommand{\GMI}[3]	{{\cnst{I}}_{{#1}}^{{\scriptscriptstyle g}}\!\left(\! {#2};\!{#3}\!\right)} 
\newcommand{\RMI}[3]	{{\cnst{I}}_{{#1}}        \!\left(        \! {#2};  \!{#3}                \!\right)} 
\newcommand{\GMIL}[4]	{{\cnst{I}}_{{#1}}^{{\scriptscriptstyle g}{#4}}	\!\left(\! {#2}; \!{#3}                \!\right)} 
\newcommand{\RMIL}[4]	{{\cnst{I}}_{{#1}}^{{#4}}	  					\!\left(\! {#2}; \!{#3}                \!\right)} 

\newcommand{\GCL}[3]	{{\cnst{C}}_{{#1},{#2}}^{{\scriptscriptstyle g}{#3}}}
\newcommand{\RCL}[3]	{{\cnst{C}}_{{#1},{#2}}^{{#3}}}
\newcommand{\RRL}[3]	{{\cnst{S}}_{{#1},{#2}}^{{#3}}}

\newcommand{\RC}[2]				{{\cnst{C}}_{{#1},{#2}}}

\newcommand{\CRC}[3]			{{\cnst{C}}_{{#1},{#2},{#3}}}
\newcommand{\CRCI}[4]			{{\cntt{C}}_{{#1},{#2},{#3}}^{#4}}

\newcommand{\RCI}[3]			{{\cntt{C}}_{{#1},{#2}}^{#3}}

\newcommand{\Aopi}[4]	{{\oper{T}}_{{#1},{#2}}^{#3}\left({#4}\right)} 
\newcommand{\Aop}[3]	{{\Aopi{#1}{#2}{}{#3}}}

\newcommand{\composition}[0]	{{\cnst{\varUpsilon}}}
\newcommand{\costf}[0]			{{\cnst{\rho}}}
\newcommand{\costc}[0]			{{\cnst{\varrho}}}
\newcommand{\lgm}[0]			{{\cnst{\lambda}}}
    
\newcommand{\rfm}[0]			{{{\msr{\nu}}}} 
 
\newcommand{\uc}[0]				{{\mathds{1}}}
\newcommand{\fcc}[1]			{{\cnst{\Gamma}_{{#1}}}}
\newcommand{\fccc}[1]			{{\cnst{\Gamma}_{{#1}}^{{\scriptscriptstyle ex}}}}

\newcommand{\cset}[0]			{{\set{A}}}
\newcommand{\cinpS}[0]			{{\set{B}}}

\newcommand{\spa}[2]			{{\cntt{E}_{sp\!}^{{#1}}}\left({#2}\right)}

\newcommand{\spe}[1]			{{\cnst{E}_{sp\!}}       \left({#1}\right)}

\newcommand{\rate}[0]			{{\cnst{R}}}

\newcommand{\cln}[1]          {{{\xi}_{{#1}}}}
\newcommand{\cla}[2]          {{{\xi}_{{#1}}^{{#2}}}}

\newcommand{\rnf}[0]          {{\cnst{\phi}}}

\newcommand{\rno}[0]          {{\cnst{\alpha}}}
\newcommand{\rnt}[0]          {{\cnst{\eta}}}
\newcommand{\rns}[0]          {\rno^{\!\ast}}
                  
\newcommand{\Pe}[0]            {{\it P_{{{\bf e}}}}}      
                  
\newcommand{\Pem}[1]           {{\it P_{{{\bf e}}}^{{#1}}}}         
 
\newcommand{\enc}[0]           {{\varPsi}} 
\newcommand{\dec}[0]           {{\varTheta}}    

\newcommand{\brl}[0]           {{\alg{B}}}
\newcommand{\rborel}[1]        {{\brl}({#1})}

\newcommand{\oev}[0]           {{\set{E}}}

\newcommand{\Pcha}[1]          {{\varLambda}^{{{#1}}}}
\newcommand{\GausDen}[1]		{{{{\cnst{\varphi}}}_{{#1}}}}

\newcommand{\fmea}[1]          {{{\alg{M}}^{^{+}}\!({#1})}}

\newcommand{\pmea}[1]          {{{\alg{P}}({#1})}}

\newcommand{\pdis}[1]          {{{\set{P}}({#1})}}

\newcommand{\dinp}[0]          {{\cnst{x}}}

\newcommand{\inpS}[0]          {{\set{X}}}
\newcommand{\inpA}[0]          {{\alg{X}}}

\newcommand{\dout}[0]          {{\cnst{y}}}
\newcommand{\out}[0]           {{\rndv{Y}}}
\newcommand{\outS}[0]          {{\set{Y}}}
\newcommand{\outA}[0]          {{\alg{Y}}}

\newcommand{\dsta}[0]          {{\cnst{z}}}

\newcommand{\staS}[0]          {{\set{Z}}}

\newcommand{\dmes}[0]          {{\cnst{m}}}

\newcommand{\mesS}[0]          {{\set{M}}}

\newcommand{\estS}[0]          {{\widehat{{\set{M}}}}}

\newcommand{\mean}[0]        {{{\msr{\mu}}}}    

\newcommand{\mma}[2]         {{{\mean}_{{#1}}^{{#2}}}}    
\newcommand{\qgn}[1]         {{{\mQ}_{{#1}}^{{\scriptscriptstyle g}}}}
\newcommand{\qga}[2]         {{{\mQ}_{{#1}}^{{\scriptscriptstyle g}{#2}}}}

\newcommand{\mA}[0]				{{\msr{a}}}

\newcommand{\mB}[0]				{{\msr{b}}}

\newcommand{\mP}[0]				{{\msr{p}}}    
\newcommand{\pmn}[1]			{{{\mP}_{{#1}}}}

\newcommand{\mQ}[0]				{{\msr{q}}}    
\newcommand{\qmn}[1]			{{{\mQ}_{{#1}}}}
\newcommand{\qma}[2]			{{{\mQ}_{{#1}}^{{#2}}}}
\newcommand{\Qm}[0]				{{{\cnst{Q}}}}

\newcommand{\mS}[0]				{{\msr{s}}}

\newcommand{\mU}[0]				{{\msr{u}}}    
\newcommand{\umn}[1]			{{{\mU}_{{#1}}}}

\newcommand{\Um}[0]				{{{\cnst{U}}}}

\newcommand{\Uma}[2]			{{{\cnst{U}}_{{#1}}^{{#2}}}}

\newcommand{\mV}[0]				{{\msr{v}}}    

\newcommand{\vma}[2]			{{{\mV}_{{#1}}^{{#2}}}}
\newcommand{\Vm}[0]				{{{\cnst{V}}}}
\newcommand{\Vmn}[1]			{{{\cnst{V}}_{{#1}}}}

\newcommand{\mW}[0]				{{\msr{w}}}    
\newcommand{\wmn}[1]			{{{\mW}_{{#1}}}}
\newcommand{\wma}[2]			{{{\mW}_{{#1}}^{{#2}}}}
\newcommand{\Wm}[0]				{{{\cnst{W}}}}
\newcommand{\Wmn}[1]			{{{\cnst{W}}_{{#1}}}}
\newcommand{\Wma}[2]			{{{\cnst{W}}_{{#1}}^{{#2}}}}

\newcommand{\ebertsthreshold}[0]					{{\cnst{N}}}

\newcommand{\harremoes}[0]							{Harremo\"{e}s~}

\newcommand{\renyi}[0]								{R\'{e}nyi~}

\makeatletter
\DeclareRobustCommand{\bigplus}{%
	\mathop{\vphantom{\sum}\mathpalette\@bigplus\relax}\slimits@
}
\newcommand{\@bigplus}[2]{\vcenter{\hbox{\make@bigplus{#1}}}}
\newcommand{\make@bigplus}[1]{%
	\sbox\z@{$\m@th#1\sum$}%
	\setlength{\unitlength}{\wd\z@}%
	\begin{picture}(1.4,1.4)
	\linethickness{.17ex}
	\Line(.7,.14)(.7,1.26)
	\Line(.14,.7)(1.26,.7)
	\end{picture}%
}
\DeclareRobustCommand{\bigtimes}{%
	\mathop{\vphantom{\sum}\mathpalette\@bigtimes\relax}\slimits@
}
\newcommand{\@bigtimes}[2]{\vcenter{\hbox{\make@bigtimes{#1}}}}
\newcommand{\make@bigtimes}[1]{%
	\sbox\z@{$\m@th#1\sum$}%
	\setlength{\unitlength}{\wd\z@}%
	\begin{picture}(1,1)
	\linethickness{.17ex}
	\Line(.1,.1)(.9,.9)
	\Line(.1,.9)(.9,.1)
	\end{picture}%
}
\makeatother

\def\thesubsection{\thesection.\arabic{subsection}}
 
\def\thesubsectiondis{\thesectiondis\arabic{subsection}.} 
\begin{document}
\pagestyle{plain}
\pagenumbering{arabic}
\hypersetup{hidelinks}
\maketitle 
\thispagestyle{empty}
\begin{abstract}
Sphere packing bounds (SPBs) ---with prefactors that are polynomial in the block length---
are derived for codes on two families of memoryless channels using Augustin's method:
(possibly non-stationary) memoryless channels with (possibly multiple) additive cost constraints  
and stationary memoryless channels with convex constraints on the composition 
(i.e. empirical distribution, type) of the input codewords.
A variant of Gallager's bound is derived in order to show that these sphere
packing bounds are tight in terms of the exponential decay rate of the error probability with 
the block length under mild hypotheses.  
\end{abstract}
\tableofcontents
\clearpage
\section{Introduction}\label{sec:introduction}
Most proofs of the sphere packing bound (SPB) have been 
either for the stationary channels with finite input sets 
\cite{elias55A,elias55B,dobrushin62B,altugW19,altugW11,shannonGB67A,gallager,fano,haroutunian68,omura75,csiszarkorner,altugW14A,nakiboglu19-ISIT}
or for the stationarity channels with a specific noise structure, 
e.g.  Poisson, Gaussian,
\cite{wyner88-b,burnashevK99,shannon59,ebert65,ebert66,richters67,lancho19,lanchoODKV19A,lanchoODKV19B}.
Proofs of the SPB based on Augustin's method are exceptions 
to this observation:
\cite{augustin69,augustin78,nakiboglu19B} do not assume
either the finiteness of the input set or a specific noise 
structure, nor do they assume the stationarity of the channel.
However, \cite{augustin69}, \cite[\S 31]{augustin78}, 
\cite{nakiboglu19B} establish the SPB 
for the product channels, rather than the memoryless channels;
hence proofs of the SPB for the composition constrained 
codes\footnote{According to \cite[p. 183]{csiszarkorner}, the SPB 
	for the constant composition codes appears in \cite{fano} with 
	an incomplete proof. 
	The first complete proof of the SPB for the constant composition 
	codes is provided in \cite{haroutunian68}.} 
on the stationary channels
\cite{fano,haroutunian68,omura75,csiszarkorner,altugW14A,nakiboglu19-ISIT,wyner88-b,burnashevK99,shannon59,ebert65,ebert66,richters67}
---which include the important special case of the cost constrained ones
\cite{shannon59,ebert65,ebert66,richters67,wyner88-b,burnashevK99}---
are not subsumed by \cite{augustin69}, \cite[\S 31]{augustin78}, or 
\cite{nakiboglu19B}.
In \cite[\S 36]{augustin78}, Augustin proved the SPB for the cost constrained 
(possibly non-stationary) memoryless channels assuming a bounded cost function.
The framework of \cite[Thm. 36.6]{augustin78} subsumes all previously considered
models \cite{elias55A,elias55B,dobrushin62B,altugW19,altugW11,shannonGB67A,gallager,fano,haroutunian68,omura75,csiszarkorner,
	altugW14A,nakiboglu19-ISIT,wyner88-b,burnashevK99,shannon59,ebert65,ebert66,richters67,lancho19,lanchoODKV19A,lanchoODKV19B},
except the Gaussian ones 
\cite{shannon59,ebert65,ebert66,richters67}.

Theorem \ref{thm:exponent-cost-ologn}, presented in \S\ref{sec:outerbound}, establishes 
the SPB for a framework that subsumes all of the models 
considered in \cite{elias55A,elias55B,dobrushin62B,altugW19,altugW11,shannonGB67A,gallager,fano,haroutunian68,omura75,
	csiszarkorner,altugW14A,nakiboglu19-ISIT,wyner88-b,burnashevK99,shannon59,ebert65,ebert66,richters67,
	lancho19,lanchoODKV19A,lanchoODKV19B,augustin69,augustin78,nakiboglu19B}
by employing \cite{nakiboglu19C}, which analyzes Augustin's information measures. 
Our use of \cite{nakiboglu19C} and Augustin's information measures
is similar to the use of \cite{nakiboglu19A} and 
\renyi\!\!'s information measures in \cite{nakiboglu19B}.
For the product channels,
\cite[Thm. \ref*{B-thm:productexponent}]{nakiboglu19B} 
improved the previous results by Augustin in \cite{augustin69}, \cite[\S 31]{augustin78}
by establishing the SPB  with 
a prefactor that is polynomial in the block length \(\blx\)
for the hypothesis that
the order \textonehalf~\renyi capacity of the component channels are 
\(\bigo{\ln\blx}\). 
For the cost constrained memoryless channels, Theorem \ref{thm:exponent-cost-ologn}
enhances the prefactor of \cite[Thm. 36.6]{augustin78} in an analogous way,
from \(e^{-\bigo{\sqrt{\blx}}}\) to \(e^{-\bigo{\ln\blx}}\).
The prefactor of Theorem \ref{thm:exponent-cost-ologn}, however, is inferior  
to the prefactors reported 
in \cite{elias55B,dobrushin62B,altugW19,altugW11} for various symmetric channels,
in \cite{altugW14A} for the constant composition codes on 
discrete stationary product channels, 
in \cite{shannon59} for the stationary Gaussian channel,
and in \cite{lancho19,lanchoODKV19A,lanchoODKV19B} for  certain non-coherent fading channels.
Determination of the optimal prefactor, in the spirit of 
\cite{elias55B,dobrushin62B,altugW19}, remains an open problem for the general 
case.\footnote{Elsewhere in \cite{nakiboglu19F}, we have derived refined SPBs 
	(which are optimal in terms of the prefactor for non-singular cases) 
	for all of the cases considered 
	in \cite{elias55B,dobrushin62B,altugW19,altugW11,altugW14A,shannon59,
		lancho19,lanchoODKV19A,lanchoODKV19B} 
	using Augustin information measures via \cite{nakiboglu19C}.}
Similar to \cite[Thm. 36.6]{augustin78}, Theorem \ref{thm:exponent-cost-ologn}
holds for non-stationary channels, as well. 
Unlike \cite[Thm. 36.6]{augustin78}, Theorem \ref{thm:exponent-cost-ologn}
does not assume the cost functions to be bounded.

The stationarity is assumed in most of the previous derivations of the SPB,
\cite{elias55A,elias55B,dobrushin62B,altugW19,altugW11,shannonGB67A,gallager,fano,
	haroutunian68,omura75,csiszarkorner,altugW14A,nakiboglu19-ISIT,shannon59,ebert65,ebert66,
	richters67,wyner88-b,burnashevK99,lancho19,lanchoODKV19A,lanchoODKV19B}.
Given a stationary product channel, one can obtain a stationary memoryless
channel by imposing composition ---i.e. type, empirical distribution---
constraints on the input codewords. 
The cost constraints can be interpreted as a particular convex case of this more 
general composition constraints.
This interpretation, considered together with the composition based expurgations, 
is one of the main motivating factors behind the study of constant composition 
codes.
The composition based expurgations, however, are useful only when the input set 
of the channel is finite.
Nevertheless, if the constraint set for the composition of the codewords is convex, 
then one can derive a SPB with a polynomial prefactor using Augustin's information 
measures, see Theorem \ref{thm:exponent-convex} in \S\ref{sec:outerbound}.  
The derivation of Theorem \ref{thm:exponent-convex} relies on the Augustin center 
of the constraint set rather than the Augustin mean of the most populous composition 
of the code. Note that the most populous composition of the code might not even have 
more than one codeword when the input set is infinite.
The framework of Theorem \ref{thm:exponent-convex} is
general enough to subsume the frameworks of all
previous proofs of the SPB for the memoryless channels that we are aware of,
except the frameworks of the proofs based on Augustin's method
\cite{augustin69,augustin78,nakiboglu19B}.
Theorems \ref{thm:exponent-convex} and \ref{thm:exponent-cost-ologn} 
are asymptotic SPBs; but they are proved using non-asymptotic SPBs presented 
in Lemmas \ref{lem:spb-convex} and \ref{lem:spb-cost-ologn}.

The SPB implies that exponential decay rate of the optimal error probability with
the block length ---i.e. the reliability function, the error exponent--- 
is bounded from above by the sphere packing exponent (SPE). 
For the memoryless channels in consideration, Augustin's variant of Gallager's 
bound implies that the SPE bounds
the reliability function from below, as well, provided that the list decoding 
is allowed.
Augustin's variant of Gallager's bound is presented in \S\ref{sec:innerbound}. 
One can use standard results such as \cite{gallager65,poltyrev82} 
with minor modifications in order to establish the SPE
as a lower bound to the reliability function for the list decoding, as well.
Thus Augustin's variant is of interest to us not 
because of what it implies about the reliability function but because of
how it implies it.
What is unique about  Augustin's variant is that it establishes an 
achievability result in terms of the Augustin information rather 
than the \renyi information used in the standard form of the Gallager's bound 
\cite{gallager65}.
Augustin's variant relies on the fixed point property of the Augustin 
mean described in \eqref{eq:augustinfixedpoint} to do that.
It is worth mentioning that \cite{poltyrev82} implicitly employs 
the same fixed point property but in a different way.

Before starting our discussion in earnest, let us point out a subtlety about 
the derivations of the SPB that is usually overlooked.
\cite{blahut74} claimed to prove the SPB for arbitrary stationary product channels, 
without using any constant composition 
arguments.\footnote{\cite[p. 413]{blahut74} reads  ``An important 
	feature of the lower bound, which will be derived, is that no assumption 
	of constant-composition codewords is made, not even as an intermediate step.''} 
The derivation of \cite[Thm. 19]{blahut74}, however, establishes an upper bound on 
the reliability function that is strictly greater than the SPE in 
many channels.
This has been demonstrated numerically in \cite[p. 1594 and Appendix A]{altugW14A}.
An analytic confirmation this observation is presented in Appendix \ref{sec:blahut}.
The problematic step in \cite{blahut74} is the application of Lagrange multiplier 
techniques, see \cite[footnote 8]{altugW14A}. 
The proof of \cite[Thm. 19]{blahut74} invokes \cite[Thm. 16]{blahut74}
that is valid for the Lagrange multiplier \(\mS\) associated with 
the input distribution \(\mP\) satisfying \(\spe{\rate,\Wm}\!=\!\spe{\rate,\Wm,\mP}\).
For an arbitrary input distribution \(\mP\), however,  the associated Lagrange multiplier may or 
may not be equal to the one for the optimal input distribution \(\mP\).
This is the reason why the upper bound to the reliability function 
established in \cite[Thm. 19]{blahut74}
is not equal to the SPE in general,
contrary to the claim repeated in \cite[Lemma 1]{blahut76}
and \cite[Thm. 10.1.4]{blahut}.
In a nutshell, the proof of \cite[Thm. 19]{blahut74} tacitly asserts
a minimax equality that does not hold in general.
For stationary memoryless channels with finite input alphabets,
one can avoid this issue using the constant composition arguments.
However, in that case, the proof presented in  \cite{blahut74} becomes a mere reproduction 
of the one in \cite{haroutunian68}.
More recently, \cite{vazquezMF15} proposed a derivation of the SPB for stationary 
channels with a single cost constraint using the approach presented in \cite{blahut74}.
Similar to \cite{blahut74}, however,
the proof in \cite{vazquezMF15} asserts a minimax equality that does not hold in general.
In particular, it is claimed that \(Q^{n}\) does not depend on \({\bf x}_{m}\) in 
\cite[(26)]{vazquezMF15}.
To assert that, one has to include an additional  supremum over \({\bf x}_{m}\) as the 
innermost optimization in both \cite[(25) and (26)]{vazquezMF15}. 
With the additional supremum, the explanation provided on 
\cite[p. 931]{vazquezMF15} is no 
longer valid.
Considering Appendix \ref{sec:blahut}, we do not believe that
the proof in \cite{vazquezMF15} can be salvaged without introducing 
major new ideas, 
such as composition based expurgations similar to \cite{haroutunian68}
or codeword cost based expurgations similar to \cite{ebert65}.
In short,  neither \cite{blahut74} nor \cite{vazquezMF15} successfully proved 
the SPB for stationary memoryless channels even for the finite input set case.

In the rest of this section, we introduce our notation and channel model and define 
the channel codes with list decoding.
In \S \ref{sec:preliminary}, we first present a brief review of the \renyi divergence, 
Augustin information measures, and the SPE;
then, we derive Augustin's variant of Gallager's bound.
In \S \ref{sec:outerbound}, we first state our main asymptotic results
---i.e., SPBs given in Theorem \ref{thm:exponent-convex} and Theorem \ref{thm:exponent-cost-ologn}---
and then derive the non-asymptotic SPBs implying them.
In \S\ref{sec:examples}, we derive the SPE for particular Gaussian and Poisson channels and 
confirm the equivalence of the definition invoked in \S \ref{sec:preliminary} 
to the ones derived for these channels previously.
In \S\ref{sec:conclusion}, we discuss why Augustin's method works briefly and compare
our results with Augustin's in \cite{augustin78} and
discuss applications of Augustin's method and  use of Augustin's information measures
in related problems.

\subsection{Notational Conventions}\label{sec:introduction-notation}
For any two vectors \(\mean\) and \(\mQ\) in \(\reals{}^{\ell}\) their inner product,
denoted by \(\mean \cdot \mQ\), is \(\sum_{\ind=1}^{\ell} \mean^{\ind} \mQ^{\ind}\).
For any \(\ell\in\integers{+}\), \(\ell\) dimensional vector whose all entries are one
is denoted by \(\uc\), the dimension \(\ell\) will be clear from the context.
We denote the closure, interior, and convex hull of a set \(\set{S}\) by 
\(\clos{\set{S}}\), \(\inte{\set{S}}\), and \(\conv{\set{S}}\), respectively;
the relevant topology or vector space structure will be evident from the context.

For any set \(\outS\), we denote the set of all probability mass functions that 
are non-zero only on finitely many members of \(\outS\) by \(\pdis{\outS}\).
For any \(\mP\in\pdis{\outS}\), we call the set of all \(\dout\)'s in \(\outS\) for 
which  \(\mP(\dout)>0\) the support of \(\mP\) and denote it by \(\supp{\mP}\).
For any measurable space \((\outS,\outA)\), we denote the set of all probability measures on it by 
\(\pmea{\outA}\) and set of all finite measures by \(\fmea{\outA}\).
We denote the integral of a measurable function \(\fX\) with respect to the measure \(\mean\) 
by \(\int \fX \mean(\dif{\dout})\) or \(\int \fX(\dout) \mean(\dif{\dout})\).
If the integral is on the real line and if it is with respect to the Lebesgue measure, we 
denote it by \(\int\fX \dif{\dout}\) or \(\int \fX(\dout) \dif{\dout}\), as well.
If \(\mean\) is a probability measure, then we also call the integral of \(\fX\) with respect \(\mean\)
the expectation of \(\fX\) or the expected value of \(\fX\) and denote it by
\(\EXS{\mean}{\fX}\) or \(\EXS{\mean}{\fX(\out)}\).

Our notation will be overloaded for certain symbols; however, the relations represented 
by these symbols will be clear from the context.
We denote the Cartesian product  of sets \cite[p. 38]{dudley} by \(\times\). 
We use \(\abs{\cdot}\) to denote the absolute value of real numbers and the size of sets. 
The sign \(\leq\) stands for the usual less than or equal to relation for real numbers
and the corresponding point-wise inequity for functions and vectors. 
For two measures \(\mean\) and \(\mQ\) on the measurable space \((\outS,\outA)\), 
\(\mean \leq \mQ\) iff \(\mean(\oev)\leq \mQ(\oev)\) for all \(\oev\in\outA\).
We denote the product of topologies \cite[p. 38]{dudley}, 
\(\sigma\)-algebras \cite[p. 118]{dudley}, and measures \cite[Thm. 4.4.4]{dudley} by \(\otimes\).
We use the shorthand 
\(\inpS_{1}^{\blx}\) for the Cartesian product of sets \(\inpS_{1},\ldots,\inpS_{\blx}\)
and 
\(\outA_{1}^{\blx}\) for the product of the \(\sigma\)-algebras  \(\outA_{1},\ldots,\outA_{\blx}\).


\subsection{Channel Model}\label{sec:introduction-model}
A \emph{channel} \(\Wm\) is a function from \emph{the input set} \(\inpS\) to  the set of all probability 
measures on \emph{the output space} \((\outS,\outA)\):
\begin{align}
\label{eq:def:channel}
\Wm:\inpS \to \pmea{\outA}
\end{align}
\(\outS\) is called \emph{the output set}, and \(\outA\) is called \emph{the \(\sigma\)-algebra  of the 
	output events}. 
We denote the set of all channels from the input set \(\inpS\) to the output space \((\outS,\outA)\)
by \(\pmea{\outA|\inpS}\).
For any \(\mP\in\pdis{\inpS}\)  and \(\Wm\in\pmea{\outA|\inpS}\), \(\mP\mtimes\Wm\) is the probability
measure whose marginal on \(\inpS\) is \(\mP\) and whose conditional distribution given \(\dinp\) is
\(\Wm(\dinp)\). 
The structure described in \eqref{eq:def:channel} is not sufficient on its own to 
ensure the existence of a unique \(\mP\mtimes\Wm\) with the desired properties for
all  \(\mP\in\pmea{\inpA}\), in general.
The existence of a unique \(\mP\mtimes\Wm\) is guaranteed for all \(\mP\in\pmea{\inpA}\),
if  \(\Wm\) is a transition probability from \((\inpS,\inpA)\) to \((\outS,\outA)\), 
i.e. a member of \(\pmea{\outA|\inpA}\) rather than \(\pmea{\outA|\inpS}\).

A channel \(\Wm\) is called a \emph{discrete channel} if both \(\inpS\) and \(\outA\) are finite sets.
For any \(\blx\in\integers{+}\) and 
channels \(\Wmn{\tin}\!:\!\inpS_{\tin}\!\to\!\pmea{\outA_{\tin}}\) for \(\tin\in\{1,\ldots,\blx\}\),
the \emph{length \(\blx\) product channel}
\(\Wmn{[1,\blx]}\!:\!\inpS_{1}^{\blx}\!\to\!\pmea{\outA_{1}^{\blx}}\) is defined via the following relation:
\begin{align}
\notag
\Wmn{[1,\blx]}(\dinp_{1}^{\blx})
&=\bigotimes\nolimits_{\tin=1}^{\blx}\Wmn{\tin}(\dinp_{\tin})
&
&\forall \dinp_{1}^{\blx}\in\inpS_{1}^{\blx}.
\end{align}
A channel \(\Um\!:\!\staS\!\to\!\pmea{\outA_{1}^{\blx}}\) is called a \emph{length \(\blx\) memoryless channel} 
iff there exists a product channel \(\Wmn{[1,\blx]}\)
satisfying both 
\(\Um(\dsta)\!=\!\Wmn{[1,\blx]}(\dsta)\) for all \(\dsta\!\in\!\staS\)
and 
\(\staS\!\subset\!\inpS_{1}^{\blx}\).
A product channel is \emph{stationary} iff \(\Wmn{\tin}\!=\!\Wm\) 
for all \(\tin\!\in\!\{1,\ldots,\blx\}\) for some \(\Wm\!:\!\inpS\!\to\!\pmea{\outA}\).
For such a channel, we denote the composition (i.e. the empirical distribution, type) 
of each \(\dinp_{1}^{\blx}\in\inpS_{1}^{\blx}\) by \(\composition(\dinp)\),
where  \(\composition(\dinp)\!\in\!\pdis{\inpS}\).

For any \(\ell\in\integers{+}\), an \(\ell\) dimensional \emph{cost function} \(\costf\) 
is a function from the input set to \(\reals{}^{\ell}\) 
that is bounded from below, i.e. that is of the form
\(\costf:\inpS \to \reals{\geq \dsta}^{\ell}\)
for some \(\dsta\in \reals{}\). 
We assume without loss of generality that\footnote{Augustin \cite[\S33]{augustin78} has an additional 
	hypothesis, \(\bigvee_{\dinp\in \inpS} \costf(\dinp)\leq\uc\),
	which excludes certain important cases
	such as the Gaussian channels.} 
\begin{align}
\notag
\inf\nolimits_{\dinp\in \inpS} \costf^{\ind}(\dinp) 
&\geq 0
&
&\forall \ind \in \{1,\ldots,\ell\}.
\end{align}
We denote the set of all cost constraints that can be satisfied by some member of \(\inpS\) by \(\fccc{\costf}\) 
and the set of all cost constraints that can be satisfied by some member of \(\pdis{\inpS}\) by \(\fcc{\costf}\):
\begin{align}
\notag
\fccc{\costf}
&\DEF\{\costc\in\reals{\geq0}^{\ell}:\exists \dinp\in\inpS \mbox{~s.t.~}\costf(\dinp) \leq \costc\},
\\
\notag
\fcc{\costf}
&\DEF\{\costc\in\reals{\geq0}^{\ell}:\exists\mP\in\pdis{\inpS} \mbox{~s.t.~}
\EXS{\mP}{\costf}\leq \costc\}.
\end{align}
Then both \(\fccc{\costf}\) and \(\fcc{\costf}\) have non-empty interiors
and \(\fcc{\costf}\) is the convex hull of \(\fccc{\costf}\), i.e. \(\fcc{\costf}=\conv{\fccc{\costf}}\).

A cost function on a product channel is said to be additive iff 
it can be written as the sum of cost functions defined on the component channels.
Given \(\Wmn{\tin}\!:\!\inpS_{\tin}\!\to\!\pmea{\outA_{\tin}}\) and 
\(\costf_{\tin}\!:\!\inpS_{\tin}\!\to\!\reals{\geq0}^{\ell}\) for \(\tin\!\in\!\{1,\ldots,\blx\}\),
we denote the resulting additive cost function on \(\inpS_{1}^{\blx}\) for the channel
\(\Wmn{[1,\blx]}\)  by \(\costf_{[1,\blx]}\), i.e. 
\begin{align}
\notag
\costf_{[1,\blx]}(\dinp_{1}^{\blx})
&=\sum\nolimits_{\tin=1}^{\blx} \costf_{\tin}(\dinp_{\tin})  
&
&\forall \dinp_{1}^{\blx}\in\inpS_{1}^{\blx}. 
\end{align}

\subsection{Codes With List Decoding}\label{sec:introduction-codes}
The pair \((\enc,\dec)\) is an \((M,L)\) \emph{channel code} on \(\Wm:\inpS\to\pmea{\outA}\) iff
\begin{itemize}
\item The \emph{encoding function} \(\enc\) is a function from the message set \(\mesS\DEF\{1,2,\ldots,M\}\) to the input
 set \(\inpS\).
\item The \emph{decoding function} \(\dec\) is a measurable function from the output space \((\outS,\outA)\) to the 
set \(\estS\DEF\{\set{L}:\set{L}\subset\mesS \mbox{~and~}\abs{\set{L}}= L\}\).
\end{itemize}
Given an \((M,L)\) channel code \((\enc,\dec)\) on \(\Wm:\inpS\to\pmea{\outA}\),
\emph{the conditional error probability} \(\Pem{\dmes}\) for \(\dmes \in \mesS\) 
and \emph{the average error probability} \(\Pem{}\)
are defined as 
\begin{align}
\notag
\Pem{\dmes}
&\DEF \EXS{\Wm(\enc(\dmes))}{\IND{\dmes\notin\dec(\dout)}},
\\
\notag
\Pem{} 
&\DEF\tfrac{1}{M} \sum\nolimits_{\dmes\in \mesS} \Pem{\dmes}.
\end{align}
An encoding function \(\enc\), hence the corresponding code, is said to satisfy 
the cost constraint \(\costc\) iff \(\bigvee_{\dmes\in\mesS} \costf(\enc(\dmes))\leq \costc\).
An encoding function \(\enc\), hence the corresponding code, on a stationary 
product channel is said to satisfy 
an empirical distribution constraint \(\cset\subset\pdis{\inpS}\)
iff the composition of all of the codewords are in \(\cset\), 
i.e. iff \(\composition(\enc(\dmes))\in\cset\) for all \(\dmes\in\mesS\).

\section{Preliminaries}\label{sec:preliminary}
The \renyi divergence, tilting, and Augustin's information measures are central to the
analysis we present in the following sections. 
We introduce these concepts in \S\ref{sec:preliminary-divergence+tilting} and 
\S\ref{sec:preliminary-informationmeasures}, a more detailed discussion can 
be found in \cite{nakiboglu19C,ervenH14}. 
In \S\ref{sec:preliminary-spherepackingexponent} we define the SPE 
and derive widely known properties of it
for our general channel model. 
In \S\ref{sec:innerbound} we derive Augustin's variant of Gallager's bound. 

\subsection{The \renyi Divergence and Tilting}\label{sec:preliminary-divergence+tilting}
\begin{definition}\label{def:divergence}
	For any \(\rno\in\reals{+}\) and \(\mW,\mQ\in\fmea{\outA}\),
	\emph{the order \(\rno\) \renyi divergence between \(\mW\) and \(\mQ\)} is
	\begin{align}
\notag
	\RD{\rno}{\mW}{\mQ}
	&\DEF \begin{cases}
	\tfrac{1}{\rno-1}\ln \int (\der{\mW}{\rfm})^{\rno} (\der{\mQ}{\rfm})^{1-\rno} \rfm(\dif{\dout})
	&\rno\neq 1\\
	\int  \der{\mW}{\rfm}\left[ \ln\der{\mW}{\rfm} -\ln \der{\mQ}{\rfm}\right] \rfm(\dif{\dout})
	&\rno=1
	\end{cases}
	\end{align}
	where \(\rfm\) is any measure satisfying \(\mW\AC\rfm\) and \(\mQ\AC\rfm\).
\end{definition}
For properties of the \renyi divergence, throughout the manuscript, 
we will refer to the comprehensive 
study provided by van Erven and \harremoes \cite{ervenH14}. 
Note that the order one \renyi divergence is the Kullback-Leibler divergence.
For other orders, the \renyi divergence can be characterized in terms of the 
Kullback-Leibler divergence, as well, see \cite[Thm. 30]{ervenH14}. 
That characterization is related to another key concept for our analysis:
the tilted probability measure. 
\begin{definition}\label{def:tiltedprobabilitymeasure}
	For any \(\rno\in\reals{+}\) and \(\mW,\mQ\in\pmea{\outA}\) satisfying 
	\(\RD{\rno}{\mW}{\mQ}<\infty\), 
	\emph{the order \(\rno\) tilted probability measure} \(\wma{\rno}{\mQ}\) is 
	\begin{align}
	\label{eq:def:tiltedprobabilitymeasure}
	\der{\wma{\rno}{\mQ}}{\rfm}
	&\DEF e^{(1-\rno)\RD{\rno}{\mW}{\mQ}}(\der{\mW}{\rfm})^{\rno} (\der{\mQ}{\rfm})^{1-\rno}.
	\end{align}
\end{definition}
The conditional \renyi divergence and the tilted channel are straight forward generalizations
of the \renyi divergence and the tilted probability measure that will allow us to express 
certain relations succinctly throughout our analysis. 
\begin{definition}\label{def:conditionaldivergence}
For any \(\rno\in\reals{+}\), \(\Wm:\inpS\to\pmea{\outA}\), \(\Qm:\inpS\to\pmea{\outA}\),
and \(\mP\in\pdis{\inpS}\) \emph{the order \(\rno\) conditional \renyi divergence for 
	the input distribution \(\mP\)} is
	\begin{align}
\notag
	\CRD{\rno}{\Wm\!}{\Qm}{\mP}
	&\DEF \sum\nolimits_{\dinp\in \inpS}  \mP(\dinp) \RD{\rno}{\Wm(\dinp)}{\Qm(\dinp)}.
	\end{align}
If \(\exists\mQ\in\pmea{\outA}\) such that  \(\Qm(\dinp)=\mQ\) for all \(\dinp\in\inpS\), 
then we denote \(\CRD{\rno}{\Wm\!}{\Qm}{\mP}\) by \(\CRD{\rno}{\Wm\!}{\mQ}{\mP}\).
\end{definition}
\begin{definition}\label{def:tiltedchannel}
	For any \(\rno\in\reals{+}\), \(\Wm:\inpS\to\pmea{\outA}\) and \(\Qm:\inpS\to\pmea{\outA}\),
	\emph{the order \(\rno\) tilted channel \(\Wma{\rno}{\Qm}\)} is a function
	from  \(\{\dinp:\RD{\rno}{\Wm(\dinp)}{\Qm(\dinp)}<\infty\}\) to \(\pmea{\outA}\)
	given by
	\begin{align}
	\label{eq:def:tiltedchannel}
	\der{\Wma{\rno}{\Qm}(\dinp)}{\rfm}
	&\DEF e^{(1-\rno)\RD{\rno}{\Wm(\dinp)}{\Qm(\dinp)}}(\der{\Wm(\dinp)}{\rfm})^{\rno} (\der{\Qm(\dinp)}{\rfm})^{1-\rno}.
	\end{align}
	If \(\exists\mQ\in\pmea{\outA}\) such that  \(\Qm(\dinp)=\mQ\) for all \(\dinp\in\inpS\), then 
	we denote \(\Wma{\rno}{\Qm}\) by \(\Wma{\rno}{\mQ}\).
\end{definition}
The following operator \(\Aop{\rno}{\mP}{\cdot}\) was considered implicitly by 
Fano \cite[Ch 9]{fano}, Haroutunian \cite{haroutunian68}, and Polytrev \cite{poltyrev82}
and explicitly by Augustin \cite[\S34]{augustin78}, 
but only for orders less than one in all four manuscripts. 
\begin{definition}\label{def:Aoperator}
	For any \(\rno\in \reals{+}\), \(\Wm:\inpS\to \pmea{\outA}\), and  \(\mP\in \pdis{\inpS}\),
	the order \(\rno\) Augustin operator for the input distribution \(\mP\), i.e. \(\Aop{\rno}{\mP}{\cdot}:\domtr{\rno,\mP}\to\pmea{\outA}\), is given by
	\begin{align}
	\label{eq:def:Aoperator}
	\Aop{\rno}{\mP}{\mQ}
	&\DEF\sum\nolimits_{\dinp}\mP(\dinp) \Wma{\rno}{\mQ}(\dinp)
	&
	&\forall \mQ\in  \domtr{\rno,\mP}
	\end{align}
	where \(\domtr{\rno,\mP}\DEF\{\mQ\in\pmea{\outA}:\CRD{\rno}{\Wm\!}{\mQ}{\mP}<\infty\}\) and
	the tilted channel \(\Wma{\rno}{\mQ}\) is defined in \eqref{eq:def:tiltedchannel}.
\end{definition}

\subsection{Augustin's Information Measures}\label{sec:preliminary-informationmeasures}
\begin{definition}\label{def:information}
	For any \(\rno\in \reals{+}\), \(\Wm:\inpS\!\to\!\pmea{\outA}\), and \(\mP\in \pdis{\inpS}\)  
	\emph{the order \(\rno\) Augustin information for the input distribution \(\mP\)} is
	\begin{align}
	\label{eq:def:information}
	\RMI{\rno}{\mP}{\Wm}
	&\DEF \inf\nolimits_{\mQ\in \pmea{\outA}} \CRD{\rno}{\Wm\!}{\mQ}{\mP}.
	\end{align}
\end{definition}
The infimum in \eqref{eq:def:information} is achieved by a unique probability measure
denoted by \(\qmn{\rno,\mP}\) and called
\emph{the order \(\rno\) Augustin mean for the input distribution \(\mP\)}.
Furthermore, the order \(\rno\) Augustin mean satisfies the following identities: 
\begin{align}
\label{eq:augustinslaw}
\RD{1\vee \rno}{\qmn{\rno,\mP}}{\mQ}
\geq
\CRD{\rno}{\Wm\!}{\mQ}{\mP}-\RMI{\rno}{\mP}{\Wm}
&\geq \RD{1 \wedge \rno}{\qmn{\rno,\mP}}{\mQ}
&
&\forall \mQ\in\pmea{\outA}, \rno\in\reals{+}.
\\
\label{eq:augustinfixedpoint}
\Aop{\rno}{\mP}{\qmn{\rno,\mP}}
&=\qmn{\rno,\mP}
&
&\forall \rno\in\reals{+}.
\end{align}
These observations are established in 
\cite[{Lemma\ref*{C-lem:information}-(\ref*{C-information:one},\ref*{C-information:zto},\ref*{C-information:oti})}]{nakiboglu19C};
previously they were reported  by Augustin \cite[Lemma 34.2]{augustin78} for orders less than one.
Throughout the manuscript, we refer to \cite{nakiboglu19C} for propositions about Augustin's information measures.
A more detailed account of the previous work on Augustin's information measures can be found in \cite{nakiboglu19C},
as well.

\begin{definition}\label{def:capacity}
	For any \(\rno\!\in\!\reals{+}\), \(\Wm\!:\!\inpS\!\to\!\pmea{\outA}\), and \(\cset\!\subset\!\pdis{\inpS}\),
	\emph{the order \(\rno\) Augustin capacity of \(\Wm\) for the constraint set \(\cset\)} is 
	\begin{align}
\notag
	\CRC{\rno}{\!\Wm\!}{\cset}
	&\DEF \sup\nolimits_{\mP \in \cset}  \RMI{\rno}{\mP}{\Wm}.
	\end{align}
	When the constraint set \(\cset\) is the whole \(\pdis{\inpS}\), we denote the order \(\rno\) 
	Augustin capacity by \(\RC{\rno}{\Wm\!}\), i.e. 
	\(\RC{\rno}{\!\Wm\!}\DEF\CRC{\rno}{\!\Wm\!}{\pdis{\inpS}}\).
\end{definition}
Using the definitions of the Augustin information and capacity we get  the following expression for 
\(\CRC{\rno}{\!\Wm\!}{\cset}\)
\begin{align}
\notag
\CRC{\rno}{\!\Wm\!}{\cset}
&=\sup\nolimits_{\mP \in \cset}\inf\nolimits_{\mQ\in\pmea{\outA}} \CRD{\rno}{\Wm\!}{\mQ}{\mP}.
\end{align}
If \(\cset\) is convex then the order of the supremum and the infimum can be changed 
as a result of \cite[Thm. \ref*{C-thm:minimax}]{nakiboglu19C}:
\begin{align}
	\label{eq:thm:minimax}
\sup\nolimits_{\mP \in \cset}\inf\nolimits_{\mQ\in\pmea{\outA}} \CRD{\rno}{\Wm\!}{\mQ}{\mP}
&=\inf\nolimits_{\mQ\in\pmea{\outA}}\sup\nolimits_{\mP \in \cset} \CRD{\rno}{\Wm\!}{\mQ}{\mP}.
\end{align}
If in addition \(\CRC{\rno}{\!\Wm\!}{\cset}\) is finite, then 
\cite[Thm. \ref*{C-thm:minimax}]{nakiboglu19C} implies that
there exists a unique probability measure \(\qmn{\rno,\!\Wm\!,\cset}\),
\emph{called the order \(\rno\) Augustin center of \(\Wm\) for the constraint set \(\cset\)}, 
satisfying
\begin{align}
\notag
\CRC{\rno}{\!\Wm\!}{\cset}
&=\sup\nolimits_{\mP \in \cset} \CRD{\rno}{\Wm\!}{\qmn{\rno,\!\Wm\!,\cset}}{\mP}.
\end{align}
We denote the set of all probability mass functions satisfying a cost constraint \(\costc\) by \(\cset(\costc)\), i.e.
\begin{align}
\notag
\cset(\costc)
&\DEF \{\mP\in\pdis{\inpS}:\EXS{\mP}{\costf}\leq\costc\}.
\end{align}
For the constraint sets defined through cost constraints we use the symbol 
\(\CRC{\rno}{\!\Wm\!}{\costc}\) rather than \(\CRC{\rno}{\!\Wm\!}{\cset(\costc)}\)
with a slight abuse of notation. 
In order to be able apply convex conjugation techniques without any significant 
modifications, we extend the definition Augustin capacity to the infeasible cost
constraints, i.e. \(\costc\)'s outside \(\fcc{\costf}\),  as follows:
\begin{align}
\notag
\CRC{\rno}{\!\Wm\!}{\costc}
&\DEF 
\begin{cases}
\sup\nolimits_{\mP\in\cset(\costc)} \RMI{\rno}{\mP}{\Wm}
&\mbox{if~}\costc\in \fcc{\costf}
\\
-\infty
&\mbox{if~}\costc\in \reals{\geq0}^{\ell}\setminus\fcc{\costf}
\end{cases}
&
&\forall\rno\in\reals{+}.
\end{align}
In order to characterize \(\CRC{\rno}{\!\Wm\!}{\costc}\)
through convex conjugation techniques, we first define
Augustin-Legendre (A-L) information and capacity.
These concepts are first introduced in \cite[\S III-A]{nakiboglu17}
and \cite[\S\ref*{C-sec:cost-AL}]{nakiboglu19C}, 
as an extension of the analogous concepts in \cite[Ch. 8]{csiszarkorner}.
\begin{definition}\label{def:Linformation}
	For any \(\rno\in\reals{+}\), channel \(\Wm\) of the form \(\Wm:\inpS\to \pmea{\outA}\) with 
	a cost function \(\costf:\inpS\to \reals{\geq0}^{\ell}\), \(\mP\in \pdis{\inpS}\),  
	and \(\lgm \in \reals{\geq0}^{\ell}\), 
	\emph{the order \(\rno\) Augustin-Legendre information for the input distribution \(\mP\) 
		and the Lagrange multiplier \(\lgm\)} is
	\begin{align}
\notag
	\RMIL{\rno}{\mP}{\Wm}{\lgm}
	&\DEF \RMI{\rno}{\mP}{\Wm}-\lgm\cdot \EXS{\mP}{\costf}.
	\end{align}
\end{definition}
\begin{definition}\label{def:Lcapacity}
	For any \(\rno\in\reals{+}\), channel \(\Wm\) of the form \(\Wm:\inpS\to \pmea{\outA}\) with 
	a cost function \(\costf:\inpS\to \reals{\geq0}^{\ell}\), and \(\lgm \in \reals{\geq0}^{\ell}\)
	\emph{the order \(\rno\) Augustin-Legendre (A-L) capacity for the Lagrange multiplier \(\lgm\)} is
	\begin{align}
\notag
	\RCL{\rno}{\Wm\!}{\lgm}
	&\DEF \sup\nolimits_{\mP\in \pdis{\inpS}} \RMIL{\rno}{\mP}{\Wm}{\lgm}.
	\end{align}
\end{definition}
 Except for certain sign changes, 
 \(\RCL{\rno}{\Wm\!}{\lgm}\) is the convex conjugate of \(\CRC{\rno}{\Wm\!}{\costc}\)
 because of an analogous relation between 
 \(\RMIL{\rno}{\mP}{\Wm}{\lgm}\) and \(\RMI{\rno}{\mP}{\Wm}\),
 see\cite[{(\ref*{C-eq:information-constrained})-(\ref*{C-eq:Linformation-conjugate}), (\ref*{C-eq:Lcapacity-astheconjugate})}]{nakiboglu19C}.
\begin{align}
\notag
\RCL{\rno}{\Wm\!}{\lgm}
&=\sup\nolimits_{\costc\geq0} \CRC{\rno}{\Wm\!}{\costc}-\lgm\cdot\costc
&
&\forall \lgm\in\reals{\geq0}^{\ell}.
\intertext{Then \(\CRC{\rno}{\Wm\!}{\costc}\) can be expressed in terms of \(\RCL{\rno}{\Wm\!}{\lgm}\)
at least for the interior points of \(\fcc{\costf}\):}
\notag
\CRC{\rno}{\!\Wm\!}{\costc}
&=\inf\nolimits_{\lgm\geq0} \RCL{\rno}{\Wm}{\lgm}+\lgm\cdot\costc.
\end{align}
Furthermore, there exists a non-empty convex compact  set of 
\(\lgm_{\rno,\!\Wm\!,\costc}\)'s satisfying
\(\CRC{\rno}{\!\Wm\!}{\costc}=\RCL{\rno}{\Wm}{\lgm_{\rno,\!\Wm\!,\costc}}+\lgm_{\rno,\!\Wm\!,\costc}\cdot \costc\)
provided that \(\CRC{\rno}{\!\Wm\!}{\costc}\) is finite, by \cite[Lemma \ref*{C-lem:Lcapacity}]{nakiboglu19C}.

On the other hand, using the definitions of \(\RMI{\rno}{\mP}{\Wm}\), \(\RMIL{\rno}{\mP}{\Wm}{\lgm}\), 
and \(\RCL{\rno}{\Wm}{\lgm}\)
we get the following expression for \(\RCL{\rno}{\Wm}{\lgm}\).
\begin{align}
\notag
\RCL{\rno}{\Wm}{\lgm}
&=\sup\nolimits_{\mP \in \pdis{\inpS}}\inf\nolimits_{\mQ\in\pmea{\outA}} \CRD{\rno}{\Wm}{\mQ}{\mP}-\lgm\cdot \EXS{\mP}{\costf}.
\end{align}
\(\RCL{\rno}{\Wm}{\lgm}\) satisfies a minimax relation similar to the one given in \eqref{eq:thm:minimax},
see \cite[Thm. \ref*{C-thm:Lminimax}]{nakiboglu19C}. 
That minimax relation, however, is best understood via the concept of Augustin-Legendre radius defined 
in the following.
\begin{definition}\label{def:Lradius}
	For any \(\rno\in\reals{+}\), channel \(\Wm:\inpS\to \pmea{\outA}\) with 
	a cost function \(\costf:\inpS\to \reals{\geq0}^{\ell}\), and \(\lgm \in \reals{\geq0}^{\ell}\),
	\emph{the order \(\rno\) Augustin-Legendre radius of \(\Wm\) for the Lagrange multiplier \(\lgm\)} is
	\begin{align}
\notag
	\RRL{\rno}{\Wm}{\lgm}
	&\DEF \inf\nolimits_{\mQ\in\pmea{\outA}} \sup\nolimits_{\dinp\in \inpS} \RD{\rno}{\Wm(\dinp)}{\mQ}-\lgm\cdot\costf(\dinp). 
	\end{align}
\end{definition}
Then as a result of \cite[Thm. \ref*{C-thm:Lminimax}]{nakiboglu19C},
for any \(\rno\in \reals{+}\), \(\Wm:\inpS\to \pmea{\outA}\) 
with \(\costf:\inpS\to \reals{\geq0}^{\ell}\), and 
\(\lgm \in \reals{\geq0}^{\ell}\) we have
\begin{align}
\label{eq:thm:Lminimaxradius}
\RCL{\rno}{\Wm}{\lgm}
&=\RRL{\rno}{\Wm}{\lgm}.
\end{align}
If in addition \(\RCL{\rno}{\Wm}{\lgm}\) is finite, then
there exits a unique \(\qma{\rno,\Wm}{\lgm}\!\in\!\pmea{\outA}\),
called \emph{the order \(\rno\) Augustin-Legendre center of \(\Wm\) for the Lagrange multiplier \(\lgm\)},
satisfying
\begin{align}
\notag
\RCL{\rno}{\Wm}{\lgm}
&=\sup\nolimits_{\dinp \in \inpS} \RD{\rno}{\Wm(\dinp)}{\qma{\rno,\Wm}{\lgm}}-\lgm\cdot\costf(\dinp).
\end{align}

The A-L information measures are defined through a standard application of 
the convex conjugation techniques.
However, starting with \cite[Thms. 8 and 10]{gallager65}
---i.e. the cost constrained variants of Gallager's bound---
the \renyi\!\!-Gallager (R-G) information measures rather than 
the A-L information measures have been the customary tools 
for applying convex conjugation techniques in the error 
exponent calculations, see for example \cite{ebert65,ebert66,richters67}.
A brief discussion of the R-G information information measures 
can be found in Appendix \ref{sec:RG-informationmeasures};
for a more detailed discussion see \cite{nakiboglu19C}. 

\subsection{The Sphere Packing Exponent}\label{sec:preliminary-spherepackingexponent}
\begin{definition}\label{def:spherepackingexponent}
For any \(\Wm:\inpS\to \pmea{\outA}\), \(\cset\subset\pdis{\inpS}\), 
and \(\rate\in\reals{\geq0}\), the SPE is
\begin{align}
\label{eq:def:spherepackingexponent}
\spe{\rate,\!\Wm\!,\cset}
&\DEF \sup\nolimits_{\rno\in (0,1)} \tfrac{1-\rno}{\rno} \left(\CRC{\rno}{\Wm\!}{\cset}-\rate\right).
\end{align}
We denote \(\cset=\pdis{\inpS}\) case by \(\spe{\rate,\Wm}\).
Furthermore, with a slight abuse of notation, we denote \(\cset=\{\mP\}\) case by \(\spe{\rate,\!\Wm\!,\mP}\)
and  \(\cset=\{\mP:\EXS{\mP}{\costf}\leq\costc\}\) case by \(\spe{\rate,\!\Wm\!,\costc}\).
\end{definition}

\begin{lemma}\label{lem:spherepacking}
	For any \(\Wm\!:\!\inpS\!\to\!\pmea{\outA}\), \(\cset\!\subset\!\pdis{\inpS}\), 
	\(\spe{\rate,\!\Wm\!,\cset}\) is nonincreasing and convex in \(\rate\) on \(\reals{\geq0}\), 
	finite on \((\CRC{0^{_{+}}\!}{\Wm\!}{\cset},\infty)\), and 
	continuous on \([\CRC{0^{_{+}}\!}{\Wm\!}{\cset},\infty)\)
	where \(\CRC{0^{_{+}}\!}{\Wm\!}{\cset}=\lim\nolimits_{\rno\downarrow0} \CRC{\rno}{\Wm\!}{\cset}\).
	In particular,
	\begin{align}
	\label{eq:lem:spherepacking}
	\spe{\rate,\!\Wm\!,\cset}
	&=\begin{cases}
	\infty 
	& \rate<\CRC{0^{_{+}}\!}{\Wm\!}{\cset}
	\\
	\sup_{\rno\in (0,1)} \tfrac{1-\rno}{\rno} \left(\CRC{\rno}{\Wm\!}{\cset}-\rate\right)
	&
	\rate=\CRC{0^{_{+}}\!}{\Wm\!}{\cset}
	\\
	\sup_{\rno\in [\rnf,1)} \tfrac{1-\rno}{\rno} \left(\CRC{\rno}{\Wm\!}{\cset}-\rate\right)
	&
	\rate=\CRC{\rnf}{\Wm}{\cset} \mbox{~for some~}\rnf\in(0,1)
	\\
	0
	&\rate\geq\CRC{1}{\Wm\!}{\cset}
	\end{cases}.
	\end{align} 
	\end{lemma}

Lemma \ref{lem:spherepacking} follows from the continuity and the monotonicity properties of 
\(\CRC{\rno}{\Wm\!}{\cset}\) established in \cite[Lemma \ref*{C-lem:capacityO}]{nakiboglu19C};
a proof can be found in Appendix \ref{sec:omitted-proofs}.
The proof of Lemma \ref{lem:spherepacking} is analogous to that of 
\cite[Lemma \ref*{B-lem:spherepackingexponent}]{nakiboglu19B}, which relies on 
\cite[Lemma \ref*{B-lem:capacityO}]{nakiboglu19B} instead of 
\cite[Lemma \ref*{C-lem:capacityO}]{nakiboglu19C}.

One can express \(\spe{\rate,\!\Wm\!,\cset}\) in terms of \(\spe{\rate,\!\Wm\!,\mP}\),
using the definitions  of
\(\CRC{\rno}{\Wm\!}{\cset}\),
\(\spe{\rate,\!\Wm\!,\cset}\),  
and \(\spe{\rate,\!\Wm\!,\mP}\):
\begin{align}
\notag
\spe{\rate,\!\Wm\!,\cset}
&=\sup\nolimits_{\rno\in(0,1)}\sup\nolimits_{\mP\in\cset}\tfrac{1-\rno}{\rno}\left(\RMI{\rno}{\mP}{\Wm}-\rate\right)
\\
\notag
&=\sup\nolimits_{\mP\in\cset}\sup\nolimits_{\rno\in(0,1)}\tfrac{1-\rno}{\rno}\left(\RMI{\rno}{\mP}{\Wm}-\rate\right)
\\
\label{eq:lem:spherepacking:compositionconstrained}
&=\sup\nolimits_{\mP\in\cset} \spe{\rate,\!\Wm\!,\mP}.
\end{align}

Lemma \ref{lem:spherepacking} holds for \(\spe{\rate,\!\Wm\!,\mP}\) 
by definition, but it can be strengthened significantly for 
\(\rate\)'s in 
\((\lim_{\rno\downarrow0}\RMI{\rno}{\mP}{\Wm},\RMI{1}{\mP}{\Wm}]\)
using the elementary properties of the Augustin information.
\begin{lemma}\label{lem:spherepacking-cc}
Let \(\Wm:\inpS\to \pmea{\outA}\) and \(\mP\in\pdis{\inpS}\) be such that 
\(\RMI{0^{_{+}}\!}{\mP}{\Wm}\neq \RMI{1}{\mP}{\Wm}\),
where
\(\RMI{0^{_{+}}\!}{\mP}{\Wm}\DEF\lim_{\rno\downarrow0}\RMI{\rno}{\mP}{\Wm}\).
Then for any rate 
\(\rate\in(\RMI{0^{_{+}}\!}{\mP}{\Wm},\RMI{1}{\mP}{\Wm}]\)
there exists a unique order \(\rns\in(0,1]\) satisfying 
\begin{align}
\label{eq:lem:spherepacking-cc:rate}
\rate
&=\RMI{1}{\mP}{\Wma{\rns}{\qmn{\rns,\mP}}}.
\end{align}
The orders \(\rns\) determined by \eqref{eq:lem:spherepacking-cc:rate} 
form an increasing continuous bijective function of rate \(\rate\), 
from \((\RMI{0^{_{+}}\!}{\mP}{\Wm},\RMI{1}{\mP}{\Wm}]\)
to \((0,1]\) satisfying 
\begin{align}
\label{eq:lem:spherepacking-cc:exponent}
\spe{\rate,\!\Wm\!,\mP}
&=\CRD{1}{\Wma{\rns}{\qmn{\rns,\mP}}}{\Wm}{\mP},
\\
\label{eq:lem:spherepacking-cc:slope}
\pder{}{\rate}\spe{\rate,\!\Wm\!,\mP}
&=\tfrac{\rns-1}{\rns}.
\end{align}
Thus \(\spe{\rate,\!\Wm\!,\mP}\) is  finite, convex, continuously differentiable, 
and decreasing in \(\rate\) on \((\RMI{0^{_{+}}\!}{\mP}{\Wm},\RMI{1}{\mP}{\Wm})\)
and its satisfies
\begin{align}
\label{eq:lem:spherepacking-cc:limit}
\spe{\RMI{0^{_{+}}\!}{\mP}{\Wm},\!\Wm\!,\mP}
&=\lim\nolimits_{\rno\downarrow0}\CRD{1}{\Wma{\rno}{\qmn{\rno,\mP}}}{\Wm}{\mP}.
\end{align}
Furthermore, if \(\spe{\RMI{0^{_{+}}\!}{\mP}{\Wm},\!\Wm\!,\mP}\) is finite then there exists a 
\(\Vm\!:\!\inpS\!\to\!\pmea{\outA}\) satisfying both 
\(\RMI{1}{\mP}{\Vm}=\RMI{0^{_{+}}}{\mP}{\Wm}\)
and
\(\CRD{1}{\Vm}{\Wm}{\mP}=\spe{\RMI{0^{_{+}}\!}{\mP}{\Wm},\!\Wm\!,\mP}\).
\end{lemma}
\begin{proof}[Proof of Lemma \ref{lem:spherepacking-cc}]
Note that \(\RMI{1}{\mP}{\Wma{\rno}{\qmn{\rno,\mP}}}\) is an increasing
and  continuous function of
the order \(\rno\) 
by \cite[Lemma \ref*{C-lem:informationO}-(\ref*{C-informationO:strictconvexity},\ref*{C-informationO:monotonicityofharoutunianinformation})]{nakiboglu19C} 
because \(\RMI{0^{_{+}}\!}{\mP}{\Wm}\neq \RMI{1}{\mP}{\Wm}\) by the hypothesis.
In addition \(\lim_{\rno\downarrow0}\RMI{1}{\mP}{\Wma{\rno}{\qmn{\rno,\mP}}}=\RMI{0^{_{+}}\!}{\mP}{\Wm}\)
by \cite[Lemma {\ref*{C-lem:informationO}-(\ref*{C-informationO:limitofharoutunianinformation})}]{nakiboglu19C}.
Then there exists a unique \(\rns\) satisfying \eqref{eq:lem:spherepacking-cc:rate}
by the intermediate value theorem \cite[4.23]{rudin}.
The function defined by \eqref{eq:lem:spherepacking-cc:rate} is an increasing continuous bijective   
function of the rate \(\rate\) from \((\RMI{0^{_{+}}\!}{\mP}{\Wm},\RMI{1}{\mP}{\Wm}]\)
to \((0,1]\)
because it is the inverse of an increasing continuous bijective function
from \((0,1]\) to \((\RMI{0^{_{+}}\!}{\mP}{\Wm},\RMI{1}{\mP}{\Wm}]\), 
i.e. \(\rno \rightsquigarrow\RMI{1}{\mP}{\Wma{\rno}{\qmn{\rno,\mP}}}\).

On the other hand \(\RMI{\rno}{\mP}{\Wm}\) is continuously differentiable in \(\rno\)
by \cite[Lemma {\ref*{C-lem:informationO}-(\ref*{C-informationO:differentiability})}]{nakiboglu19C};
then
\cite[{(\ref*{C-eq:lem:information:alternative:opt})} and
{(\ref*{C-eq:lem:informationO:differentiability-alt})}]{nakiboglu19C} 
imply
\begin{align}
\label{eq:spherepacking-cc-1}
\pder{}{\rno}\tfrac{1-\rno}{\rno}\left(\RMI{\rno}{\mP}{\Wm}-\rate\right)
&=\tfrac{1}{\rno^{2}}\left(\rate-\RMI{1}{\mP}{\Wma{\rno}{\qmn{\rno,\mP}}}\right).
\end{align}
Hence for any \(\rate\!\in\!(\RMI{0^{_{+}}\!}{\mP}{\Wm},\RMI{1}{\mP}{\Wm}]\),
the supremum in the definition of \(\spe{\rate,\!\Wm\!,\mP}\) is achieved at
the order \(\rns\)
satisfying \eqref{eq:lem:spherepacking-cc:rate}.
Then \eqref{eq:lem:spherepacking-cc:exponent} follows from
\cite[{(\ref*{C-eq:lem:information:alternative:opt})}]{nakiboglu19C}.
Furthermore, for any  \(\rate\in(\RMI{0^{_{+}}\!}{\mP}{\Wm},\RMI{1}{\mP}{\Wm})\)
and \(\overline{\rate}\geq0\) we have
\begin{align}
\notag
\spe{\overline{\rate},\!\Wm\!,\mP}
&\geq \tfrac{1-\rns(\rate)}{\rns(\rate)}(\RMI{\rns}{\mP}{\Wm}-\overline{\rate})
\\
\label{eq:spherepacking-cc-2}
&=\spe{\rate,\!\Wm\!,\mP}+\tfrac{1-\rns(\rate)}{\rns(\rate)}(\rate-\overline{\rate}).
\end{align}
For any  \(\overline{\rate}\in(\RMI{0^{_{+}}\!}{\mP}{\Wm},\RMI{1}{\mP}{\Wm})\)
and \(\rate\geq0\), 
following a similar analysis and reversing the roles of \(\rate\) and \(\overline{\rate}\) we obtain
\begin{align}
\label{eq:spherepacking-cc-3}
\spe{\rate,\!\Wm\!,\mP}
&\geq \spe{\overline{\rate},\!\Wm\!,\mP}+\tfrac{1-\rns(\overline{\rate})}{\rns(\overline{\rate})}(\overline{\rate}-\rate).
\end{align}
Since \(\rns\) is increasing and continuous in the rate, \eqref{eq:spherepacking-cc-2} and \eqref{eq:spherepacking-cc-3}
imply \eqref{eq:lem:spherepacking-cc:slope}
for all \(\rate\)'s in \((\RMI{0^{_{+}}\!}{\mP}{\Wm},\RMI{1}{\mP}{\Wm}]\).

For \(\rate=\RMI{0^{_{+}}\!}{\mP}{\Wm}\) case,
note that \(\tfrac{1-\rno}{\rno}(\RMI{\rno}{\mP}{\Wm}-\rate)\)
is decreasing \(\rno\) on \((0,1)\) by
\eqref{eq:spherepacking-cc-1} and 
\cite[Lemma {\ref*{C-lem:informationO}-(\ref*{C-informationO:monotonicityofharoutunianinformation},\ref*{C-informationO:limitofharoutunianinformation})}]{nakiboglu19C}. Thus
\begin{align}
\notag
\spe{\RMI{0^{_{+}}\!}{\mP}{\Wm},\!\Wm\!,\mP}
&=\lim\nolimits_{\rno\downarrow 0} \tfrac{1-\rno}{\rno}(\RMI{\rno}{\mP}{\Wm}-\RMI{0^{_{+}}\!}{\mP}{\Wm}).
\end{align}
Then  \eqref{eq:lem:spherepacking-cc:limit} follows from 
the mean value theorem \cite[5.10]{rudin}
and \cite[(\ref*{C-eq:lem:informationO:differentiability-alt})]{nakiboglu19C}.
Furthermore, if \(\spe{\RMI{0^{_{+}}\!}{\mP}{\Wm},\!\Wm\!,\mP}=\gamma\) for a
\(\gamma\in\reals{+}\), then
\(\RD{1}{\Wma{\rno}{\qmn{\rno,\mP}}(\dinp)}{\Wm(\dinp)}\leq \tfrac{\gamma}{\mP(\dinp)}\)
as a result of non-negativity of the \renyi divergence. Hence
\begin{align}
\notag
\int \GX\left(\der{\Wma{\rno}{\qmn{\rno,\mP}}(\dinp)}{\Wm(\dinp)}\right)
\Wm(\dif{\dout}|\dinp)\leq \tfrac{\gamma}{\mP(\dinp)}+\tfrac{1}{e}+1
\end{align}
for \(\GX(\tau)=\tau\IND{0\leq\tau<e}+\tau\ln\tau \IND{\tau\geq e}\)
because \(\tau\ln\tau\geq-\sfrac{1}{e}\).
Then \(\{\der{\Wma{\rno}{\qmn{\rno,\mP}}(\dinp)}{\Wm(\dinp)}\}_{\rno\in(0,1)}\)
are uniformly \(\Wm(\dinp)\)-integrable by \cite[Thm 4.5.9]{bogachev}, i.e.
by the necessary and sufficient condition for the uniform 
integrability determined by de la Vall\'{e}e Poussin.
Thus any sequence of members of \(\{\Wma{\rno}{\qmn{\rno,\mP}}(\dinp)\}_{\rno\in(0,1)}\)
has a convergent subsequence for the topology of setwise convergence 
by \cite[Thm. 4.7.25]{bogachev}.
For each \(\dinp\in\supp{\mP}\), 
let \(\Vm(\dinp)\) be the limit point for the aforementioned subsequence for
the sequence \(\{\Wma{\sfrac{1}{\knd}}{\qmn{\sfrac{1}{\knd},\mP}}(\dinp)\}_{\knd\in\integers{+}}\).
Then \eqref{eq:lem:spherepacking-cc:rate}, \eqref{eq:lem:spherepacking-cc:exponent},
and the lower semicontinuity of the \renyi divergence in its arguments for the topology of
setwise convergence, i.e.\cite[Thm. 15]{ervenH14}, imply
\(\RMI{1}{\mP}{\Vm}\leq\RMI{0^{_{+}}}{\mP}{\Wm}\)
and
\(\CRD{1}{\Vm}{\Wm}{\mP}\leq \gamma\).
On the other hand as a result of the definition of \(\spe{\rate,\!\Wm\!,\mP}\) and
\cite[Lemma {\ref*{C-lem:information}-(\ref*{C-information:alternative})}]{nakiboglu19C},
we have
\begin{align}
\notag
\spe{\rate,\!\Wm\!,\mP}
&=\sup\nolimits_{\rno\in (0,1)}\inf\nolimits_{\Vm\in\pmea{\outA|\inpS}}
\CRD{1}{\Vm}{\Wm}{\mP}+\tfrac{1-\rno}{\rno} \left(\RMI{1}{\mP}{\Vm}-\rate\right)
\\
\notag
&\leq \sup\nolimits_{\rno\in (0,1)}\CRD{1}{\Vm}{\Wm}{\mP}+\tfrac{1-\rno}{\rno} \left(\RMI{1}{\mP}{\Vm}-\rate\right)
\\
\notag
&=
\begin{cases}
\CRD{1}{\Vm}{\Wm}{\mP}
&\rate\geq\RMI{1}{\mP}{\Vm}
\\
\infty
&\rate<\RMI{1}{\mP}{\Vm}
\end{cases}.
\end{align} 
Thus \(\RMI{1}{\mP}{\Vm}\) cannot be less than \(\RMI{0^{_{+}}}{\mP}{\Wm}\) because 
\(\spe{\rate,\Wm,\mP}\)
is infinite for all  \(\rate<\RMI{0^{_{+}}}{\mP}{\Wm}\).
Hence\(\RMI{1}{\mP}{\Vm}=\RMI{0^{_{+}}}{\mP}{\Wm}\).
Consequently \(\CRD{1}{\Vm}{\Wm}{\mP}\) cannot be less than \(\gamma\) because 
\(\spe{\RMI{0^{_{+}}}{\mP}{\Wm},\Wm,\mP}=\gamma\).
Hence \(\CRD{1}{\Vm}{\Wm}{\mP}=\gamma\).
\end{proof}

Lemma \ref{lem:spherepacking-cc} provides a simple confirmation of the alternative expression for
\(\spe{\rate,\!\Wm\!,\mP}\), which is commonly known as Haroutunian's form \cite{haroutunian68}. 
\begin{lemma}\label{lem:haroutunianform}
	For any \(\Wm:\inpS\to \pmea{\outA}\), \(\mP\in\pdis{\inpS}\), and \(\rate\in\reals{+}\)
	\begin{align}
	\label{eq:lem:haroutunianform}
	\spe{\rate,\!\Wm\!,\mP}
	&=\inf\nolimits_{\Vm:\RMI{1}{\mP}{\Vm}\leq\rate}
	\CRD{1}{\Vm}{\Wm}{\mP}.
	\end{align}
\end{lemma}
\begin{proof}[Proof of Lemma \ref{lem:haroutunianform}]
If \(\rate\in[\RMI{1}{\mP}{\Wm},\infty)\),
then \eqref{eq:lem:haroutunianform} holds because 
the expression on the right hand side of \eqref{eq:lem:haroutunianform}
is zero as a result of the substitution \(\Vm\!=\!\Wm\)
and the non-negativity of the \renyi divergence.

On the other hand, as a result of the definition of \(\spe{\rate,\!\Wm\!,\mP}\),
\cite[Lemma {\ref*{C-lem:information}-(\ref*{C-information:alternative})}]{nakiboglu19C},
and the max-min inequality we have
\begin{align}
\notag
\spe{\rate,\!\Wm\!,\mP}
&=\sup\nolimits_{\rno\in (0,1)}\inf\nolimits_{\Vm\in\pmea{\outA|\inpS}}
\CRD{1}{\Vm}{\Wm}{\mP}+\tfrac{1-\rno}{\rno} \left(\RMI{1}{\mP}{\Vm}-\rate\right).
\\
\notag
&\leq \inf\nolimits_{\Vm\in\pmea{\outA|\inpS}}\sup\nolimits_{\rno\in (0,1)}
\CRD{1}{\Vm}{\Wm}{\mP}+\tfrac{1-\rno}{\rno} \left(\RMI{1}{\mP}{\Vm}-\rate\right)
\\
\notag
&=
\inf\nolimits_{\Vm:\RMI{1}{\mP}{\Vm}\leq\rate}
\CRD{1}{\Vm}{\Wm}{\mP}.
\end{align}
Then \eqref{eq:lem:haroutunianform} holds whenever \(\spe{\rate,\!\Wm\!,\mP}\) is infinite,
i.e. for all \(\rate\!\in\![0,\RMI{0^{_{+}}\!}{\mP}{\Wm})\)
and possibly for \(\rate=\RMI{0^{_{+}}\!}{\mP}{\Wm}\), trivially 
and whenever \(\spe{\rate,\!\Wm\!,\mP}\) is finite 
as a result of Lemma \ref{lem:spherepacking-cc}. 
\end{proof}

Haroutunian's form implies the following sufficient condition for the optimality
of an order \(\rno\) in the definition of the SPE given in \eqref{eq:def:spherepackingexponent}.

\begin{lemma}\label{lem:spherepacking-optimality}
	For any \(\Wm\!:\!\inpS\!\to\!\pmea{\outA}\), \(\cset\!\subset\!\pdis{\inpS}\), 
	and \(\rate\in (\CRC{0^{_{+}}\!}{\Wm\!}{\cset},\CRC{1}{\Wm\!}{\cset})\) if
	there exists an \(\rns\in(0,1)\) 
	and a function \(\Vmn{\mP}\) of \(\mP\) from \(\cset\) to
	\(\pmea{\outA|\inpS}\)  satisfying the following two inequalities
	\begin{align}
	\label{eq:lem:spherepacking-optimality:hypothesis-rate}
	\CRD{1}{\Vmn{\mP}}{\qmn{\rns,\!\Wm\!,\cset}}{\mP}
	&\leq \rate
	&
	&\forall \mP\in\cset,
	\\
	\label{eq:lem:spherepacking-optimality:hypothesis-exponent}
	\CRD{1}{\Vmn{\mP}}{\Wm}{\mP}
	&\leq \tfrac{1-\rns}{\rns}(\CRC{\rns}{\Wm}{\cset}-\rate)
	&
	&\forall \mP\in\cset,
	\end{align}
	then \(\spe{\rate,\!\Wm\!,\cset}=\tfrac{1-\rns}{\rns}(\CRC{\rns}{\Wm}{\cset}-\rate)\).
\end{lemma}
For some channels, \(\!\Vmn{\mP}\!=\!\Wma{\rns}{\qmn{\rns,\Wm,\cset}}\) satisfies
both \eqref{eq:lem:spherepacking-optimality:hypothesis-rate} 
and
\eqref{eq:lem:spherepacking-optimality:hypothesis-exponent};
for these channels, the value of SPE can be determined using 
Lemma \ref{lem:spherepacking-optimality}.
However, for an arbitrary channel, rate, 
and the corresponding optimal order \(\rns\) in \eqref{eq:def:spherepackingexponent}
a \(\Vmn{\mP}\) satisfying both
\eqref{eq:lem:spherepacking-optimality:hypothesis-rate} 
and
\eqref{eq:lem:spherepacking-optimality:hypothesis-exponent}
might not exist,
e.g. the \(Z\)-channel discussed in Appendix \ref{sec:blahut}.
Had the sufficient condition for the optimality of the order 
\(\rns\) 
given in 
\eqref{eq:lem:spherepacking-optimality:hypothesis-rate} 
and
\eqref{eq:lem:spherepacking-optimality:hypothesis-exponent}
been also necessary, Blahut's proof in \cite{blahut74}
would have been correct; this, 
however, is not the case in general as we demonstrate
in Appendix \ref{sec:blahut}.
It is worth mentioning that for the channels satisfying
the necessary conditions given 
\eqref{eq:lem:spherepacking-optimality:hypothesis-rate} 
and
\eqref{eq:lem:spherepacking-optimality:hypothesis-exponent},
one can derive the SPB using the approach presented in  \cite{blahut74}.

\begin{proof}[Proof of Lemma \ref{lem:spherepacking-optimality}]
	Note that  as a result of \eqref{eq:augustinslaw} we have
	\begin{align}
	\notag
	\CRD{1}{\Vmn{\mP}}{\qmn{\rns,\!\Wm\!,\cset}}{\mP}
	&=\RMI{1}{\mP}{\Vmn{\mP}}
	+\RD{1}{\sum\nolimits_{\dinp}\mP(\dinp)\Vmn{\mP}(\dinp)}{\qmn{\rns,\!\Wm\!,\cset}}.
	\end{align}
	Thus \(\RMI{1}{\mP}{\Vmn{\mP}}\!\leq\!\rate\) 
	for all \(\mP\!\in\!\cset\) because the \renyi divergence is non-negative.
	Then \(\spe{\rate,\!\Wm\!,\mP}\!\leq\!\tfrac{1-\rns}{\rns}(\CRC{\rns}{\Wm}{\cset}-\rate)\)
	for all \(\mP\in\cset\) by Lemma \ref{lem:haroutunianform}. Then 
	\(\spe{\rate,\!\Wm\!,\cset}\leq\tfrac{1-\rns}{\rns}(\CRC{\rns}{\Wm}{\cset}-\rate)\)
	by \eqref{eq:lem:spherepacking:compositionconstrained}.
	On the other hand \(\spe{\rate,\!\Wm\!,\cset}\geq\tfrac{1-\rns}{\rns}(\CRC{\rns}{\Wm}{\cset}-\rate)\)
	by definition. 
	Thus \(\spe{\rate,\!\Wm\!,\cset}=\tfrac{1-\rns}{\rns}(\CRC{\rns}{\Wm}{\cset}-\rate)\).
\end{proof}

\subsection{Augustin's Variant of Gallager's Bound}\label{sec:innerbound}
The SPE is an upper bound on the exponential decay rate of the optimal error probability 
with block length, i.e. on the reliability function, for memoryless channels
satisfying rather mild hypotheses, with or without the list decoding, as a result of the SPBs
given Theorem \ref{thm:exponent-convex} and \ref{thm:exponent-cost-ologn} of \S \ref{sec:outerbound}.
For the list decoding, the SPE is also a lower bound 
on the exponential decay rate of the optimal error probability with block length
\cite[ex 5.20]{gallager}, \cite[ex 10.28]{csiszarkorner},\cite{elias57}.
The latter observation can be confirmed using standard results such as 
\cite{gallager65,poltyrev82}, albeit with minor modifications, as well.
In the following, we confirm this observation using a variant of Gallager's 
bound in terms of the Augustin information.
Recall that Gallager's bound is derived, customarily, for
the \renyi information, rather than the Augustin information.
The fixed point property described in \eqref{eq:augustinfixedpoint}
plays a critical role in the proof. 
We name this variant of Gallager's bound after 
Augustin because \cite[Lemma 36.1]{augustin78} of Augustin is 
the first achievability result making use of the fixed point property 
described in \eqref{eq:augustinfixedpoint},
to the best of out knowledge. 
\begin{lemma}\label{lem:AVGIB}
For any 
\(M,L\in \integers{+}\) s.t. \(L<M\),
\(\Wm\!:\inpS\to\pmea{\outA}\), \(\mP\in\pdis{\inpS}\),
\(\cinpS\subset \inpS\),
and \(\rno\in [\tfrac{1}{1+L},1)\)
there exists an \((M,L)\) channel code with an encoding function of the form
\(\enc:\mesS\to\cinpS\) satisfying
\begin{align}
\notag
\ln \Pe
&\leq \tfrac{\rno-1}{\rno} \left[\tau-\ln \tfrac{(M-1)e}{L}\right]-\tfrac{\ln \mP(\cinpS)}{\rno}
\end{align}
where \(\tau=\inf\nolimits_{\dinp \in \cinpS}\RD{\rno}{\Wm(\dinp)}{\qmn{\rno,\mP}}\).
\end{lemma}

We do not present a separate proof for Lemma \ref{lem:AVGIB}
because Lemma \ref{lem:AVGIB} is zero Lagrange multiplier special case of 
Lemma \ref{lem:LAVGIB} presented in the following.
Before stating Lemma \ref{lem:LAVGIB}, let us point out an immediate consequence of 
Lemma \ref{lem:AVGIB} for the constant composition codes on the stationary memoryless 
channels. 
Recall that for any composition \(\mP\) for the block length \(\blx\), 
the probability of the set of all composition \(\mP\) sequences
(i.e. the probability of \(\cal{T}_{\mP,\blx}\))
for i.i.d. samples
(i.e. \(\pmn{\tin}=\mP\) for all \(\tin\)) satisfies the following
identity for some \(\xi\in[0,1]\) by \cite[p. 26]{csiszarkorner}
\begin{align}
\notag
\left(\bigotimes\nolimits_{\tin=1}^{\blx}\pmn{\tin}\right) (\cal{T}_{\mP,\blx})
&=e^{-\xi\frac{\abs{\supp{\mP}}}{12\ln2}}
(2\pi \blx)^{-\frac{\abs{\supp{\mP}}-1}{2}}\sqrt{\prod\nolimits_{\dinp:\mP(\dinp)>0}\tfrac{1}{\mP(\dinp)}}.
\end{align}
Then using 
\(\qmn{\rno,\bigotimes\nolimits_{\tin=1}^{\blx}\pmn{\tin}}=
\bigotimes\nolimits_{\tin=1}^{\blx}\qmn{\rno,\pmn{\tin}}\),
established in \cite[Lemma \ref*{C-lem:information:product}]{nakiboglu19C}, we get the following corollary
by setting \(\cinpS\) to \({\cal T}_{\mP,\blx}\).

\begin{corollary}\label{cor:composition}
For any \(\blx\in\integers{+}\), \(\rnf\in(0,1)\),
\(\{\Wmn{\tin}\}_{\tin\in\integers{+}}\)
satisfying \(\Wmn{\tin}=\Wm\) for some \(\Wm\!:\inpS\to\pmea{\outA}\), 
\(\mP\in\pdis{\inpS}\)
satisfying \(\blx\mP(\dinp)\in\integers{\geq0}\) for all \(\dinp\in\inpS\), 
integers \(M,L\) satisfying
both \(\tfrac{1}{L+1}<\rnt\) and \(\tfrac{1}{\blx}\ln\tfrac{M}{L}=\RMI{\rnt}{\mP}{\Wm}\)
for some \(\rnt\in [\rnf,1)\),
there exists an \((M,L)\)  channel code 
on \(\bigotimes\nolimits_{\tin=1}^{\blx}\Wmn{\tin}\) 
with an encoding function of the form 
\(\enc:\mesS\to\cal{T}_{\mP,\blx}\) satisfying
\begin{align}
\notag
\ln \Pe
&\leq -\blx\spe{\tfrac{1}{\blx}\ln\tfrac{M}{L},\Wm\!,\mP}
+\tfrac{1}{\rnt}\left[1-\rnt+
\tfrac{\abs{\supp{\mP}}}{12\ln2} 
+\tfrac{\abs{\supp{\mP}}-1}{2}\ln(2\pi \blx)
+\tfrac{1}{2} \sum\nolimits_{\dinp:\mP(\dinp)>0}\ln \mP(\dinp)
\right].
\end{align}
\end{corollary}

\begin{lemma}\label{lem:LAVGIB}
For any \(\ell,M,L\in \integers{+}\) s.t. \(L<M\),
\(\Wm\!:\inpS\to\pmea{\outA}\),
\(\costf:\inpS\to\reals{\geq0}^{\ell}\),
 \(\mP\in\pdis{\inpS}\), 
 \(\cinpS\subset \inpS\), 
 and \(\rno\in [\tfrac{1}{1+L},1)\)
there exists an \((M,L)\) channel code with an encoding function of the form
\(\enc:\mesS\to\cinpS\) such that
\begin{equation}
\label{eq:lem:LAVGIB}
\ln \Pe
\leq \tfrac{\rno-1}{\rno}\left[\tau
+(\inf\nolimits_{\dinp\in\cinpS}\lgm\cdot\costf(\dinp))-\ln \tfrac{(M-1)e}{L}\right]-\tfrac{\ln\mP(\cinpS)}{\rno}
\end{equation}
where \(\tau=\inf\nolimits_{\dinp \in \cinpS}\RD{\rno}{\Wm(\dinp)}{\qmn{\rno,\mP}}-\lgm\cdot\costf(\dinp)\).
\end{lemma}
It is worth mentioning that one can obtain a similar bound in terms of the R-G information, 
see Lemma \ref{lem:LGIB} in Appendix \ref{sec:RG-informationmeasures}.
Before presenting  the proof of Lemma \ref{lem:LAVGIB}, let us discuss its consequences briefly.
For the cost constrained memoryless channels with additive cost functions one can apply
Lemma \ref{lem:LAVGIB} for \(\lgm=0\) and for the product input distributions satisfying 
the cost constraint \(\costc-\epsilon\). 
Then using the weak law of large numbers together with 
\cite[Lemmas \ref*{C-lem:information:product}, \ref*{C-lem:CCcapacity},
\ref*{C-lem:CCcapacityproduct}]{nakiboglu19C}
one can conclude that the reliability function is bounded from below
by the SPE, under rather mild hypotheses.
This observation can be strengthened for \(\ell=1\) case by invoking 
Berry-Esseen theorem via \cite[Lemma \ref*{B-lem:berryesseenN}]{nakiboglu19B},
under appropriate hypotheses.
The following corollary can be generalized to \(\ell>1\) case in a straightforward 
way once a result analogous to \cite[Lemma \ref*{B-lem:berryesseenN}]{nakiboglu19B}
for sums of random vectors is obtained.

\begin{corollary}\label{cor:cost}
	Let \(\blx\in\!\integers{+}\) s.t \(\blx\geq 21\), \(\rnf\!\in\!(0,1)\),
	\(\Wmn{[1,\blx]}\!:\!\inpS_{1}^{\blx}\to\pmea{\outA_{1}^{\blx}}\) be a length \(\blx\)  product channel
	with an additive cost function \(\costf_{[1,\blx]}:\inpS_{1}^{\blx}\to\reals{\geq0}\) 
	satisfying \(\inf_{\dinp_{\tin}\in\inpS_{\tin}} \costf_{\tin}(\dinp_{\tin})=0\) for all 
	\(\tin\in[1,\blx]\) and
	\begin{align}
	\label{eq:cor:cost:hypothesis}
	\CRC{\rno}{\widetilde{\Wm}_{[1,\blx]}}{\blx \costc}
	&\geq \CRC{\rno}{\Wmn{[1,\blx]}}{\blx\costc}-\epsilon
	&
	&\forall\rno\in[\rnf,1),~\costc\in[\widetilde{\costc},\widetilde{\costc}+\tfrac{9\varsigma}{\blx}]
	\end{align}
for some
\(\widetilde{\costc}\!\in\!\reals{+}\), \(\varsigma\!\in\!\reals{\geq0}\), 
\(\epsilon\!\in\!\reals{\geq0}\),
and product channel 
\(\widetilde{\Wm}_{[1,\blx]}\!:\!\widetilde{\inpS}_{1}^{\blx}\!\to\!\pmea{\outA_{1}^{\blx}}\)
with the input sets \(\widetilde{\inpS}_{\tin}\!=\!\{\dinp_{\tin}\in\inpS_{\tin}\!:\!\costf_{\tin}(\dinp_{\tin})\leq\varsigma\}\)
satisfying 
\(\widetilde{\Wm}_{[1,\blx]}(\dinp_{1}^{\blx})=\Wmn{[1,\blx]}(\dinp_{1}^{\blx})\) 
for all \(\dinp_{1}^{\blx}\in\widetilde{\inpS}_{1}^{\blx}\).
Then for any \(\delta\!\in\!(0,\widetilde{\costc})\),
\(\costc\!\in\!(\widetilde{\costc}+\tfrac{3 e\varsigma}{\blx},\widetilde{\costc}+\tfrac{9\varsigma}{\blx})\),
and integers \(M,L\) satisfying 
both \(\tfrac{1}{L+1}<\rnt\) and \(\ln\tfrac{M}{L}=\CRC{\rnt}{\Wmn{[1,\blx]}}{\blx\costc}\)
for some \(\rnt\in [\rnf,1)\),
there exists an \((M,L)\) channel code with an encoding function of the form
\(\enc:\mesS\to\{\dinp_{1}^{\blx}\in\widetilde{\inpS}_{1}^{\blx}:\costf_{[1,\blx]}(\dinp_{1}^{\blx})\leq\blx \costc\}\)
such that
\begin{align}
	\label{eq:cor:cost}
\ln \Pe
&\leq -\spe{\ln\tfrac{M}{L},\Wmn{[1,\blx]}\!,\blx\costc}
+\tfrac{1-\rnt}{\rnt}
\left(\tfrac{\CRC{\rnt}{\Wmn{[1,\blx]}}{\blx\costc}}{\blx}\tfrac{6 \varsigma e}{\delta}+2\epsilon+1\right)
+\tfrac{\ln 4\blx}{2\rnt}+\tfrac{1}{\blx}.
	\end{align}
\end{corollary}

Note that if the region for \(\costc\) had not depended on 
\(\varsigma\) in \eqref{eq:cor:cost:hypothesis} and 
the interval for \(\rno\) had been a compact subset of \((0,1)\), 
then 
one could have found a \(\varsigma\) for any positive \(\epsilon\) 
by invoking Dini's theorem \cite[2.4.10]{dudley} through a construction
similar to the one used in the proof of 
\cite[Lemma \ref*{A-lem:capacityO}-(\ref*{A-capacityO-dini})]{nakiboglu19A}
 because 
\(\CRC{\rno}{\Wm}{\costc}\) is continuous in \((\rno,\costc)\) on 
\((0,1)\!\times\!\reals{+}\) by
\cite[Lemma \ref*{C-lem:CCcapacity}-(\ref*{C-CCcapacity:zeroonedichotomy})]{nakiboglu19C}.
This, however, is not the case and there might be product channels for which 
\eqref{eq:cor:cost:hypothesis} is not satisfied for any 
\((\epsilon,\varsigma,\widetilde{\costc})\) triple.
Furthermore, even when the hypotheses of Corollary \ref{cor:cost} are satisfied 
the tightness of the bound given in \eqref{eq:cor:cost} depend on the values of
the constants satisfying the hypotheses.
Nevertheless it is easy to see that 
if either the cost functions of the component channels are all bounded
or
 \(\CRC{\rnf}{\Wmn{[1,\blx]}}{\blx\costc}\) is \(\bigtheta{\blx}\)
---i.e. scaling linearly with  \(\blx\)---
and 
\eqref{eq:cor:cost:hypothesis} is satisfied for 
\(\epsilon\) and \(\varsigma\) that are \(\bigo{\ln \blx}\)
---i.e. scaling no faster than linearly with \(\ln \blx\)---
then  \eqref{eq:cor:cost} has a prefactor
that is \(e^{-\bigo{\ln \blx}}\), i.e. vanishing polynomially with \(\blx\).
Proof of Corollary \ref{cor:cost} is presented in Appendix \ref{sec:omitted-proofs}.

\begin{proof}[Proof of Lemma \ref{lem:LAVGIB}]
	We establish the existence of the code with the desired properties through a random coding 
	argument. 
	For any \(\mP\in\pdis{\inpS}\) and \(\cinpS\subset\inpS\) 
	satisfying \(\mP(\cinpS)>0\), let  \(\pmn{\cinpS}\in\pdis{\inpS}\) 
	be
	\begin{align}
	\label{eq:def:samplingdistribution}
	\pmn{\cinpS}(\dinp)
	&\DEF\tfrac{\IND{\dinp \in \cinpS} \mP(\dinp)}{\mP(\cinpS)}
	&
	&\forall \dinp\in\inpS.
	\end{align}
	Let us consider an ensemble of codes in which the assignments of 
	the messages to the elements of \(\cinpS\) are jointly independent
	and \(\enc(\dmes)=\dinp\) with probability  \(\pmn{\cinpS}(\dinp)\)
	for all \(\dmes\) in the message set.
	The decoder chooses \(L\) messages  with the greatest \(\tfrac{\fX_{\enc(\dmes)}}{\hX_{\enc(\dmes)}}\) 
	for \(\fX_{\dinp}\) and \(\hX_{\dinp}\) defined in \eqref{eq:def:likelihood} and
	\eqref{eq:def:weight} in the following.
	When there is a tie, the decoder prefers the message or messages with the lower index among the 
	tied messages.
	\begin{align}
	\label{eq:def:likelihood}
	\fX_{\dinp}(\dout)
	&\DEF \der{\Wm(\dinp)}{\mQ}(\dout)
	&
	&\forall \dinp\in\inpS,\dout\in\outS
	\\
	\label{eq:def:weight}
	\hX_{\dinp}
	&\DEF e^{\frac{\rno-1}{\rno}\lgm\cdot\costf(\dinp)}
	&
	&\forall \dinp\in\inpS
	\end{align}
	where \(\mQ\in\pmea{\outA}\) to be chosen later satisfies \(\qmn{\rno,\mP}\AC\mQ\).
	In order to bound the expected value of the error probability of the code over the ensemble, 
	let us consider the expected value  of the conditional error probability of the message with 
	the greatest index. An error occurs only when \(\tfrac{\fX}{\hX}\) associated with \(L\) or
	more other messages are at least as large as \(\tfrac{\fX}{\hX}\) associated with 
	the transmitted message. 
	We can bound this probability using a \(\dout\)-dependent auxiliary threshold 
	\(\gamma(\dout)\) and obtain the following
	bound on the expected value of \(\Pe\) over 
	the ensemble.\footnote{Note that
		\(\sum_{\tin=L}^{M}\binom{M}{\tin}\mS^{\tin}(1-\mS)^{M-\tin}
		=\binom{M}{L}\mS^{L}\sum_{\tin=0}^{M-L}\tfrac{L!(M-L)!}{(L+\tin)!(M-L-\tin)\!}\mS^{\tin}(1-\mS)^{M-L-\tin}
		\leq\binom{M}{L}\mS^{L}\) for all \(\mS\in[0,1]\).}
	\begin{align}
	\label{eq:boundontheensembleaverage}
	\EX{\Pe} 
	&\leq 
	\sum\nolimits_{\dinp} \pmn{\cinpS}(\dinp)
	\EXS{\mQ}{\IND{\frac{\fX_{\dinp}}{\hX_{\dinp}}\leq \gamma}\fX_{\dinp}}
	+\tbinom{M-1}{L}
	\sum\nolimits_{\dinp} \pmn{\cinpS}(\dinp)
	\EXS{\mQ}{\IND{\frac{\fX_{\dinp}}{\hX_{\dinp}}>\gamma} 
		\left[\sum\nolimits_{\dsta} \pmn{\cinpS}(\dsta)\IND{\frac{\fX_{\dsta}}{\hX_{\dsta}}\geq \frac{\fX_{\dinp}}{\hX_{\dinp}}} \right]^{L}
		\fX_{\dinp}}
	\end{align}
We invoke \(\rno\leq 1\) to bound the first term in the above sum,
and \(\rno\geq\tfrac{1}{1+L}\) to bound the second term. 
	\begin{align}
	\notag
	\EXS{\mQ}{\IND{\frac{\fX_{\dinp}}{\hX_{\dinp}}\leq \gamma}\fX_{\dinp}}
	&\leq \hX_{\dinp}
	\EXS{\mQ}{\IND{\frac{\fX_{\dinp}}{\hX_{\dinp}}\leq \gamma}
		\left(\tfrac{\fX_{\dinp}}{\hX_{\dinp}}\right)^{\rno}\gamma^{1-\rno}}
	\\
	\notag
	\EXS{\mQ}{
		\IND{\frac{\fX_{\dinp}}{\hX_{\dinp}}>\gamma}
		\left[\sum\nolimits_{\dsta} \pmn{\cinpS}(\dsta)\IND{\frac{\fX_{\dsta}}{\hX_{\dsta}}\geq \frac{\fX_{\dinp}}{\hX_{\dinp}}} \right]^{L}
		\fX_{\dinp}}
	&\leq \hX_{\dinp}\EXS{\mQ}{\IND{\frac{\fX_{\dinp}}{\hX_{\dinp}}>\gamma}
		\left[\sum\nolimits_{\dsta} \pmn{\cinpS}(\dsta)
		\left(\tfrac{\fX_{\dsta}}{\hX_{\dsta}}\right)^{\rno}\right]^{L}
		\left(\tfrac{\fX_{\dinp}}{\hX_{\dinp}}\right)^{1-L \rno}}
	\\
	\notag
	&\leq \hX_{\dinp}\EXS{\mQ}{\IND{\frac{\fX_{\dinp}}{\hX_{\dinp}}>\gamma}
		\left[\sum\nolimits_{\dsta} \pmn{\cinpS}(\dsta)
		\left(\tfrac{\fX_{\dsta}}{\hX_{\dsta}}\right)^{\rno}\right]^{L}
		\left(\tfrac{\fX_{\dinp}}{\hX_{\dinp}}\right)^{\rno}
		\gamma^{1-\rno-L \rno}}
	\end{align}
	If we set 
	\(\gamma=\left[\sum\nolimits_{\dinp} \pmn{\cinpS}(\dinp)
	\left(\tfrac{\fX_{\dinp}}{\hX_{\dinp}}\right)^{\rno}\right]^{\frac{1}{\rno}}
	\left[\tbinom{M-1}{L}\right]^{\frac{1}{L\rno}}
	\) we get 
	\begin{align}
	\notag
	\EX{\Pe} 
	&\leq 
	\sum\nolimits_{\dinp} \pmn{\cinpS}(\dinp)\hX_{\dinp} 
	\EXS{\mQ}{\left(\tfrac{\fX_{\dinp}}{\hX_{\dinp}}\right)^{\rno}\gamma^{1-\rno}}
	\\
	\notag
	&=\EXS{\mQ}{
		\left[\sum\nolimits_{\dinp} \pmn{\cinpS}(\dinp) \left(\tfrac{\fX_{\dinp}}{\hX_{\dinp}}\right)^{\rno}\right]^{\frac{1}{\rno}}}
	(\sup\nolimits_{\dinp\in\cinpS}\hX_{\dinp})
	\left[\tbinom{M-1}{L}\right]^{\frac{1-\rno}{L\rno}}.
	\end{align}
	On the other hand the Stirling's approximation for the factorials,   
	i.e. \(\sqrt{2\pi \blx} (\sfrac{\blx}{e})^{\blx}\leq  \blx!  \leq  e \sqrt{\blx} (\sfrac{\blx}{e})^{\blx}\),
	and the identity \(\ln \dsta \leq \dsta-1\) imply that
	\begin{align}
	\notag
	\tfrac{1}{L}\ln \tbinom{M-1}{L}
	&\leq \tfrac{1}{L}\ln \tfrac{e\sqrt{M-1}}{2\pi \sqrt{L(M-1-L)}}  +\ln \tfrac{M-1}{L}+\tfrac{M-1-L}{L}  \ln \left(1+\tfrac{L}{M - 1-L}\right)   
	\\
	\notag
	&\leq \ln \tfrac{M-1}{L}+1.
	\end{align}
	Then we get the following bound on the expected value of \(\Pe\) over the ensemble.
	\begin{align}
	\label{eq:lem:LIB-General}
	\ln \EX{\Pe} 
	&\leq 
	\ln \EXS{\mQ}{
		\left(\sum\nolimits_{\dinp} \pmn{\cinpS}(\dinp)
		e^{(1-\rno)\lgm\cdot\costf(\dinp)}
		\left(\der{\Wm(\dinp)}{\mQ}\right)^{\rno}
		\right)^{\frac{1}{\rno}}}
	+\tfrac{\rno-1}{\rno}\left[(\inf\nolimits_{\dinp\in\cinpS}\lgm\cdot\costf(\dinp))
	-\ln \tfrac{(M-1)e}{L}\right]
\intertext{If we set \(\mQ=\qmn{\rno,\mP}\), using first the definition of the tilted channel given in \eqref{eq:def:tiltedchannel} 
	and the definition of \(\tau\) we get}
	\notag
	\ln \EX{\Pe} 
	&\leq  \ln \EXS{\qmn{\rno,\mP}}{
		\left(\sum\nolimits_{\dinp} \pmn{\cinpS}(\dinp)
		e^{(\rno-1)[\RD{\rno}{\Wm(\dinp)}{\qmn{\rno,\mP}}-\lgm\cdot\costf(\dinp)]}
		\der{\Wma{\rno}{\qmn{\rno,\mP}}(\dinp)}{\qmn{\rno,\mP}}\right)^{\frac{1}{\rno}}}
	+\tfrac{\rno-1}{\rno}\left[(\inf\nolimits_{\dinp\in\cinpS}\lgm\cdot\costf(\dinp))
	-\ln \tfrac{(M-1)e}{L}\right]
	\\
	\notag
	&\leq \ln \EXS{\qmn{\rno,\mP}}{
		\left(\sum\nolimits_{\dinp} \pmn{\cinpS}(\dinp)
		\der{\Wma{\rno}{\qmn{\rno,\mP}}(\dinp)}{\qmn{\rno,\mP}}\right)^{\frac{1}{\rno}}}
	+\tfrac{\rno-1}{\rno}\left[\tau+(\inf\nolimits_{\dinp\in\cinpS}\lgm\cdot\costf(\dinp))
	-\ln \tfrac{(M-1)e}{L}\right].
	\end{align}
	Since there exists a code with \(\Pe\) less than or equal to \(\EX{\Pe}\),
	the existence of a code satisfying \eqref{eq:lem:LAVGIB}
	with an encoding function of the form \(\enc:\mesS\to\cinpS\)
	follows from 
	\eqref{eq:def:Aoperator}, \eqref{eq:augustinfixedpoint}, 
	and \eqref{eq:def:samplingdistribution}.
\end{proof}

\section[SPBs For Memoryless Channels]{Sphere Packing Bounds For Memoryless Channels}\label{sec:outerbound}
\begin{assumption}\label{assumption:convex-ologn}
	The maximum \(\RC{\sfrac{1}{2}}{\Uma{\cset}{(\tin)}}\) for 
	integers \(\tin\) less than or equal to \(\blx\)  is \(\bigo{\ln \blx}\), i.e. 
	\begin{align}
	\notag
	&\exists \blx_{0}\in \integers{+}, K\in\reals{+} 
	&
	&\mbox{~such that~}  
	&
	\max\nolimits_{\tin:\tin\leq \blx} \RC{\sfrac{1}{2}}{\Uma{\cset}{(\tin)}}
	&\leq K \ln(\blx)
	&
	&\forall \blx\geq\blx_{0},
	\end{align}
	where \(\Uma{\cset}{(\tin)}:\inpS_{\cset}^{(\tin)}\to\pmea{\outA}\) satisfies
	\(\Uma{\cset}{(\tin)}(\dinp)\!=\!\Wm(\dinp)\) for all \(\dinp\in\inpS_{\cset}^{(\tin)}\)
	and \(\inpS_{\cset}^{(\tin)}\) is defined for any \(\tin\in\integers{+}\) and \(\cset\subset\pdis{\inpS}\) as
	\begin{align}
\notag
	\inpS_{\cset}^{(\tin)}
	&\DEF \{\dinp\in\inpS:\exists\mP\in\cset\mbox{~s.t.~}	
	\mP(\dsta)\tin\in\integers{}~\forall \dsta\in\inpS
	\mbox{~for which~}\mP(\dinp)>0\}.
	\end{align}
\end{assumption}

\begin{theorem}\label{thm:exponent-convex}
	Let \(\{\Wmn{\tin}\}_{\tin\in\integers{+}}\) be a stationary sequence of channels satisfying
	\(\Wmn{\tin}=\Wm\) for all \(\tin\in\integers{+}\), 
	\(\cset\subset\pdis{\inpS}\) be a convex constraint set satisfying Assumption \ref{assumption:convex-ologn}, 
	\(\varepsilon\), \(\rno_{0}\), \(\rno_{1}\) be positive parameters satisfying \(0\!<\!\rno_{0}\!<\!\rno_{1}\!<\!1\).
	Then for any sequence of codes \(\{(\enc^{(\blx)},\dec^{(\blx)})\}_{\blx\in\integers{+}}\) 
	on the product channels \(\{\Wmn{[1,\blx]}\}_{\blx\in\integers{+}}\) 
	satisfying 		
	\begin{align}
	\label{eq:thm:exponent-convex-hypothesis}
	\CRC{\rno_{1}}{\Wm\!}{\cset}
	\geq \tfrac{1}{\blx}\ln \tfrac{M_{\blx}}{L_{\blx}}
	&\geq \CRC{\rno_{0}}{\Wm\!}{\cset}+\tfrac{\ln \blx}{\blx}[\ln (\ln \blx)]^{\varepsilon}
	&
	&\forall \blx\geq\blx_{0} 
	\end{align}
	and \(\composition(\enc^{(\blx)}(\dmes))\in\cset\) for all \(\dmes\in\mesS^{(\blx)}\),
	there exists a \(\tau\in \reals{+}\) and an \(\blx_{1}\geq\blx_{0}\) such that 
	\begin{align}
	\label{eq:thm:exponent-convex}
	\Pem{(\blx)}
	&\geq \blx^{-\tau}  e^{-\blx\spe{\frac{1}{\blx}\ln\frac{M_{\blx}}{L_{\blx}},\Wm\!,\cset}} 
	&
	&\forall \blx\geq\blx_{1}.
	\end{align}
\end{theorem}
It is worth mentioning that the hypotheses of Theorem \ref{thm:exponent-convex} 
are satisfied by the Gaussian and Poisson models considered in 
\cite{wyner88-b,burnashevK99,shannon59,ebert65,ebert66,richters67,lancho19,lanchoODKV19A,lanchoODKV19B};
thus, Theorem \ref{thm:exponent-convex} implies an asymptotic SPB for these channels.
The non-asymptotic counterpart of this implication is presented in \S\ref{sec:outerbound-nonasymptotic-stationary},
between Lemma \ref{lem:spb-convex} and Corollary \ref{lem:spb-stationarycost}.

For codes on the cost constraint memoryless channels one can remove the stationarity hypothesis
using the convex conjugation techniques, \cite[Thm. 36.6]{augustin78}.
Theorem \ref{thm:exponent-cost-ologn} strengthens \cite[Thm. 36.6]{augustin78} 
by removing the bounded cost function hypothesis and 
by establishing the SPB with a prefactor that is polynomial in the block length 
\(\blx\) rather than a prefactor of the from \(e^{-\bigo{\sqrt{\blx}}}\).
\begin{assumption}\label{assumption:cost-ologn}
	The maximum \(\CRC{\sfrac{1}{2}}{\Wmn{\tin}}{\blx \costc}\) for 
	integers \(\tin\) less than or equal to \(\blx\)
	is \(\bigo{\ln \blx}\). 
	\begin{align}
	\notag
	&\exists \blx_{0}\in \integers{+}, K\in\reals{+} 
	&
	&\mbox{~such that~}  
	&
	\max\nolimits_{\tin:\tin\leq \blx} \CRC{\sfrac{1}{2}}{\Wmn{\tin}}{\blx \costc}
	&\leq K \ln(\blx)
	&
	&\forall \blx\geq\blx_{0}.
	\end{align}
\end{assumption}
	
\begin{theorem}\label{thm:exponent-cost-ologn}
	Let \(\{(\Wmn{\tin},\costf_{\tin})\}_{\tin\in\integers{+}}\) be a sequence of channels 
	with associated cost functions satisfying  Assumption \ref{assumption:cost-ologn},
	\(\varepsilon\), \(\rno_{0}\), \(\rno_{1}\) be positive parameters satisfying \(0\!<\!\rno_{0}\!<\!\rno_{1}\!<\!1\).
	Then on the product channels \(\{\Wmn{[1,\blx]}\}_{\blx\in\integers{+}}\) 
	with additive cost functions \(\{\costf_{[1,\blx]}\}_{\blx\in\integers{+}}\) of the form 
	\(\costf_{[1,\blx]}(\dinp_{1}^{\blx})=\sum_{\tin=1}^{\blx} \costf_{\tin}(\dinp_{\tin})\),
	any sequence of codes \(\{(\enc^{(\blx)},\dec^{(\blx)})\}_{\blx\in\integers{+}}\) 
	satisfying
	\begin{align}
	\label{eq:thm:exponent-cost-ologn-hypothesis}
	\CRC{\rno_{1}}{\Wmn{[1,\blx]}}{\blx \costc }
	\geq \ln \tfrac{M_{\blx}}{L_{\blx}}
	&\geq \CRC{\rno_{0}}{\Wmn{[1,\blx]}}{\blx \costc }+\varepsilon (\ln \blx )^2
	&
	&\forall \blx\geq\blx_{0} 
	\end{align}
and	\(\costf_{[1,\blx]}(\enc^{(\blx)}(\dmes))\leq \blx \costc\) 
for all \(\dmes\in\mesS^{(\blx)}\)
for a per channel use cost constraint \(\costc\) satisfying
\(\blx \costc\in \inte{\fcc{\costf_{[1,\blx]}}}\),
there exists a \(\tau\in \reals{+}\) and an \(\blx_{1}\geq\blx_{0}\) such that 
	\begin{align}
	\label{eq:thm:exponent-cost-ologn}
	\Pem{(\blx)}
	&\geq \blx^{-\tau}  e^{-\spe{\ln\frac{M_{\blx}}{L_{\blx}},\Wmn{[1,\blx]},\blx \costc }} 
	&
	&\forall \blx\geq\blx_{1}.
	\end{align}
\end{theorem}

In order to establish the asymptotic SPB with a polynomial prefactor in 
Theorem \ref{thm:exponent-cost-ologn}, we have assumed that the order 
\textonehalf~cost constrained Augustin capacity is scaling no faster than 
logarithmically with the cost. 
Although this hypothesis is true for most cases of interest, including various 
Gaussian and Poisson channels considered in 
\cite{wyner88-b,burnashevK99,shannon59,ebert65,ebert66,richters67,lancho19,lanchoODKV19A,lanchoODKV19B}
and many of their non-stationary variants, there does exist channels violating 
it, see \cite[Example \ref*{C-eg:affine}]{nakiboglu19C}.
Thus one might want to remove the logarithmic growth with 
the cost constraint hypothesis. 
One can do so by using Lemma \ref{lem:spb-cost-conjugate} instead of 
Lemma \ref{lem:spb-cost-ologn} in the proof of Theorem \ref{thm:exponent-cost-ologn}.

Our ultimate aim in this section is to prove the two asymptotic SPBs
given in Theorems \ref{thm:exponent-convex} and \ref{thm:exponent-cost-ologn},
which constitute the main contribution of this article.
To that end, we first provide an impossibility result for the hypothesis testing 
problem with independent samples using Berry-Esseen theorem in 
\S\ref{sec:outerbound-hypothesistesting}. 
Then, we introduce the concepts of averaged Augustin capacity and 
averaged SPE in \S\ref{sec:outerbound-averaging}.
We derive non-asymptotic ---but parametric--- SPBs 
in terms of these averaged quantities
for the composition constrained codes on stationary memoryless channels in 
\S\ref{sec:outerbound-nonasymptotic-stationary} 
and for the cost constrained codes on (possibly non-stationary) memoryless channels in
\S\ref{sec:outerbound-nonasymptotic-cost}.
The derivation of the asymptotic SPBs using the non-asymptotic ones 
and Lemma \ref{lem:avspherepacking} of \S \ref{sec:outerbound-averaging},
which is rather straightforward, is presented in Appendix \ref{sec:omitted-proofs}.

\subsection{An Impossibility Result For Hypothesis Testing}\label{sec:outerbound-hypothesistesting}
\begin{lemma}\label{lem:HT}
For any 
\(\knd\!\geq\!3\), \(\rno\!\in\!(0,1)\), \(\blx\!\in\!\integers{+}\),
and product measures \(\mW,\mQ\!\in\!\pmea{\outA_{1}^{\blx}}\) of the form 
	\(\mW\!=\!\bigotimes\nolimits_{\tin=1}^{\blx} \wmn{\tin}\) 
	and \(\mQ\!=\!\bigotimes\nolimits_{\tin=1}^{\blx} \qmn{\tin}\),
	let  \(\cln{\tin}\DEF\ln \der{\wmn{\tin,ac}}{\qmn{\tin}}-\EXS{\wma{\rno}{\mQ}}{\ln \der{\wmn{\tin,ac}}{\qmn{\tin}}}\)
	where \(\wmn{\tin,ac}\) is the component of \(\wmn{\tin}\) that is absolutely continuous in \(\qmn{\tin}\)
	and 
	\begin{align}
\notag
	\cm{\knd}
	&\DEF\left(\sum\nolimits_{\tin=1}^{\blx} \EXS{\wma{\rno}{\mQ}}{\abs{\cln{\tin}}^{\knd}}\right)^{\sfrac{1}{\knd}}.
	\end{align}
Then any \(\oev\in\outA\) satisfying 
\(\mQ(\oev)\leq \tfrac{1}{4\sqrt{\blx}}e^{-\RD{1}{\wma{\rno}{\mQ}}{\mQ}-\rno 3\cm{\knd}}\)
also satisfies
	\begin{align}
\label{eq:lem:HT-claim}
\mW(\outS_{1}^{\blx}\setminus\oev)
&\geq \tfrac{1}{4\sqrt{\blx}}e^{-\RD{1}{\wma{\rno}{\mQ}}{\mW}-(1-\rno)3\cm{\knd}}.
\end{align}
\end{lemma}
Lemma \ref{lem:HT}  provides an impossibility result for the hypothesis testing problem
with independent samples in the spirit of \cite[Thm. 5]{shannonGB67A}.
It, however, relies on the Berry-Esseen Theorem via \cite[Lemma \ref*{B-lem:berryesseenN}]{nakiboglu19B} 
rather than the Chebyshev inequality.
\begin{proof}[Proof of Lemma \ref{lem:HT}]
	Let \(\oev_{0}=\{\dout:\abs{\sum\nolimits_{\tin=1}^{\blx}\cln{\tin} } \leq 3 \cm{\knd}\}\)
	then 
	\begin{align}
	\notag
	\abs{\ln \der{\wma{\rno}{\mQ}}{\mQ}(\dout_{1}^{\blx})-\RD{1}{\wma{\rno}{\mQ}}{\mQ}} 
	&\leq \rno  3 \cm{\knd}
	&
	&\forall \dout_{1}^{\blx} \in \oev_{0},
	\\
	\notag
	\abs{\ln \der{\wma{\rno}{\mQ}}{\mW}(\dout_{1}^{\blx})-\RD{1}{\wma{\rno}{\mQ}}{\mW}} 
	&\leq (1-\rno)  3 \cm{\knd}
	&
	&\forall \dout_{1}^{\blx} \in \oev_{0}.
	\end{align}
	Hence,
	\begin{align}
	\label{eq:HT-A}
	\wma{\rno}{\mQ}(\oev \cap \oev_{0})
	&\leq \mQ(\oev \cap \oev_{0}) e^{\RD{1}{\wma{\rno}{\mQ}}{\mQ}+\rno 3\cm{\knd}},
	\\
	\label{eq:HT-B}
	\wma{\rno}{\mQ}(\oev \cap \oev_{0})
	&\leq \mW(\oev \cap \oev_{0}) e^{\RD{1}{\wma{\rno}{\mQ}}{\mW}+(1-\rno)3\cm{\knd}}.
\intertext{	On the other hand as a result of the definition of \(\oev_{0}\)
	and \cite[Lemma \ref*{B-lem:berryesseenN}]{nakiboglu19B} we have}
	\label{eq:HT-C}
	\PXS{\wma{\rno}{\mQ}}{\oev_{0}}
	&\geq \tfrac{1}{2\sqrt{\blx}}.
	\end{align}
	\eqref{eq:lem:HT-claim} follows from \eqref{eq:HT-A}, \eqref{eq:HT-B}, and \eqref{eq:HT-C}.
\end{proof}

\subsection{Augustin's Averaging}\label{sec:outerbound-averaging}
Our non-asymptotic SPBs are expressed in terms of the averaged Augustin capacity 
and averaged SPE defined in the following for all \(\epsilon\in(0,1)\) and 
\(\rate\in\reals{\geq0}\).
\begin{align}
\label{eq:def:avcapacity}
\CRCI{\rno}{\Wm\!}{\cset}{\epsilon}
&\DEF \tfrac{1}{\epsilon}\int_{\rno-\epsilon\rno}^{\rno+\epsilon(1-\rno)} 
\left[ 1\vee \left(\tfrac{\rno}{1-\rno}\tfrac{1-\rnt}{\rnt}\right)\right] \CRC{\rnt}{\Wm\!}{\cset}  \dif{\rnt}
\\
\label{eq:def:avspherepacking}
\spa{\epsilon}{\rate,\Wm\!,\cset}
&\DEF \sup\nolimits_{\rno\in (0,1)} \tfrac{1-\rno}{\rno} \left(\CRCI{\rno}{\Wm\!}{\cset}{\epsilon}-\rate\right)
\end{align}

Note that \(\lim\nolimits_{\epsilon \downarrow 0}\CRCI{\rno}{\Wm\!}{\cset}{\epsilon}=\CRC{\rno}{\Wm\!}{\cset}\)
for all \(\rno\in (0,1)\) because \(\CRC{\rno}{\Wm\!}{\cset}\) is continuous in \(\rno\) on \((0,1)\) by 
\cite[Lemma \ref*{C-lem:capacityO}-(\ref*{C-capacityO-continuity})]{nakiboglu19C}.
Furthermore, one can show that the convergence is uniform on compact subsets of \((0,1)\)
using the monotonicity of \(\CRC{\rno}{\Wm\!}{\cset}\) 
and \(\frac{1-\rno}{\rno}\CRC{\rno}{\Wm\!}{\cset}\) in \(\rno\), i.e.
\cite[Lemma \ref*{C-lem:capacityO}-(\ref*{C-capacityO-ilsc},\ref*{C-capacityO-decreasing})]{nakiboglu19C}.
Nevertheless,
\(\CRC{\rno}{\Wm\!}{\cset}\) cannot be approximated by \(\CRCI{\rno}{\Wm\!}{\cset}{\epsilon}\) 
uniformly on \((0,1)\) itself
because \(\lim\nolimits_{\rno \uparrow 1}\CRCI{\rno}{\Wm\!}{\cset}{\epsilon}=\infty\)
for all positive \(\epsilon\)'s whenever \(\CRC{\sfrac{1}{2}}{\Wm\!}{\cset}\) is positive,
even if \(\CRC{1}{\Wm\!}{\cset}\)  is finite. 
The latter observation follows from
the monotonicity of \(\CRC{\rno}{\Wm\!}{\cset}\) 
and \(\frac{1-\rno}{\rno}\CRC{\rno}{\Wm\!}{\cset}\) in \(\rno\):
\begin{align}
\notag
\CRCI{\rno}{\Wm\!}{\cset}{\epsilon}
&\geq \tfrac{1}{\epsilon}\int_{\rno-\epsilon\rno}^{\rno-\frac{\epsilon}{2}\rno} 
\tfrac{\rno}{1-\rno}\tfrac{1-\rnt}{\rnt} \CRC{\rnt}{\Wm\!}{\cset} \dif{\rnt}
\\
\notag
&\geq \tfrac{\rno}{2-\epsilon}(1+\tfrac{\rno \epsilon}{2(1-\rno)})\CRC{\rno-\epsilon\rno}{\Wm\!}{\cset}.
\end{align}

As a function of the rate,
the averaged SPE converges uniformly to 
the SPE on any compact set of rates less than 
the order one Augustin capacity, as demonstrated by the following lemma.

\begin{lemma}\label{lem:avspherepacking}
	For any \(\Wm:\inpS\to \pmea{\outA}\), 
	\(\cset\subset\pdis{\inpS}\) satisfying \(\CRC{\sfrac{1}{2}}{\Wm\!}{\cset}\in\reals{+}\),
	\(\rnf\in (0,1)\),
	\(\rate\in [\CRC{\rnf}{\Wm\!}{\cset},\infty)\),
	and \(\epsilon\in(0,\rnf)\),
	\begin{align}
	\label{eq:lem:avspherepacking}
	0\leq 
	\spa{\epsilon}{\rate,\Wm\!,\cset}-\spe{\rate,\Wm\!,\cset}
	&\leq \tfrac{\epsilon}{1-\epsilon} 
	\tfrac{\rate \vee \spe{\rate,\Wm\!,\cset}}{\rnf}
	\\
	\label{eq:lem:avspherepackingB}
	&\leq \tfrac{\epsilon}{1-\epsilon} \tfrac{\rate}{\rnf^{2}}.
	\end{align}
\end{lemma}
Lemma \ref{lem:avspherepacking} follows from the monotonicity of \(\CRC{\rno}{\Wm\!}{\cset}\) 
and \(\frac{1-\rno}{\rno}\CRC{\rno}{\Wm\!}{\cset}\) in \(\rno\) established in 
\cite[Lemma \ref*{C-lem:capacityO}-(\ref*{C-capacityO-ilsc},\ref*{C-capacityO-decreasing})]{nakiboglu19C};
a proof can be found in Appendix \ref{sec:omitted-proofs}.
The proof of Lemma \ref{lem:avspherepacking} is nearly identical to the proof of 
\cite[Lemma \ref*{B-lem:avspherepacking}]{nakiboglu19B},
establishing the very same result for \(\cset=\pdis{\inpS}\) case.

\subsection[Non-asymptotic SPBs for the Composition Constrained Codes]{Non-asymptotic Sphere Packing Bounds for the Composition Constrained Codes}\label{sec:outerbound-nonasymptotic-stationary}
The composition constrained codes on the stationary memoryless channels 
have been analyzed extensively, but only for the finite input set case. 
It is often assumed that all codewords of the code have precisely 
the same composition. 
If that is not the case, one usually invokes an expurgation based on 
the compositions of the codewords and focus on the most populous 
composition. 
However, such an expurgation leads to a non-trivial result only 
when the input set is finite. 
Bounds derived through the analysis of the most populous composition 
in such cases can be derived using Augustin information and mean. 
The following lemma, on the other hand, use the concepts of 
the Augustin capacity and center together with  Augustin’s averaging 
in order to avoid assuming a finite input set 
or a composition that has a fraction of the codewords that is no smaller 
than a polynomial function of the block length.

\begin{lemma}\label{lem:spb-convex}
Let \(\blx\in\integers{+}\), \(\Wm:\inpS\to\pmea{\outA}\),
\(\cset\) be a convex subset of \(\pdis{\inpS}\),
\(\knd\), \(\rnf\), \(\epsilon\) be positive parameters satisfying
\(\knd\geq3\), \(\rnf<1\), 
\(\epsilon<1\),
and \(\gamma\) be
\begin{align}
\label{eq:lem:spb-convex:gamma}
\gamma
&\DEF 3\sqrt[\knd]{3\blx}
\left([\RC{\sfrac{1}{2}}{\Uma{\cset}{(\blx)}}+(1-\rnf)\ln(1+\blx)]\vee \knd \right).
\end{align}
If \(\Wmn{\tin}=\Wm\) for all \(\tin\leq\blx\) and \(M,L\in\integers{+}\) are such that
\(\ln\tfrac{M}{L}
>\blx \CRCI{\rnf}{\Wm\!}{\cset}{\epsilon}+\tfrac{\gamma}{1-\rnf}
+\ln\tfrac{8 e^{3}\blx^{1.5} }{\epsilon}\),
then any \((M,L)\) channel code on \(\Wmn{[1,\blx]}\)
satisfying \(\composition(\enc(\dmes))\in\cset\) for all 
\(\dmes\in\mesS\) satisfies
	\begin{align}
	\label{eq:lem:spb-convex}
	\Pe
	&\geq \left(\tfrac{\epsilon e^{-2\gamma}}{8e^{3}\blx^{1.5}}\right)^{\sfrac{1}{\rnf}}  
	e^{-\blx\spa{\epsilon}{\frac{1}{\blx}\ln \frac{M}{L},\Wm\!,\cset}}. 
	\end{align}
\end{lemma} 
The absence of a back-off term for the rate in \eqref{eq:lem:spb-convex} can 
be interpreted as an advantage; but this is possible only with a prefactor
that makes the bound trivial, i.e. zero,  as \(\rnf\) converges to zero.
By changing the analysis slightly it is possible to obtain the following
alternative bound:
\begin{align}
\label{eq:lem:spb-convex-alt}
\Pe
&\geq  \tfrac{\epsilon e^{-2\gamma}}{8\blx^{1.5}}
e^{-\blx \spa{\epsilon}{\rate,\Wm\!,\cset}}
&
\rate
&=\tfrac{1}{\blx}\ln \tfrac{M}{L}-\tfrac{1}{\blx}\ln \tfrac{8e^{3}\blx^{1.5}}{\epsilon e^{-2\gamma}}.
\end{align}
When considered together with the identity \(\RC{\rno}{\Uma{\cset(\costc)}{(\blx)}}\leq\CRC{\rno}{\Wm}{\blx \costc }\),
Lemma \ref{lem:spb-convex} and \eqref{eq:lem:spb-convex-alt} imply 
Corollary \ref{lem:spb-stationarycost} and \eqref{eq:lem:spb-stationarycost-alt} given 
in the following.
\begin{corollary}\label{lem:spb-stationarycost} 
	Let \(\blx,\ell\in\integers{+}\), \(\Wm:\inpS\to\pmea{\outA}\), \(\costf:\inpS\to\reals{\geq0}^{\ell}\),
	\(\costc\in \fcc{\costf}\), \(\knd\), \(\rnf\), \(\epsilon\)
	be positive parameters satisfying \(\knd\geq3\), \(\rnf<1\), 
	\(\epsilon<1\), and \(\gamma\) be
	\begin{align}
\notag
	\gamma
	&\DEF 3\sqrt[\knd]{3\blx}
	\left([\CRC{\sfrac{1}{2}}{\Wm\!}{\blx \costc }+(1-\rnf)\ln(1+\blx)]\vee \knd \right).
	\end{align}
	If \(\Wmn{\tin}=\Wm\) for all \(\tin\leq\blx\) and \(M,L\in\integers{+}\) are such that	
	\(\ln\tfrac{M}{L}
	>\blx \CRCI{\rnf}{\Wm\!}{\costc}{\epsilon}+\tfrac{\gamma}{1-\rnf}
	+\ln\tfrac{8e^{3}\blx^{1.5} }{\epsilon}\),
	then any \((M,L)\) channel code on \(\Wmn{[1,\blx]}\)
	satisfying \(\sum_{\tin=1}^{\blx}\costf(\enc_{\tin}(\dmes))\leq \blx \costc \)
	for all \(\dmes\in\mesS\)	
	satisfies 
	\begin{align}
\notag
	\Pe
	&\geq \left(\tfrac{\epsilon e^{-2\gamma}}{8e^{3}\blx^{1.5}}\right)^{\sfrac{1}{\rnf}}  
	e^{-\blx\spa{\epsilon}{\frac{1}{\blx}\ln\frac{M}{L},\Wm\!,\costc}}.
	\end{align}	
\end{corollary} 
By changing the the analysis slightly it is possible to obtain the following alternative bound:
\begin{align}
\label{eq:lem:spb-stationarycost-alt}
\Pe
&\geq  \tfrac{\epsilon e^{-2\gamma}}{8\blx^{1.5}}
e^{-\blx\spa{\epsilon}{\rate,\Wm\!,\costc}}
&
\rate
&=\tfrac{1}{\blx}\ln \tfrac{M}{L}
-\tfrac{1}{\blx}\ln \tfrac{8e^{3} \blx^{1.5}}{\epsilon e^{-2\gamma}}.
\end{align}

\begin{proof}[Proof of Lemma \ref{lem:spb-convex} and \eqref{eq:lem:spb-convex-alt}]
The constrained Augustin center \(\qmn{\rno,\Wm\!,\cset}\) is continuous in \(\rno\)
on \((0,1)\) for the total variation topology on \(\pmea{\outA}\) by 
\cite[Lemmas \ref*{C-lem:capacityO}-(\ref*{C-capacityO-continuity}) and \ref*{C-lem:centercontinuity}]{nakiboglu19C}.
Thus \(\qmn{\cdot,\Wm\!,\cset}\) is a transition probability from  \(((0,1),\rborel{(0,1)})\) 
to \((\outS,\outA)\). 
For each \(\tin\leq\blx\), we define the averaged center \(\qma{\rno,\tin}{\epsilon}\) 
as the \(\outS\) marginal of the probability measure  \(\umn{\rno,\epsilon}\mtimes\qmn{\cdot,\Wm\!,\cset}\) 
where \(\umn{\rno,\epsilon}\) is the uniform probability distribution on  
\((\rno-\epsilon\rno,\rno+\epsilon(1-\rno))\)
\begin{align}
\label{eq:spb-convex-1}
\qma{\rno,\tin}{\epsilon}
&\DEF\tfrac{1}{\epsilon} \int_{\rno-\rno\epsilon}^{\rno+(1-\rno)\epsilon} \qmn{\rnt,\Wm\!,\cset} \dif{\rnt}
&
&\forall \tin\in\{1,\ldots,\blx\}.
\end{align}
Let \(\qmn{\rno,\tin}\in \pmea{\outA_{\tin}}\) and  \(\qmn{\rno}\in \pmea{\outA_{1}^{\blx}}\) 
be
\begin{align}
\notag
\qmn{\rno,\tin}
&\DEF \tfrac{\blx}{\blx+1}\qma{\rno,\tin}{\epsilon} 
+\tfrac{1}{\blx+1}\qmn{\sfrac{1}{2},\Uma{\cset}{(\blx)}}
&
&\forall \tin\in\{1,\ldots,\blx\},
\\
\notag
\qmn{\rno}
&\DEF\bigotimes\nolimits_{\tin=1}^{\blx} \qmn{\rno,\tin}.
\end{align}
Let us denote the probability measure generated by \(\enc(\dmes)\), i.e. \(\bigotimes_{\tin=1}^{\blx}\Wm(\enc_{\tin}(\dmes))\),
by \(\wma{}{\dmes}\) for brevity.
	\begin{align}
	\notag
	\RD{\rno}{\wma{}{\dmes}}{\qmn{\rno}}
	&=\sum\nolimits_{\tin=1}^{\blx}  \RD{\rno}{\Wm(\enc_{\tin}(\dmes))}{\qmn{\rno,\tin}}
	&
	&\mbox{by \cite[Thm. 28]{ervenH14}},
	\\
	\notag
	&\leq \sum\nolimits_{\tin=1}^{\blx} \left(\ln\tfrac{\blx+1}{\blx}+ 
	\RD{\rno}{\Wm(\enc_{\tin}(\dmes))}{\qma{\rno,\tin}{\epsilon}} \right)	  
	&
	&\mbox{by \cite[Lemma \ref*{C-lem:divergence-RM}]{nakiboglu19C} because
	\(\tfrac{\blx}{\blx+1}\qma{\rno,\tin}{\epsilon}\leq \qmn{\rno,\tin}\)},
	\\
	\notag
	&\leq 1+\sum\nolimits_{\tin=1}^{\blx}  \RD{\rno}{\Wm(\enc_{\tin}(\dmes))}{\qma{\rno,\tin}{\epsilon}}
	&
	&\mbox{by \(\ln \tau\leq \tau-1\)},
	\\
	\notag
	&\leq 1+\sum\nolimits_{\tin=1}^{\blx} 	
	\tfrac{1}{\epsilon} \int_{\rno(1-\epsilon)}^{\rno+(1-\rno)\epsilon} \RD{\rno}{\Wm(\enc_{\tin}(\dmes))}{\qmn{\rnt,\Wm,\cset}}\dif{\rnt}
	&
	&\mbox{by the Jensen's inequality and \cite[Thm. 12]{ervenH14}},
	\\
	\notag
	&= 1
	+\tfrac{\blx}{\epsilon} \int_{\rno(1-\epsilon)}^{\rno+(1-\rno)\epsilon} 
	\CRD{\rno}{\Wm}{\qmn{\rnt,\Wm\!,\cset}}{\composition(\enc(\dmes))}\dif{\rnt}
	&
	&\mbox{by Definition \ref{def:conditionaldivergence} 
		and the definition of \(\composition\),}
	\\
	\notag
	&\leq 1
	+ \tfrac{\blx}{\epsilon} \int_{\rno(1-\epsilon)}^{\rno+(1-\rno)\epsilon} 
	\left(1 \vee \tfrac{1-\rnt}{\rnt}\tfrac{\rno}{1-\rno}\right)
	\CRD{\rnt}{\Wm}{\qmn{\rnt,\Wm\!,\cset}}{\composition(\enc(\dmes))}
	&
	&\mbox{by \cite[Thm. 3 and Proposition 2]{ervenH14},}
	\\
	\notag
	&\leq 1
	+ \tfrac{\blx}{\epsilon} \int_{\rno(1-\epsilon)}^{\rno+(1-\rno)\epsilon}\!
	\!\left(1 \vee \tfrac{1-\rnt}{\rnt}\tfrac{\rno}{1-\rno}\right)\!\CRC{\rnt}{\Wm\!}{\cset}\dif{\rnt}
	&
	&\mbox{by \cite[Thm. \ref{C-thm:minimax}]{nakiboglu19C} because \(\composition(\enc(\dmes))\in\cset\),}
	\\
	\label{eq:spb-convex-2}
	&= 1
	+ \blx\CRCI{\rno}{\Wm\!}{\cset}{\epsilon}
	&
	&\mbox{by the definition of \(\CRCI{\rno}{\Wm\!}{\cset}{\epsilon}\).}
	\end{align}
Let the probability measure \(\vma{\rno}{\dmes}\) be the order \(\rno\)
tilted probability measure between \(\wma{}{\dmes}\) and \(\qmn{\rno}\)
defined in \eqref{eq:def:tiltedprobabilitymeasure}.
Then \(\vma{\rno}{\dmes}\) is of the form
\(\vma{\rno}{\dmes}=\bigotimes\nolimits_{\tin=1}^{\blx} \Wma{\rno}{\qmn{\rno,\tin}}(\enc_{\tin}(\dmes))\)
as a result of product structure of \(\wma{}{\dmes}\) and \(\qmn{\rno}\).
Let the random variables \(\cla{\rno,\tin}{\dmes}\) and \(\cla{\rno}{\dmes}\) be
\begin{align}
\notag
\cla{\rno,\tin}{\dmes}
&\DEF \ln \der{[\Wm(\enc_{\tin}(\dmes))]_{ac}}{\qmn{\rno,\tin}}-\EXS{\vma{\rno}{\dmes}}{\ln \der{[\Wm(\enc_{\tin}(\dmes))]_{ac}}{\qmn{\rno,\tin}}},
\\
\notag
\cla{\rno}{\dmes}
&\DEF \sum\nolimits_{\tin=1}^{\blx} \cla{\rno,\tin}{\dmes},
\end{align}
	where \([\Wm(\enc_{\tin}(\dmes))]_{ac}\) is the component of \(\Wm(\enc_{\tin}(\dmes))\) that is absolutely continuous in \(\qmn{\rno,\tin}\). 
	Then for all \(\knd\in\reals{+}\) and \(\rno\in(0,1)\)
	\begin{align}
	\notag
	\EXS{\vma{\rno}{\dmes}}{\abs{\cla{\rno,\tin}{\dmes}}^{\knd}}^{\sfrac{1}{\knd}}
	&\leq 3^{\sfrac{1}{\knd}} \tfrac{\left[(1-\rno)\RD{\rno}{\Wm(\enc_{\tin}(\dmes))}{\qmn{\rno,\tin}}\right]\vee \knd}{\rno (1-\rno)}
	&&&
	&\mbox{by \cite[Lemma \ref*{B-lem:MomentBound}]{nakiboglu19B},}
	\\
	\notag
	&\leq 3^{\sfrac{1}{\knd}} \tfrac{\left[(1-\rno)\RD{\rno}{\Wm(\enc_{\tin}(\dmes))}{\qmn{\sfrac{1}{2},\Uma{\cset}{(\blx)}}}
		+(1-\rno)\ln (1+\blx)\right]\vee \knd}{\rno (1-\rno)}
	&&&
	&\mbox{by \cite[Lemma \ref*{C-lem:divergence-RM}]{nakiboglu19C}
	because \(\tfrac{\qmn{\sfrac{1}{2},\Uma{\cset}{(\blx)}}}{\blx+1}\leq\qmn{\rno,\tin}\),}
	\\
	\notag
	&\leq 3^{\sfrac{1}{\knd}} \tfrac{\left[\RD{\sfrac{1}{2}}{\Wm(\enc_{\tin}(\dmes))}{\qmn{\sfrac{1}{2},\Uma{\cset}{(\blx)}}}
		+(1-\rno)\ln (1+\blx)\right]\vee \knd}{\rno (1-\rno)}
	&&&
	&\mbox{by \cite[Thm. 3 and Proposition 2]{ervenH14},}
	\\
	\notag
	&\leq 3^{\sfrac{1}{\knd}} \tfrac{\left[\RC{\sfrac{1}{2}}{\Uma{\cset}{(\blx)}}+(1-\rno)\ln (1+\blx)\right]\vee \knd}{\rno (1-\rno)}
	&&&
	&\mbox{by \cite[Thm. \ref*{C-thm:minimax}]{nakiboglu19C} because \(\enc_{\tin}(\dmes)\in\inpS_{\cset}^{(\blx)}\).}
	\end{align}
	Then using the definition of \(\gamma\) given in \eqref{eq:lem:spb-convex:gamma}, we get
	\begin{align}
	\label{eq:spb-convex-3}
	\left[\sum\nolimits_{\tin=1}^{\blx} \EXS{\vma{\rno}{\dmes}}{\abs{\cla{\rno,\tin}{\dmes}}^{\knd}}\right]^{\sfrac{1}{\knd}}
	&\leq \tfrac{\gamma}{3\rno(1-\rno)}
	&
	&\forall \rno\in[\rnf,1).
	\end{align}
	On the other hand, \cite[Thm. 30]{ervenH14} implies
	\begin{align}
	\label{eq:spb-convex-4}
	\RD{1}{\vma{\rno}{\dmes}}{\qmn{\rno}}
	&=\RD{\rno}{\wma{}{\dmes}}{\qmn{\rno}}-\tfrac{\rno}{1-\rno}\RD{1}{\vma{\rno}{\dmes}}{\wma{}{\dmes}}
	&
	&\forall \dmes \in \mesS,\rno\in (0,1). 
	\end{align}
	Thus we can bound \(\RD{1}{\vma{\rno}{\dmes}}{\qmn{\rno}}\) using the non-negativity of the \renyi divergence, i.e. 
	\cite[Thm. 8]{ervenH14},
	and \eqref{eq:spb-convex-2}
	\begin{align}
	\notag
	0\leq 
	\RD{1}{\vma{\rno}{\dmes}}{\qmn{\rno}}
	&\leq 1+\blx \CRCI{\rno}{\Wm\!}{\cset}{\epsilon}.
	\end{align}
	Hence,
	\begin{align}
	\notag
	\lim\nolimits_{\rno \downarrow \rnf} 
	\RD{1}{\vma{\rno}{\dmes}}{\qmn{\rno}}+\tfrac{\gamma}{1-\rno}
	&<\ln \tfrac{M}{L}\tfrac{\epsilon }{8 e^{2}\blx^{1.5}},
	\\
	\notag
	\lim\nolimits_{\rno \uparrow 1} 
	\RD{1}{\vma{\rno}{\dmes}}{\qmn{\rno}}+\tfrac{\gamma}{1-\rno}
	&=\infty.
	\end{align}
\(\RD{1}{\vma{\rno}{\dmes}}{\qmn{\rno}}\) is continuous in \(\rno\) by
\cite[Lemma \ref*{B-lem:tilting}]{nakiboglu19B}
because 
\(\qmn{\rno}\) is continuous in \(\rno\) for the total variation 
topology\footnote{In particular \(\lon{\qmn{\rno}-\qmn{\rnt}}\leq\sqrt{8\ln \tfrac{\epsilon}{\epsilon-(1-\epsilon)\abs{\rno-\rnt}}}\)
	because
	\(\lon{\qmn{\rno}-\qmn{\rnt}}\leq\sqrt{4\RD{1/2}{\qmn{\rno}}{\qmn{\rnt}}}\)
	by \cite[Thm. 31]{ervenH14}, \(\RD{1/2}{\qmn{\rno}}{\qmn{\rnt}}=\blx\RD{1/2}{\qmn{\rno,\tin}}{\qmn{\rnt,\tin}}\)
	by \cite[Thm. 28]{ervenH14} and the definition of \(\qmn{\rno\!}\)'s, 
	\(\RD{1/2}{\qmn{\rno,\tin}}{\qmn{\rnt,\tin}}\leq2\ln\tfrac{2}{2-\lon{\qmn{\rno,\tin}-\qmn{\rnt,\tin}}}\) 
	by \cite[(\ref*{C-eq:shiryaev})]{nakiboglu19C},
	and
	\(\lon{\qmn{\rno,\tin}-\qmn{\rnt,\tin}}\leq 2\tfrac{1-\epsilon}{\epsilon}\abs{\rnt-\rno}\)
	by the definition of \(\qmn{\rno\!}\)'s.} 
on \(\pmea{\outA}\).
	Then as a result of the intermediate value theorem \cite[4.23]{rudin},
	for each \(\dmes\in\mesS\) there exists an \(\rno_{\dmes}\) in \((\rnf,1)\)
	satisfying
	\begin{align}
	\label{eq:spb-convex-5}
	\left.\left(\RD{1}{\vma{\rno}{\dmes}}{\qmn{\rno}}+\tfrac{\gamma}{1-\rno}\right)\right\vert_{\rno=\rno_{\dmes}} 
	&=\ln \tfrac{M}{L}\tfrac{\epsilon}{8 e^{2} \blx^{1.5}}.
	\end{align}
	For any \(K\in\integers{+}\), 
	there exists a \(\sfrac{1}{K}\) long closed subinterval of \((0,1)\) 
	with  \(\lceil\tfrac{M}{K}\rceil\) or more of these \(\rno_{\dmes}\)'s.
	Let \([\rnt,\rnt+\sfrac{1}{K}]\) be this interval and 
\(\widetilde{\mQ}\) and \(\widetilde{\qmn{\tin}}\) be
\begin{align}
\notag
\widetilde{\mQ}
&\DEF\bigotimes\nolimits_{\tin=1}^{\blx} \widetilde{\qmn{\tin}}
&
\widetilde{\qmn{\tin}}
&=\tfrac{\blx}{\blx+1} \qma{\widetilde{\rno},\tin}{\widetilde{\epsilon}}+\tfrac{1}{\blx+1}\qmn{\sfrac{1}{2},\Uma{\cset}{(\blx)}},
\intertext{where \(\widetilde{\epsilon}\) and \(\widetilde{\rno}\) are given by}
	\notag
	\widetilde{\epsilon}
	&=\tfrac{1}{K}+\epsilon(1-\tfrac{1}{K})
	&
	\widetilde{\rno}
	&=\tfrac{1-\epsilon}{1-\widetilde{\epsilon}}\rnt.
\intertext{Then for all \(\rno\) in \([\rnt,\rnt+\tfrac{1}{K}]\) 
			using the definition of the averaged center \(\qma{\rno,\tin}{\epsilon}\) 
			given in \eqref{eq:spb-convex-1}  we get}
	\notag
	\qmn{\rno,\tin}
	&\leq
	\tfrac{\widetilde{\epsilon}}{\epsilon}
	\widetilde{\qmn{\tin}}
	&
	\qmn{\rno}
	&\leq
	(\tfrac{\widetilde{\epsilon}}{\epsilon})^{\blx}\widetilde{\mQ}.
	\end{align}
At least half of the messages with \(\rno_{\dmes}\)'s in \([\rnt,\rnt+\sfrac{1}{K}]\), i.e. 
at least \(\lceil\tfrac{1}{2}\lceil\tfrac{M}{K}\rceil\rceil\) messages,  satisfy
\(\widetilde{\mQ}(\oev_{\dmes})\leq 2 \tfrac{L}{\lceil \sfrac{M}{K} \rceil}\)
as a result of Markov's inequality because
\(\sum_{\dmes\in\widetilde{\mesS}} \widetilde{\mQ}(\oev_{\dmes})\leq L\)
by the definition of list decoding, where
\(\widetilde{\mesS}=\{\dmes:\rno_{\dmes}\in [\rnt,\rnt+\sfrac{1}{K}]\}\).
Then at least \(\lceil\tfrac{1}{2}\lceil\tfrac{M}{K}\rceil\rceil\) messages
with \(\rno_{\dmes}\)'s in \([\rnt,\rnt+\sfrac{1}{K}]\) satisfy
	\begin{align} 
	\notag
	\qmn{\rno_{\dmes}}(\oev_{\dmes})
	&\leq \tfrac{2LK}{M}\left(1+\tfrac{1}{K}\tfrac{1-\epsilon}{\epsilon}\right)^{\blx}
	\\
	\notag
	&\leq \tfrac{2LK}{M}\left(1+\tfrac{1}{K\epsilon}\right)^{\blx}.
	\end{align}
	Note that \(\tfrac{\blx}{\epsilon}>1\) because we have assumed that
	\(\epsilon<1\) and \(\blx\geq1\).
	If we set \(K\) to be \(K=\lfloor\tfrac{\blx}{\epsilon}\rfloor\) and use the identity 
	\((1+\tau)^{\sfrac{1}{\tau}}<e\) together with  \eqref{eq:spb-convex-5}  we get
	\begin{align}
	\notag
	\qmn{\rno_{\dmes}}(\oev_{\dmes})
	&\leq \tfrac{2LK}{M}\left(1+\tfrac{1}{K \epsilon}\right)^{Ke\frac{\blx}{K e}}
	\\
	\notag
	&\leq \tfrac{2L}{M}\tfrac{\blx}{\epsilon} e^{2}
	\\
	\notag
	&\leq \tfrac{1}{4\sqrt{\blx}} e^{-\RD{1}{\vma{\rno_{\dmes}}{\dmes}}{\qmn{\rno_{\dmes}}}-\frac{\gamma}{1-\rno_{\dmes}}}.
	\end{align}
	Then \eqref{eq:spb-convex-3} and Lemma \ref{lem:HT} imply
	\begin{align}
	\label{eq:spb-convex-6}
	\Pem{\dmes}
	&\geq \tfrac{1}{4\sqrt{\blx}}e^{-\RD{1}{\vma{\rno_{\dmes}}{\dmes}}{\wma{}{\dmes}}-\frac{\gamma}{\rno_{\dmes}}}.
	\end{align}
	Using \eqref{eq:spb-convex-2}, \eqref{eq:spb-convex-4}, \eqref{eq:spb-convex-5}, and \eqref{eq:spb-convex-6} we get
	\begin{align}
	\label{eq:spb-convex-7}
	\Pem{\dmes}
	&\geq \tfrac{1}{4\sqrt{\blx}}
	\left(\tfrac{\epsilon }{8e^{2}\blx^{1.5}}\tfrac{M}{L}\right)^{\frac{1-\rno_{\dmes}}{\rno_{\dmes}}}
	e^{-\frac{1-\rno_{\dmes}}{\rno_{\dmes}}(1+\blx\RCI{\rno_{\dmes}}{\Wm\!}{\epsilon})
	 -\frac{2\gamma}{\rno_{\dmes}}}.
	\end{align}
	Hence, for all  \(\dmes\) satisfying \eqref{eq:spb-convex-5} as a result of 
	the definition of 
	\(\spa{\epsilon}{\rate,\Wm\!,\cset}\) given in \eqref{eq:def:avspherepacking}  we have
	\begin{align}
	\notag
	\Pem{\dmes}
	&\geq \tfrac{e^{\sfrac{-2\gamma}{\rnf}}}{4\sqrt{\blx}}
	\left(\tfrac{\epsilon }{8e^{3}\blx^{1.5}}\right)^{\frac{1-\rnf}{\rnf}}
	e^{-\blx\spa{\epsilon}{\rate,\Wm\!,\cset}}
	&
	\rate
	&=\tfrac{1}{\blx}\ln \tfrac{M}{L}.
	\end{align}
	Since there are at least \(\lceil\tfrac{1}{2}\lceil\tfrac{M}{K}\rceil\rceil\) such messages 
	and \(\lceil\tfrac{1}{2}\lceil\tfrac{M}{K}\rceil\rceil\geq \tfrac{M \epsilon}{2\blx}\)
	by construction, 
	we get the lower bound given in \eqref{eq:lem:spb-convex}.
	
	Note that \(\tfrac{\epsilon e^{-2\gamma}}{8e^{3}\blx^{1.5}}<1\)
	because \(\epsilon\in(0,1)\), \(\blx\geq1\), and \(\gamma\geq0\).
	Thus the prefactor in \eqref{eq:lem:spb-convex} converges to zero as \(\rnf\) goes to zero.
	In order avoid this phenomena, one can change the analysis after \eqref{eq:spb-convex-7} 
	and introduce an approximation error term to the rate of the averaged SPE term:
	For all for all \(\dmes\) satisfying \eqref{eq:spb-convex-5} as a result of the definition of \(\spa{\epsilon}{\rate,\Wm\!,\costc}\) given 
	in \eqref{eq:def:avspherepacking} and \eqref{eq:spb-convex-7} we have
	\begin{align}
	\notag
	\Pem{\dmes}
	&\geq \tfrac{e^{-2\gamma}}{4\sqrt{\blx}} e^{-\blx\spa{\epsilon}{\rate,\Wm}}  
&	\rate
	&=\tfrac{1}{\blx}\ln\tfrac{M}{L}-\tfrac{2\gamma}{\blx}
	-\tfrac{1}{\blx}\ln \tfrac{8 e^{3}\blx^{1.5}}{\epsilon}.
	\end{align}
	Since there are at least \(\lceil\tfrac{1}{2}\lceil\tfrac{M}{K}\rceil\rceil\) such messages 
and \(\lceil\tfrac{1}{2}\lceil\tfrac{M}{K}\rceil\rceil\geq \tfrac{M \epsilon}{2\blx}\) by construction, 
we get the lower bound given in \eqref{eq:lem:spb-convex-alt}.
\end{proof}

\subsection[Non-asymptotic SPBs for the Cost Constrained Codes]{Non-asymptotic Sphere Packing Bounds for the Cost Constrained Codes}\label{sec:outerbound-nonasymptotic-cost}
In \S \ref{sec:outerbound-nonasymptotic-stationary},
we were primarily interested in the stationary memoryless channels with convex composition 
constraints; 
the cost constrained memoryless channels, 
considered in Corollary \ref{lem:spb-stationarycost},
were merely an afterthought.
For the cost constrained memoryless channels it is possible to establish the SPB
even for the non-stationary channels using the convex-conjugation techniques,
as we will demonstrate in the following.
The only drawback of the use the convex-conjugation techniques is that we will be able to establish 
the SPB for the cost constraints in \(\inte{\fcc{\costf}}\), rather than the cost constraints in 
\(\fcc{\costf}\).


\begin{lemma}\label{lem:spb-cost-ologn}
Let \(\blx\in\integers{+}\), \(\Wmn{[1,\blx]}:\inpS_{1}^{\blx}\to\pmea{\outA_{1}^{\blx}}\) 
be a length \(\blx\)  product channel with an additive cost function 
\(\costf_{[1,\blx]}:\inpS_{1}^{\blx}\to\reals{\geq0}^{\ell}\) 
satisfying 
\(\costf_{[1,\blx]}(\dinp_{1}^{\blx})=\sum_{\tin=1}^{\blx} \costf_{\tin}(\dinp_{\tin})\) 
for some  \(\costf_{\tin}:\inpS_{\tin}\to\reals{\geq0}^{\ell}\),
\(\costc\) be a per channel use cost constraint satisfying 
\(\blx\costc\in \inte{\fcc{\costf_{[1,\blx]}}}\),
 \(\knd\), \(\rnf\), \(\epsilon\) 
be arbitrary positive parameters satisfying
\(\knd\geq3\), \(\rnf<1\), \(\epsilon<1\),
and \(\gamma\) be
\begin{align}
\label{eq:lem:spb-cost-ologn:gamma}
\gamma
&\DEF 3\left[3\sum\nolimits_{\tin=1}^{\blx} 
	\left([\CRC{\sfrac{1}{2}}{\Wmn{\tin}}{\blx\costc}+(1-\rnf)\ln (1+\blx)]\vee \knd \right)^{\knd}\right]^{\sfrac{1}{\knd}}.
\end{align}
If \(M\) and \(L\) are integers such that
\(\ln\tfrac{M}{L}>\CRCI{\rnf}{\Wmn{[1,\blx]}}{\blx\costc}{\epsilon}+\tfrac{\gamma}{1-\rnf}
+\ln\tfrac{8e^{3}\blx^{1.5} }{\epsilon}\),
then any \((M,L)\) channel code on \(\Wmn{[1,\blx]}\)
satisfying \(\bigvee_{\dmes=1}^{M} \costf_{[1,\blx]}(\enc(\dmes))\leq\blx\costc\) satisfies
\begin{align}
\label{eq:lem:spb-cost-ologn}
\Pe
&\geq \left(\tfrac{\epsilon e^{-2\gamma}}{8e^{3}\blx^{1.5}}\right)^{\sfrac{1}{\rnf}}  
e^{-\spa{\epsilon}{\ln\frac{M}{L},\Wmn{[1,\blx]},\blx\costc}}.
\end{align}
\end{lemma}
The bound given in \eqref{eq:lem:spb-cost-ologn} does not have a back-off term
for the rate in the averaged SPE but it becomes trivial, i.e. zero, as \(\rnf\) converges to zero.
By changing the analysis slightly it is possible to obtain the following alternative  bound:
\begin{align}
\label{eq:lem:spb-cost-ologn-alt}
\Pe
&\geq  \tfrac{\epsilon e^{-2\gamma}}{8\blx^{1.5}}
e^{-\spa{\epsilon}{\rate,\Wmn{[1,\blx]},\blx\costc}}
&
\rate
&=\ln \tfrac{M}{L}-\ln \tfrac{8e^{3} \blx^{1.5}}{\epsilon e^{-2\gamma}}.
\end{align}
The bounds in Lemma \ref{lem:spb-cost-ologn} and \eqref{eq:lem:spb-cost-ologn-alt}
are good enough for proving the asymptotic SPB given in Theorem \ref{thm:exponent-cost-ologn}.
However, the \(\gamma\) term can be improved by a more careful calculation.
The resulting expression, however, includes an additional optimization.
\begin{lemma}\label{lem:spb-cost-conjugate}
Lemma \ref{lem:spb-cost-ologn} and \eqref{eq:lem:spb-cost-ologn-alt}
are valid as they are if the definition of \(\gamma\) given 
in \eqref{eq:lem:spb-cost-ologn:gamma} is replaced by the one given in
\eqref{eq:lem:spb-cost-conjugate:gamma}. 
	\begin{align}
	\label{eq:lem:spb-cost-conjugate:gamma}
	\gamma
	&\DEF \max\nolimits_{\{\costc_{\tin}\}:\costc_{\tin}\geq0 \forall\tin, \sum_{\tin=1}^{\blx}\costc_{\tin}\leq\blx \costc} 3
	\left[	3\sum\nolimits_{\tin=1}^{\blx} 
	\left([\RCL{\sfrac{1}{2}}{\Wmn{\tin}}{\lgm}+\lgm\cdot\costc_{\tin}
	+(1-\rnf)\ln(1+\blx)]\vee \knd \right)^{\knd}
	\right]^{\sfrac{1}{\knd}}
	\end{align}
where \(\lgm=\lgm_{\sfrac{1}{2},\Wmn{[1,\blx]},\blx\costc}\).
If \(\Wmn{\tin}=\Wm\) and \(\costf_{\tin}=\costf\) for all \(\tin\in[1,\blx]\)
then \(\gamma\) defined in \eqref{eq:lem:spb-cost-conjugate:gamma} satisfies
\begin{align}
\label{eq:lem:spb-cost-conjugate-stationary:gamma}
\gamma
&=\max\nolimits_{\{\costc_{\tin}\}:\costc_{\tin}\geq0 \forall\tin, \sum_{\tin=1}^{\blx}\costc_{\tin}\leq \blx \costc} 3
\left[3\sum\nolimits_{\tin=1}^{\blx} 
\left(\left[\CRC{\sfrac{1}{2}}{\Wm}{\costc}+\lgm\cdot(\costc_{\tin}-\costc)
+(1-\rnf)\ln(1+\blx)\right]\vee \knd \right)^{\knd}
\right]^{\sfrac{1}{\knd}}.
\end{align}
\end{lemma} 

\begin{proof}[Proof of Lemma \ref{lem:spb-cost-ologn} and \eqref{eq:lem:spb-cost-ologn-alt}]
For each \(\rno\!\in\!(0,1)\) there exists a \(\lgm_{\rno}\!\in\!\reals{\geq0}^{\ell}\) satisfying 
\begin{align}
\label{eq:spb-cost-ologn-0}
\CRC{\rno}{\Wmn{[1,\blx]}}{\blx\costc}
&=\RCL{\rno}{\Wmn{[1,\blx]}\!}{\lgm_{\rno}}\!+\!\lgm_{\rno}\!\cdot\!\costc\blx
\end{align}
by \cite[Lemma \ref*{C-lem:Lcapacity}-(\ref*{C-Lcapacity:interior})]{nakiboglu19C} 
because \(\blx\costc\in\inte{\fcc{\costf_{[1,\blx]}}}\).
Hence, \(\qmn{\rno,\Wmn{[1,\blx]},\blx\costc}=\qma{\rno,\Wmn{[1,\blx]}}{\lgm_{\rno}}\) by 
\cite[Lemma \ref*{C-lem:Lcenter}]{nakiboglu19C}.
Furthermore, \(\qma{\rno,\Wmn{[1,\blx]}}{\lgm_{\rno}}\) is of the form 
\(\qma{\rno,\Wmn{[1,\blx]}}{\lgm_{\rno}}=\bigotimes\nolimits_{\tin=1}^{\blx} \qma{\rno,\Wmn{\tin}}{\lgm_{\rno}}\)
by \cite[Lemma \ref*{C-lem:Lcapacityproduct}]{nakiboglu19C}.
Then
\begin{align}
\notag
\lon{\qma{\rno,\Wmn{\tin}}{\lgm_{\rno}}-\qma{\rnf,\Wmn{\tin}}{\lgm_{\rnf}}}^{2}
&\leq 4\RD{\sfrac{1}{2}}{\qma{\rno,\Wmn{\tin}}{\lgm_{\rno}}}{\qma{\rnf,\Wmn{\tin}}{\lgm_{\rnf}}}
&
&\mbox{by \cite[Thm. 31]{ervenH14},}
\\
\notag
&\leq 4\RD{\sfrac{1}{2}}{\qma{\rno,\Wmn{[1,\blx]}}{\lgm_{\rno}}}{\qma{\rnf,\Wmn{[1,\blx]}}{\lgm_{\rnf}}}
&
&\mbox{by \cite[Thms. 8 and 28]{ervenH14},}
\\
\notag
&=4\RD{\sfrac{1}{2}}{\qmn{\rno,\Wmn{[1,\blx]},\blx\costc}}{\qmn{\rnf,\Wmn{[1,\blx]},\blx\costc}}
&
&\mbox{because \(\qmn{\rno,\Wmn{[1,\blx]},\blx\costc}\!=\!\qma{\rno,\Wmn{[1,\blx]}}{\lgm_{\rno}}\) \(\forall\rno\in(0,1)\),}
\\
\notag
&\leq 8\ln \tfrac{2}{2-\lon{\qmn{\rno,\Wmn{[1,\blx]},\costc}-\qmn{\rnf,\Wmn{[1,\blx]},\costc}}}
&
&\mbox{by \cite[(\ref*{C-eq:shiryaev})]{nakiboglu19C}.}
\end{align}
On the other hand, the constrained Augustin center \(\qmn{\rno,\Wmn{[1,\blx]},\blx\costc}\) is continuous in \(\rno\)
on \((0,1)\) for the total variation topology on \(\pmea{\outA_{1}^{\blx}}\) by 
\cite[Lemmas \ref*{C-lem:capacityO}-(\ref*{C-capacityO-continuity}) and \ref*{C-lem:centercontinuity}]{nakiboglu19C}.
Thus \(\qma{\rno,\Wmn{\tin}}{\lgm_{\rno}}\) is a continuous function of \(\rno\) for 
the total variation topology on \(\pmea{\outA_{\tin}}\), 
as well.
Then \(\qma{\cdot,\Wmn{\tin}}{\lgm_{\cdot}}\) is a transition probability from 
\(((0,1),\rborel{(0,1)})\) to \((\outS_{\tin},\outA_{\tin})\). 
We define \(\qma{\rno,\tin}{\epsilon}\) as the \(\outS_{\tin}\) marginal of 
the probability measure  \(\umn{\rno,\epsilon}\mtimes\qma{\cdot,\Wmn{\tin}}{\lgm_{\cdot,\Wm\!,\costc}}\) 
where \(\umn{\rno,\epsilon}\) is the uniform probability distribution on  
\((\rno-\epsilon\rno,\rno+\epsilon(1-\rno))\):
\begin{align}
\notag
\qma{\rno,\tin}{\epsilon}
&=\tfrac{1}{\epsilon} \int_{\rno-\rno\epsilon}^{\rno+(1-\rno)\epsilon} 
\qma{\rnt,\Wmn{\tin}}{\lgm_{\rnt}} \dif{\rnt}
&
&\forall \tin \in \{1,\ldots,\blx\}.
\end{align}
Let \(\qmn{\rno,\tin}\in \pmea{\outA_{\tin}}\) and \(\qmn{\rno}\in \pmea{\outA_{1}^{\blx}}\) 
\begin{align}
\notag
\qmn{\rno,\tin}
&\DEF \tfrac{\blx}{\blx+1}\qma{\rno,\tin}{\epsilon} + \tfrac{1}{\blx+1}\qmn{\sfrac{1}{2},\Wmn{\tin},\blx\costc}
&
&\forall \tin \in \{1,\ldots,\blx\},
\\
\notag
\qmn{\rno}
&\DEF\bigotimes\nolimits_{\tin=1}^{\blx} \qmn{\rno,\tin}.
\end{align}
Let us denote the probability measure generated by \(\enc(\dmes)\), i.e. 
\(\bigotimes_{\tin=1}^{\blx}\Wmn{\tin}(\enc_{\tin}(\dmes))\),
by \(\wma{}{\dmes}\) for brevity.
\begin{align}
\notag
\RD{\rno}{\wma{}{\dmes}}{\qmn{\rno}}
&=\sum\nolimits_{\tin=1}^{\blx}  \RD{\rno}{\Wmn{\tin}(\enc_{\tin}(\dmes))}{\qmn{\rno,\tin}}
&
&\mbox{by \cite[Thm. 28]{ervenH14},}
\\
\notag
&\leq 
\sum\nolimits_{\tin=1}^{\blx}  
\left(\ln\tfrac{\blx+1}{\blx} +
\RD{\rno}{\Wmn{\tin}(\enc_{\tin}(\dmes))}{\qma{\rno,\tin}{\epsilon}}\right)
&
&\mbox{by \cite[Lemma \ref*{C-lem:divergence-RM}]{nakiboglu19C}
because \(\tfrac{\blx}{\blx+1}\qma{\rno,\tin}{\epsilon}\leq \qmn{\rno,\tin}\),}
\\
\notag
&\leq 1+\sum\nolimits_{\tin=1}^{\blx}  \RD{\rno}{\Wmn{\tin}(\enc_{\tin}(\dmes))}{\qma{\rno,\tin}{\epsilon}}
&
&\mbox{by \(\ln\tau\leq \tau-1\),}
\\
\notag
&\leq 1+\sum\nolimits_{\tin=1}^{\blx}  
\tfrac{1}{\epsilon} \int_{\rno(1-\epsilon)}^{\rno+(1-\rno)\epsilon} \RD{\rno}{\Wmn{\tin}(\enc_{\tin}(\dmes))}{\qma{\rnt,\Wmn{\tin}}{\lgm_{\rnt}}} \dif{\rnt}
&
&\mbox{by the Jensen's inequality and \cite[Thm. 12]{ervenH14},}
\\
\notag
&= 1+\tfrac{1}{\epsilon} \int_{\rno(1-\epsilon)}^{\rno+(1-\rno)\epsilon} 
\RD{\rno}{\wma{}{\dmes}}{\qma{\rnt,\Wmn{[1,\blx]}}{\lgm_{\rnt}}}\dif{\rnt}
&
&\mbox{by \cite[Thm. 28]{ervenH14} and \cite[Lemma \ref*{C-lem:Lcapacityproduct}]{nakiboglu19C},}
\\
\notag
&\leq 1+\tfrac{1}{\epsilon} \int_{\rno(1-\epsilon)}^{\rno+(1-\rno)\epsilon} 
\left(1 \vee \tfrac{1-\rnt}{\rnt}\tfrac{\rno}{1-\rno}\right)\RD{\rnt}{\wma{}{\dmes}}{\qma{\rnt,\Wmn{[1,\blx]}}{\lgm_{\rnt}}}\dif{\rnt}
&
&\mbox{by \cite[Thm. 3 and Proposition 2]{ervenH14},}
\\
\notag
&\leq 1+\tfrac{1}{\epsilon} \int_{\rno(1-\epsilon)}^{\rno+(1-\rno)\epsilon} 
\left(1 \vee \tfrac{1-\rnt}{\rnt}\tfrac{\rno}{1-\rno}\right) 
\left(\RCL{\rnt}{\Wm\!}{\lgm_{\rnt}}+\lgm_{\rnt}\cdot\costc\blx\right)
\dif{\rnt}
&
&\mbox{by \cite[Thm. \ref*{C-thm:Lminimax}]{nakiboglu19C}
	because \(\costf_{[1,\blx]}(\enc(\dmes))\!\leq\!\blx\costc\),}
\\
\label{eq:spb-cost-ologn-2}
&= 1+\CRCI{\rno}{\Wm\!}{\costc}{\epsilon}
&
&\mbox{by \eqref{eq:def:avcapacity} and \eqref{eq:spb-cost-ologn-0}.}
\end{align}
Let \(\vma{\rno}{\dmes}\) be the order \(\rno\)
tilted probability measure between \(\wma{}{\dmes}\) and \(\qmn{\rno}\)
defined in \eqref{eq:def:tiltedprobabilitymeasure}.
Then \(\vma{\rno}{\dmes}\) is a probability measure of the form
\(\vma{\rno}{\dmes}\!=\!\bigotimes\nolimits_{\tin=1}^{\blx}[\Wmn{\tin}]_{\rno}^{\qmn{\rno,\tin}}(\enc_{\tin}(\dmes))\)
as a result of the product structure of \(\wma{}{\dmes}\) and \(\qmn{\rno}\).
Let the random variables \(\cla{\rno,\tin}{\dmes}\) and \(\cla{\rno}{\dmes}\) be
\begin{align}
\notag
\cla{\rno,\tin}{\dmes}
&\DEF \ln \der{[\Wmn{\tin}(\enc_{\tin}(\dmes))]_{ac}}{\qmn{\rno,\tin}}-\EXS{\vma{\rno}{\dmes}}{\ln \der{[\Wmn{\tin}(\enc_{\tin}(\dmes))]_{ac}}{\qmn{\rno,\tin}}},
\\
\notag
\cla{\rno}{\dmes}
&\DEF \sum\nolimits_{\tin=1}^{\blx} \cla{\rno,\tin}{\dmes},
\end{align}
where \([\Wmn{\tin}(\enc_{\tin}(\dmes))]_{ac}\) is the component of \(\Wmn{\tin}(\enc_{\tin}(\dmes))\) 
that is absolutely continuous in \(\qmn{\rno,\tin}\). 
Then for all \(\knd\in\reals{+}\) and \(\rno\in(0,1)\)
\begin{align}
\notag
\EXS{\vma{\rno}{\dmes}}{\abs{\cla{\rno,\tin}{\dmes}}^{\knd}}^{\sfrac{1}{\knd}}
&\leq 3^{\sfrac{1}{\knd}} \tfrac{
	\left[(1-\rno)\RD{\rno}{\Wmn{\tin}(\enc_{\tin}(\dmes))}{\qmn{\rno,\tin}}\right]
	\vee \knd}{\rno (1-\rno)}
&&&
&\mbox{by \cite[Lemma \ref*{B-lem:MomentBound}]{nakiboglu19B},}
\\
\notag
&\leq 3^{\sfrac{1}{\knd}} 
\tfrac{
	\left[(1-\rno)\RD{\rno}{\Wmn{\tin}(\enc_{\tin}(\dmes))}{\qmn{\sfrac{1}{2},\Wmn{\tin},\blx\costc}}
	+(1-\rno)\ln (1+\blx)\right]
	\vee \knd}{\rno (1-\rno)}
&&&
&\mbox{by \cite[Lemma \ref*{C-lem:divergence-RM}]{nakiboglu19C}
because \(\tfrac{\qmn{\sfrac{1}{2},\Wmn{\tin},\blx\costc}}{\blx+1}\leq \qmn{\rno,\tin}\),}
\\
\notag
&\leq 3^{\sfrac{1}{\knd}} \tfrac{
	\left[\RD{\sfrac{1}{2}}{\Wmn{\tin}(\enc_{\tin}(\dmes))}{\qmn{\sfrac{1}{2},\Wmn{\tin},\blx\costc}}
	+(1-\rno)\ln (1+\blx)\right]
	\vee \knd}{\rno (1-\rno)}
&&&
&\mbox{by \cite[Thm. 3 and Proposition 2]{ervenH14},}
\\
\notag
&\leq 3^{\sfrac{1}{\knd}} \tfrac{
	\left[\CRC{\sfrac{1}{2}}{\Wmn{\tin}}{\blx\costc}+(1-\rno)\ln (1+\blx)\right]
	\vee \knd}{\rno (1-\rno)}
&&&
&\mbox{by \cite[Thm. \ref*{C-thm:minimax}]{nakiboglu19C}
because \(\costf_{\tin}(\enc_{\tin}(\dmes))\leq \blx\costc\).}
\end{align}
Then using the definition of \(\gamma\) given in \eqref{eq:lem:spb-cost-ologn:gamma}, we get
\begin{align}
\label{eq:spb-cost-ologn-3}
\left[\sum\nolimits_{\tin=1}^{\blx} \EXS{\vma{\rno}{\dmes}}{\abs{\cla{\rno,\tin}{\dmes}}^{\knd}}\right]^{\sfrac{1}{\knd}}
&\leq \tfrac{\gamma}{3\rno(1-\rno)}
&
&\forall \rno\in[\rnf,1).
\end{align}
The rest of the proof is analogous to the proof of Lemma \ref{lem:spb-convex}
and \eqref{eq:lem:spb-convex-alt} after \eqref{eq:spb-convex-3}.
The differences that might worth mentioning are that
bounds given in \eqref{eq:spb-cost-ologn-2} and \eqref{eq:spb-cost-ologn-3}
are invoked instead of \eqref{eq:spb-convex-2} and \eqref{eq:spb-convex-3},
and
\(\widetilde{\qmn{\tin}}=\tfrac{\blx}{\blx+1}
\qma{\widetilde{\rno},\tin}{\widetilde{\epsilon}}+\tfrac{1}{\blx+1}\qmn{\sfrac{1}{2},\Wmn{\tin},\costc}\)
is used instead of 
\(\widetilde{\qmn{\tin}}=\tfrac{\blx}{\blx+1} \qma{\widetilde{\rno},\tin}{\widetilde{\epsilon}}+\tfrac{1}{\blx+1}\qmn{\sfrac{1}{2},\Uma{\cset}{(\blx)}}\).

\end{proof}

\begin{proof}[Proof of Lemma \ref{lem:spb-cost-conjugate}]
The proof is identical to that of Lemma \ref{lem:spb-cost-ologn} and \eqref{eq:lem:spb-cost-ologn-alt}
except for the definition of \(\qmn{\rno,\tin}\) and the bound on \(\EXS{\vma{\rno}{\dmes}}{\abs{\cla{\rno,\tin}{\dmes}}^{\knd}}^{\sfrac{1}{\knd}}\).
In particular, we set \(\qmn{\rno,\tin}\!\in\!\pmea{\outA_{\tin}}\) to be
\(\tfrac{\blx}{\blx+1}\qma{\rno,\tin}{\epsilon} + \tfrac{1}{\blx+1}\qma{\sfrac{1}{2},\Wmn{\tin}}{\lgm_{\sfrac{1}{2}}}\)
for all \(\tin\leq \blx\) and bound 
\(\EXS{\vma{\rno}{\dmes}}{\abs{\cla{\rno,\tin}{\dmes}}^{\knd}}^{\sfrac{1}{\knd}}\)
as follows
	\begin{align}
	\notag
	\EXS{\vma{\rno}{\dmes}}{\abs{\cla{\rno,\tin}{\dmes}}^{\knd}}^{\sfrac{1}{\knd}}
	&\leq 3^{\sfrac{1}{\knd}}\tfrac{
\left[(1-\rno)\RD{\rno}{\Wmn{\tin}(\enc_{\tin}(\dmes))}{\qmn{\rno,\tin}}\right]\vee \knd}{\rno (1-\rno)}
	&&&
	&\mbox{by \cite[Lemma \ref*{B-lem:MomentBound}]{nakiboglu19B},}
	\\
	\notag
	&\leq 3^{\sfrac{1}{\knd}}\tfrac{
\left[(1-\rno)\RD{\rno}{\Wmn{\tin}(\enc_{\tin}(\dmes))}{\qma{\sfrac{1}{2},\Wmn{\tin}}{\lgm_{\sfrac{1}{2}}}}
+(1-\rno)\ln(1+\blx)\right]
\vee \knd}{\rno (1-\rno)}
	&&&
	&\mbox{by \cite[Lemma \ref*{C-lem:divergence-RM}]{nakiboglu19C}
	because \(\tfrac{\qma{\sfrac{1}{2},\Wmn{\tin}}{\lgm_{\sfrac{1}{2}}}}{\blx+1}\leq \qmn{\rno,\tin}\),}
	\\
	\notag
	&\leq 3^{\sfrac{1}{\knd}}\tfrac{
\left[\RD{\sfrac{1}{2}}{\Wmn{\tin}(\enc_{\tin}(\dmes))}{\qma{\sfrac{1}{2},\Wmn{\tin}}{\lgm_{\sfrac{1}{2}}}}
+(1-\rno)\ln(1+\blx)\right]	
\vee \knd}{\rno (1-\rno)}
	&&&
	&\mbox{by \cite[Thm. 3 and Proposition 2]{ervenH14},}
	\\
	\notag
	&\leq 3^{\sfrac{1}{\knd}} \tfrac{
\left[\RCL{\sfrac{1}{2}}{\Wmn{\tin}}{\lgm_{\sfrac{1}{2}}}+\lgm_{\sfrac{1}{2}}\cdot\costf_{\tin}(\enc_{\tin}(\dmes))
+(1-\rno)\ln(1+\blx)\right]
\vee \knd}{\rno (1-\rno)}
	&&&
	&\mbox{by \cite[Thm. \ref*{C-thm:Lminimax}]{nakiboglu19C}}.
	\end{align}
	Then using the definition of \(\gamma\) given in \eqref{eq:lem:spb-cost-conjugate:gamma}, we get
	\begin{align}
\notag
	\left[\sum\nolimits_{\tin=1}^{\blx} \EXS{\vma{\rno}{\dmes}}{\abs{\cla{\rno,\tin}{\dmes}}^{\knd}}\right]^{\sfrac{1}{\knd}}
	&\leq \tfrac{\gamma}{3\rno(1-\rno)}
&
&\forall \rno\in[\rnf,1).
	\end{align}
The rest of the proof is identical to that of
Lemma \ref{lem:spb-cost-ologn} and \eqref{eq:lem:spb-cost-ologn-alt}.

If \(\Wmn{\tin}\!=\!\Wm\!\) and \(\costf_{\tin}\!=\!\costf\) for all \(\tin\!\in\![1,\blx]\),
then \(\CRC{\rno}{\Wmn{[1,\blx]}}{\blx\costc}\!=\!\blx\CRC{\rno}{\Wm}{\costc}\) 
for all \(\costc\in\fcc{\costf}\)
and \(\RCL{\rno}{\Wmn{[1,\blx]}}{\lgm}\!=\!\blx\RCL{\rno}{\Wm}{\lgm}\) 
for all \(\lgm\in\reals{\geq0}^{\ell}\)
by  \cite[Lemmas \ref*{C-lem:CCcapacityproduct} and \ref*{C-lem:Lcapacityproduct}]{nakiboglu19C}.
Then \eqref{eq:lem:spb-cost-conjugate-stationary:gamma}
follows from  
\(\CRC{\rno}{\Wm}{\costc}=\RCL{\rno}{\Wm}{\lgm_{\rno,\Wm,\costc}}+\lgm_{\rno,\Wm,\costc}\cdot \costc\),
established in \cite[Lemma \ref*{C-lem:Lcapacity}-(\ref*{C-Lcapacity:interior})]{nakiboglu19C}.
\end{proof}
\section{Examples}\label{sec:examples}
As a result of \S\ref{sec:innerbound} and \S \ref{sec:outerbound},
we can conclude that the SPE governs the exponential decay  rate of the error 
probability of channel codes with list decoding on memoryless channels 
under rather mild hypotheses. 
The calculation of the SPE itself, however, is a
separate issue that is essential from a practical standpoint. 
In this section, we derive the SPE for 
various Gaussian and Poisson channels and demonstrate that
it is possible to obtain parametric forms for these channels
similar to the one given in Lemma \ref{lem:spherepacking-cc}
for \(\spe{\rate,\Wm,\mP}\). 
We believe these parametric forms are more straightforward
and intuitive than commonly used equivalent parametric forms 
that were previously derived. 
\subsection{Gaussian Channels}\label{sec:examples-gaussian}
We denote the probability density function of the zero mean 
Gaussian random variable with variance \(\sigma^{2}\) by 
\(\GausDen{\sigma^{2}}\), i.e.	
\begin{align}
\notag
\GausDen{\sigma^{2}}(\dsta)
&\DEF\tfrac{1}{\sqrt{2 \pi} \sigma} e^{-\frac{\dsta^{2}}{2\sigma^{2}}}
&
&\forall \dsta\in\reals{}.
\end{align}
With a slight abuse of notation, we denote the corresponding probability measure on 
\(\rborel{\reals{}}\) by \(\GausDen{\sigma^{2}}\), as well.	
\begin{example}[The Scalar Gaussian Channel]\label{eg:SGauss}
	Let \(\Wm\) be the scalar Gaussian channel with noise variance \(\sigma^2\) 
	and the associated cost function \(\costf\) be the quadratic one:
	\begin{align}
	\notag
	\Wm(\oev|\dinp)
	&=\int_{\oev} \GausDen{\sigma^{2}}(\dout-\dinp) \dif{\dout}
	&
	&\forall \oev\in\rborel{\reals{}},
	\\
	\notag
	\costf(\dinp)
	&=\dinp^{2}
	&
	&\forall \dinp\in\reals{}.
	\end{align}
The cost constrained Augustin capacity and center of this channel are 
	determined in \cite[Example \ref{C-eg:SGauss}]{nakiboglu19C}:
	\begin{align}
	\label{eq:eg:SGauss-capacity}
	\CRC{\rno}{\Wm}{\costc}
	&=\begin{cases}
	\tfrac{\rno \costc}{2(\rno \theta_{\rno,\sigma,\costc}+(1-\rno)\sigma^{2})}
	+\tfrac{1}{\rno-1}\ln\tfrac{(\theta_{\rno,\sigma,\costc})^{\sfrac{\rno}{2}}\sigma^{(1-\rno)}}{\sqrt{\rno\theta_{\rno,\sigma,\costc}+(1-\rno) \sigma^{2}}}
	&\rno\in\reals{+}\setminus\{1\}
	\\
	\tfrac{1}{2}\ln \left(1+\tfrac{\costc}{\sigma^{2}}\right)
	&\rno=1
	\end{cases},
	\\
	\label{eq:eg:SGauss-center}
	\qmn{\rno,\Wm,\costc}
	&=\GausDen{\theta_{\rno,\sigma,\costc}},
	\\
	\label{eq:eg:SGauss-center-variance}
	\theta_{\rno,\sigma,\costc}
	&\DEF\sigma^{2}+\tfrac{\costc}{2}-\tfrac{\sigma^{2}}{2\rno}+\sqrt{(\tfrac{\costc}{2}-\tfrac{\sigma^{2}}{2\rno})^{2}+ \costc\sigma^2}.
	\end{align}	
	It is worth mentioning that \(\CRC{\rno}{\Wm}{\costc}=\RMI{\rno}{\GausDen{\costc}}{\Wm}\)
	and \(\qmn{\rno,\Wm\!,\costc}=\qmn{\rno,\GausDen{\costc}}\)
	for all positive orders \(\rno\),
	i.e. zero mean Gaussian distribution 
	with variance \(\costc\) is the optimal input distribution for all orders.
	Thus \(\spe{\rate,\Wm,\costc}=\spe{\rate,\Wm,\GausDen{\costc}}\).

The SPE of the scalar Gaussian channel can be characterized using 
Lemma \ref{lem:spherepacking-optimality}.
To see how, first note that 
for any \(\theta>0\) and the corresponding the Gaussian probability measure \(\GausDen{\theta}\), 
the order \(\rno\) tilted channel 
\(\Wma{\rno}{\GausDen{\theta}}\), defined in \eqref{eq:def:tiltedchannel}, is given by  
\begin{align}
\label{eq:eg:SGauss-parametric-tiltedchannel}
\Wma{\rno}{\GausDen{\theta}}(\oev|\dinp)
&=\int_{\oev} \GausDen{\frac{\sigma^{2} \theta}{\rno\theta+(1-\rno)\sigma^{2}}}
\left(\dout-\tfrac{\rno \theta}{\rno\theta+(1-\rno)\sigma^{2}}\dinp\right)\dif{\dout}
&
&\forall \oev\in\rborel{\reals{}}.
\end{align}
Since 
\(\theta_{\rno,\sigma,\costc}\) is a root of the equality 
\(\theta^{2}-\theta[\costc+(2-\frac{1}{\rno}) \sigma^{2}]+(1-\frac{1}{\rno})\sigma^{4}=0\)
for \(\theta\)
by \cite[(\ref*{C-eq:eg:SGauss-Augustinoperator}) and 
(\ref*{C-eq:eg:SGauss-necessarycondition})]{nakiboglu19C},
one can confirm 
using \cite[(\ref*{C-eq:eg:SGauss-divergence})]{nakiboglu19C}
by substitution that 
\begin{align}
\label{eq:eg:SGauss-divergence-rate}
\CRD{1}{\Wma{\rno}{\GausDen{\theta_{\rno,\sigma,\costc}}}}{\GausDen{\theta_{\rno,\sigma,\costc}}}{\mP}
&=\tfrac{\rno^{2}\theta_{\rno,\sigma,\costc}}{2(\rno \theta_{\rno,\sigma,\costc}+(1-\rno)\sigma^{2})^{2}}(\EXS{\mP}{\costf}-\costc)
+\tfrac{1}{2}\ln \tfrac{\rno \theta_{\rno,\sigma,\costc}+(1-\rno) \sigma^{2}}{\sigma^{2}},
\\
\label{eq:eg:SGauss-divergence-spe}
\CRD{1}{\Wma{\rno}{\GausDen{\theta_{\rno,\sigma,\costc}}}}{\Wm}{\mP}
&=\tfrac{(1-\rno)^{2}\sigma^{2}}{2(\rno \theta_{\rno,\sigma,\costc}+(1-\rno)\sigma^{2})^{2}}(\EXS{\mP}{\costf}-\costc)
+\tfrac{(1-\rno)\costc}{2(\rno\theta_{\rno,\sigma,\costc}+(1-\rno)\sigma^{2})}
+\tfrac{1}{2}\!\ln \tfrac{\rno \theta_{\rno,\sigma,\costc}+(1-\rno) \sigma^{2}}{\theta_{\rno,\sigma,\costc}}.
\end{align}
Thus for each \(\rns\!\in\!(0,1)\), 
\(\Vmn{\mP}\!=\!\Wma{\rns}{\GausDen{\theta_{\rns,\sigma,\costc}}}\)
satisfies the hypotheses of Lemma \ref{lem:spherepacking-optimality} 
given in \eqref{eq:lem:spherepacking-optimality:hypothesis-rate} and 
\eqref{eq:lem:spherepacking-optimality:hypothesis-exponent}
for \(\rate\!=\!\tfrac{1}{2}\ln \tfrac{\rns \theta_{\rno,\sigma,\costc}+(1-\rns) \sigma^{2}}{\sigma^{2}}\)
as a result of \eqref{eq:eg:SGauss-capacity} 
and the constraint \(\EXS{\mP}{\costf}\leq\costc\).
Furthermore, 
\(\fX(\rno)\!\DEF\!\tfrac{1}{2}\ln \tfrac{\rno \theta_{\rno,\sigma,\costc}+(1-\rno) \sigma^{2}}{\sigma^{2}}\)
is a continuous and increasing function of \(\rno\)
satisfying  \(\lim\nolimits_{\rno\downarrow0}\fX(\rno)=0\)
and \(\fX(1)=\CRC{1}{\Wm}{\costc}\).
Thus the SPE can be written in the following parametric form
in terms of \(\rno\in[0,1]\) for all rates in \([0,\CRC{1}{\Wm}{\costc}]\): 
\begin{align}
\label{eq:eg:SGauss-parametric-rate}
\rate
&=\tfrac{1}{2}\ln \tfrac{\rno \theta_{\rno,\sigma,\costc}+(1-\rno) \sigma^{2}}{\sigma^{2}},
\\
\label{eq:eg:SGauss-parametric-spe}
\spe{\rate,\Wm,\costc}
&=\tfrac{(1-\rno)\costc}{2(\rno\theta_{\rno,\sigma,\costc}+(1-\rno)\sigma^{2})}
+\tfrac{1}{2}\ln \tfrac{\rno \theta_{\rno,\sigma,\costc}+(1-\rno) \sigma^{2}}{\theta_{\rno,\sigma,\costc}}.
\end{align}
Using \eqref{eq:eg:SGauss-divergence-rate} and \eqref{eq:eg:SGauss-divergence-spe}
we can express both the rate and the SPE in terms of
the tilted channel \(\Wma{\rno}{\GausDen{\theta_{\rno,\sigma,\costc}}}\). 
On the other hand, 
\(\GausDen{\theta_{\rno,\sigma,\costc}}\) is the output distribution
for the input distribution \(\GausDen{\costc}\) on 
\(\Wma{\rno}{\GausDen{\theta_{\rno,\sigma,\costc}}}\) because \(\GausDen{\theta_{\rno,\sigma,\costc}}\) 
is the Augustin mean \(\qmn{\rno,\GausDen{\costc}}\) 
for the input distribution \(\GausDen{\costc}\) satisfying 
the fixed point property
\(\Aop{\rno}{\GausDen{\costc}}{\qmn{\rno,\GausDen{\costc}}}=\qmn{\rno,\GausDen{\costc}}\), as well. Thus
we can rewrite \eqref{eq:eg:SGauss-parametric-rate} and \eqref{eq:eg:SGauss-parametric-spe},
using \eqref{eq:eg:SGauss-divergence-rate} and \eqref{eq:eg:SGauss-divergence-spe}, as follows:
\begin{align}
\label{eq:eg:SGauss-parametric-rate-alt}
\rate
&=\RMI{1}{\GausDen{\costc}}{\Wma{\rno}{\GausDen{\theta_{\rno,\sigma,\costc}}}},
\\
\label{eq:eg:SGauss-parametric-spe-alt}
\spe{\rate,\Wm,\costc}
&=\CRD{1}{\Wma{\rno}{\GausDen{\theta_{\rno,\sigma,\costc}}}}{\Wm}{\GausDen{\costc}}.
\end{align}
To obtain an expression for the  SPE  that does not depend on \(\theta_{\rno,\sigma,\costc}\) explicitly,
we first note that \eqref{eq:eg:SGauss-center-variance} and
\eqref{eq:eg:SGauss-parametric-rate} imply 
\begin{align}
\label{eq:eg:SGauss-rno}
\rno=\tfrac{e^{2\rate}-1}{2}\left(\sqrt{1+\tfrac{4\sigma^2}{\costc}\tfrac{e^{2\rate}}{e^{2\rate}-1}}-1\right).
\end{align}
On the other hand, \(\GausDen{\theta_{\rno,\sigma,\costc}}\) is the output distribution
for the input distribution \(\GausDen{\costc}\) on channel 
\(\Wma{\rno}{\GausDen{\theta_{\rno,\sigma,\costc}}}\).
Thus \eqref{eq:eg:SGauss-parametric-tiltedchannel} implies
\begin{align}
\notag
\tfrac{\sigma^{2} \theta_{\rno,\sigma,\costc}}{\rno\theta_{\rno,\sigma,\costc}+(1-\rno)\sigma^{2}}
+\left(\tfrac{\rno \theta_{\rno,\sigma,\costc}}{\rno\theta_{\rno,\sigma,\costc}+(1-\rno)\sigma^{2}}\right)^{2}
\costc
&=\theta_{\rno,\sigma,\costc}.
\end{align}
Thus 
\begin{align}
\notag
\tfrac{\sigma^{2}}{\theta_{\rno,\sigma,\costc}}
&=1-\tfrac{\rno\costc}{\rno\theta_{\rno,\sigma,\costc}+(1-\rno)\sigma^{2}}
\\
\label{eq:eg:SGauss-parametric-inversetheta}
&=1-\tfrac{\costc \rno}{\sigma^{2}e^{2\rate}},
\end{align}
where \eqref{eq:eg:SGauss-parametric-inversetheta} follows from \eqref{eq:eg:SGauss-parametric-rate}.

Using first
\eqref{eq:eg:SGauss-parametric-rate}\!
and
\!\eqref{eq:eg:SGauss-parametric-inversetheta}
in \eqref{eq:eg:SGauss-parametric-spe},
and then invoking \eqref{eq:eg:SGauss-rno}
we get the following expression
for SPE:
\begin{align}
\notag
\spe{\rate,\Wm,\costc}
&=\tfrac{1}{2}\tfrac{(1-\rno)\costc}{\sigma^{2}e^{2\rate}}
+\rate
+\tfrac{1}{2}\ln \left(1-\tfrac{\rno\costc }{\sigma^{2}e^{2\rate}}\right)
\\
\label{eq:eg:SGauss-spe}
&=\tfrac{\costc}{4\sigma^{2}}\left[1+
\tfrac{1}{e^{2\rate}}-(1-\tfrac{1}{e^{2\rate}})\sqrt{1+\tfrac{4\sigma^2}{\costc}\tfrac{e^{2\rate}}{e^{2\rate}-1}}\right]
+\tfrac{1}{2}\ln\left[e^{2\rate}-\tfrac{\costc}{\sigma^{2}}\left(\tfrac{e^{2\rate}-1}{2}\right)
\left(\sqrt{1+\tfrac{4\sigma^2}{\costc}\tfrac{e^{2\rate}}{e^{2\rate}-1}}-1\right)\right].
\end{align}
The expression given in  \eqref{eq:eg:SGauss-spe} for the SPE is equivalent to \cite[(7.4.33)]{gallager}.

The parametric characterization given in 
\eqref{eq:eg:SGauss-parametric-rate} and \eqref{eq:eg:SGauss-parametric-spe}
can be obtained by a more direct approach using the differentiability of 
\(\theta_{\rno,\sigma,\costc}\) and \(\CRC{\rno}{\Wm}{\costc}\)
in \(\rno\), as well.
In particular, since \(\theta_{\rno,\sigma,\costc}\) is a root of the equality 
	\(\theta^{2}-\theta[\costc+(2-\frac{1}{\rno}) \sigma^{2}]+(1-\frac{1}{\rno})\sigma^{4}=0\)
	for \(\theta\), we get the following closed form expression for the derivative of the 
	Augustin capacity with respect to the order:
\begin{align}
\notag
\pder{}{\rno}\CRC{\rno}{\Wm}{\costc}
&=\tfrac{1}{2(1-\rno)^{2}}
\left[\tfrac{(1-\rno)\costc}{\rno \theta_{\rno,\sigma,\costc}+(1-\rno)\sigma^{2}}
+\ln \tfrac{\rno \theta_{\rno,\sigma,\costc}+(1-\rno)\sigma^{2}}{\theta_{\rno,\sigma,\costc}}
\right]
+\tfrac{\theta_{\rno,\sigma,\costc}^{2}-\theta_{\rno,\sigma,\costc}[\costc+(2-\frac{1}{\rno}) \sigma^{2}]+(1-\frac{1}{\rno})\sigma^{4}}{2(\rno \theta_{\rno,\sigma,\costc}+(1-\rno)\sigma^{2})^{2}}\left[\tfrac{\rno}{1-\rno}+\tfrac{\rno^{2}}{\theta_{\rno,\sigma,\costc}}
\pder{}{\rno}\theta_{\rno,\sigma,\costc}\right]
\\
\label{eq:eg:SGauss-capacity-derivative}
&=\tfrac{1}{2(1-\rno)^{2}}
\left[\tfrac{(1-\rno)\costc}{\rno \theta_{\rno,\sigma,\costc}+(1-\rno)\sigma^{2}}
+\ln \tfrac{\rno \theta_{\rno,\sigma,\costc}+(1-\rno)\sigma^{2}}{\theta_{\rno,\sigma,\costc}}
\right].
\end{align} 
Using first \eqref{eq:eg:SGauss-capacity-derivative}, and then \eqref{eq:eg:SGauss-capacity} we get
\begin{align}
\notag
\der{}{\rno}\tfrac{1-\rno}{\rno}(\CRC{\rno}{\!\Wm\!}{\costc}-\rate)
&=\tfrac{1-\rno}{\rno}
\left(\tfrac{1}{2(1-\rno)^{2}}
\left[\tfrac{(1-\rno)\costc}{\rno \theta_{\rno,\sigma,\costc}+(1-\rno)\sigma^{2}}
+\ln \tfrac{\rno \theta_{\rno,\sigma,\costc}+(1-\rno)\sigma^{2}}{\theta_{\rno,\sigma,\costc}}
\right]
\right)
-\tfrac{1}{\rno^{2}}(\CRC{\rno}{\!\Wm\!}{\costc}-\rate)
\\
\notag
&=\tfrac{1}{\rno^{2}}\left(\rate-\tfrac{1}{2}\ln \tfrac{\rno\theta_{\rno,\sigma,\costc}+(1-\rno)\sigma^{2}}{\sigma^{2}}
\right).
\end{align} 
Then the derivative test implies the parametric form given in \eqref{eq:eg:SGauss-parametric-rate} and \eqref{eq:eg:SGauss-parametric-spe}
as a result of \eqref{eq:eg:SGauss-capacity}.
\end{example}

\begin{example}[The Parallel Gaussian Channels]\label{eg:PGauss}
Let \(\Wmn{[1,\blx]}\) be the product of scalar Gaussian channels with noise variance  \(\sigma_{\ind}^{2}\) 
for \(\ind\in\{1,\ldots,\blx\}\) and the cost function \(\costf_{[1,\blx]}\) be the additive quadratic one, i.e. 
	\begin{align}
	\notag
	\Wmn{[1,\blx]}(\oev|\dinp_{1}^{\blx})
	&=\int_{\oev}\left[
	\prod\nolimits_{\ind=1}^{\blx}\GausDen{\sigma_{\ind}^{2}}(\dout_{\ind}-\dinp_{\ind})
	\right]\dif{\dout_{1}^{\blx}}
	&
	&\forall \oev\in\rborel{\reals{}^{\blx}},
	\\
	\notag
	\costf_{[1,\blx]}(\dinp_{1}^{\blx})
	&=\sum\nolimits_{\ind=1}^{\blx} \dinp_{\ind}^{2}
	&
	&\forall \dinp_{1}^{\blx}\in\reals{}^{\blx}.
	\end{align}
	The constrained Augustin capacity and center of \(\Wmn{[1,\blx]}\) were determined in 
	\cite[Example \ref{C-eg:PGauss}]{nakiboglu19C}:
	\begin{align}
\label{eq:eg:PGauss-capacity}
\CRC{\rno}{\!\Wmn{[1,\blx]}\!}{\costc}
&=\sum\nolimits_{\ind=1}^{\blx} \CRC{\rno}{\Wmn{\ind}}{\costc_{\rno,\ind}},
\\
\label{eq:eg:PGauss-capacitycenter}
\qmn{\rno,\!\Wmn{[1,\blx]}\!,\costc}
&=\bigotimes\nolimits_{\ind=1}^{\blx}\GausDen{\theta_{\rno,\sigma_{\ind},\costc_{\rno,\ind}}},
\\
\label{eq:eg:PGauss-input-variance}
\costc_{\rno,\ind}
&=
\tfrac{\abp{\rno-2\sigma_{\ind}^{2}\lgm_{\rno}}}{2\lgm_{\rno}(\rno +2(\rno-1)\sigma_{\ind}^2\lgm_{\rno})},
\end{align}
where \(\theta_{\rno,\sigma,\costc}\) is defined in \eqref{eq:eg:SGauss-center-variance}
and \(\lgm_{\rno}\) is determined by \(\sum_{\ind} \costc_{\rno,\ind}=\costc\)
uniquely.\footnote{The constraint \(\sum_{\ind} \costc_{\rno,\ind}=\costc\) determines \(\lgm_{\rno}\)
uniquely because the expression on the right hand side of \eqref{eq:eg:PGauss-input-variance} 
is a nonincreasing function of  \(\lgm_{\rno}\).} 
Furthermore,
\(\theta_{\rno,\sigma_{\ind},\costc_{\rno,\ind}}\) can be expressed 
in terms of \(\sigma_{\ind}\) and \(\lgm_{\rno}\) without explicitly referring to 
\(\costc_{\rno,\ind}\) as follows:
\begin{align}
\label{eq:eg:PGauss-center-variance}
\theta_{\rno,\sigma_{\ind},\costc_{\rno,\ind}}
&=\sigma_{\ind}^{2}+\abp{\tfrac{1}{2\lgm_{\rno}}-\tfrac{\sigma_{\ind}^{2}}{\rno}}.
\end{align}
On the other hand
\(\der{}{\costc_{\ind}}\CRC{\rno}{\!\Wmn{\ind}}{\costc_{\ind}}\vert_{\costc_{\ind}=\costc_{\rno,\ind}}\!=\!\lgm_{\rno}\)
for all \(\ind\)'s with a positive \(\costc_{\rno,\ind}\) and
\(\der{}{\costc_{\ind}}\CRC{\rno}{\!\Wmn{\ind}}{\costc_{\ind}}\vert_{\costc_{\ind}=\costc_{\rno,\ind}}\!\leq\!\lgm_{\rno}\)
for all \(\ind\)'s.
Then using the chain rule of derivatives 
together with 
\eqref{eq:eg:SGauss-capacity-derivative}
and \eqref{eq:eg:SGauss-capacity} we get
\begin{align}
\notag
\der{}{\rno}\tfrac{1-\rno}{\rno}(\CRC{\rno}{\!\Wmn{[1,\blx]}\!}{\costc}-\rate)
&=\tfrac{1-\rno}{\rno}
\sum\nolimits_{\ind=1}^{\blx}
\left[\left.\pder{}{\rnf}\CRC{\rnf}{\Wmn{\ind}}{\costc_{\rno,\ind}}\right\vert_{\rnf=\rno}
+\left.\pder{}{\costc}\CRC{\rno}{\Wmn{\ind}}{\costc}\right\vert_{\costc=\costc_{\rno,\ind}}
\der{}{\rno}\costc_{\rno,\ind}
\right]
-\tfrac{1}{\rno^{2}}(\CRC{\rno}{\!\Wmn{[1,\blx]}\!}{\costc}-\rate),
\\
\notag
&=\tfrac{1}{\rno^{2}}\left[\rate-
\tfrac{1}{2}\sum\nolimits_{\ind=1}^{\blx}
\ln \tfrac{\rno\theta_{\rno,\sigma_{\ind},\costc_{\rno,\ind}}+(1-\rno)\sigma_{\ind}^{2}}{\sigma_{\ind}^{2}}
\right]+\tfrac{1-\rno}{\rno} \lgm_{\rno}\sum\nolimits_{\ind=1}^{\blx}\der{}{\rno}\costc_{\rno,\ind},
\\
\notag
&=\tfrac{1}{\rno^{2}}\left[\rate-
\tfrac{1}{2}\sum\nolimits_{\ind=1}^{\blx}
\ln \tfrac{\rno\theta_{\rno,\sigma_{\ind},\costc_{\rno,\ind}}+(1-\rno)\sigma_{\ind}^{2}}{\sigma_{\ind}^{2}}
\right].
\end{align} 
Thus we obtain the following parametric form for
\(\spe{\rate,\Wmn{[1,\blx]},\costc}\) in term of \(\rno\) for all \(\rate\in
[0,\CRC{\rno}{\!\Wmn{[1,\blx]}\!}{\costc}]\).
\begin{align}
\label{eq:eg:PGauss-parametric-rate}
\rate
&=\tfrac{1}{2} 
\sum\nolimits_{\ind=1}^{\blx}
\ln \tfrac{\rno\theta_{\rno,\sigma_{\ind},\costc_{\rno,\ind}}+(1-\rno)\sigma_{\ind}^{2}}{\sigma_{\ind}^{2}},
\\
\label{eq:eg:PGauss-parametric-spe}
\spe{\rate,\Wmn{[1,\blx]},\costc}
&=\tfrac{1}{2}
\sum\nolimits_{\ind=1}^{\blx}
\left[\tfrac{(1-\rno)\costc_{\rno,\ind}}{\rno\theta_{\rno,\sigma_{\ind},\costc_{\rno,\ind}}+(1-\rno)\sigma_{\ind}^{2}}
+\ln \tfrac{\rno\theta_{\rno,\sigma_{\ind},\costc_{\rno,\ind}}+(1-\rno)\sigma_{\ind}^{2}}{\theta_{\rno,\sigma_{\ind},\costc_{\rno,\ind}}}
\right].
\end{align}
Thus \(\spe{\rate,\Wmn{[1,\blx]},\costc}\!=\!\CRD{1}{\Vmn{\rno}}{\Wmn{[1,\blx]}}{\varPhi_{\rno}}\)
for \(\rate\!=\!\RMI{1}{\varPhi_{\rno}}{\Vmn{\rno}}\)
where the input distribution
\(\varPhi_{\rno}\) is a zero mean the Gaussian distribution with the diagonal covariance
matrix whose eigenvalues  are \(\costc_{\rno,1},\ldots,\costc_{\rno,\blx}\)
and \(\Vmn{\rno}\) is the order \(\rno\) tilted channel between
\(\Wmn{[1,\blx]}\) and \(\qmn{\rno,\varPhi_{\rno}}\).

If \(\sigma_{\ind}^{2}\!\geq\!\tfrac{\rno}{2\lgm_{\rno}}\) for an \(\ind\),  
then \(\costc_{\rno,\ind}\!=\!0\) and the corresponding 
terms in the sums given in \eqref{eq:eg:PGauss-parametric-rate}
and \eqref{eq:eg:PGauss-parametric-spe} are zero. 
In \cite{ebert66}, Ebert provided an alternative parametric form for 
the SPE relying on this observation.
To obtain Ebert's characterization first note that 
\begin{align}
\notag
\rno\theta_{\rno,\sigma_{\ind},\costc_{\rno,\ind}}+(1-\rno)\sigma_{\ind}^{2}
&=\tfrac{\rno}{2\lgm_{\rno}}\vee \sigma_{\ind}^{2}
&
&\forall \ind\in\{1,\ldots,\blx\}
\end{align}
by \eqref{eq:eg:PGauss-center-variance}. Thus
\eqref{eq:eg:PGauss-center-variance},
\eqref{eq:eg:PGauss-parametric-rate}, \eqref{eq:eg:PGauss-parametric-spe} and the constraint
\(\sum_{\ind} \costc_{\rno,\ind}=\costc\) imply the following parametric form in terms of 
\(\ebertsthreshold=\tfrac{\rno}{2\lgm_{\rno}}\), which is  
equivalent to Ebert's characterization 
\cite[(7.5.28), (7.5.32), (7.5.34)]{gallager},
\cite[p. 294]{ebert65}, \cite[(20)]{ebert66}
\begin{align}
\label{eq:eg:PGauss-ebert-rate}
\rate
&=\tfrac{1}{2}\sum\nolimits_{\ind:\sigma_{\ind}^2\leq\ebertsthreshold}\ln \tfrac{\ebertsthreshold}{\sigma_{\ind}^{2}},
\\
\label{eq:eg:PGauss-ebert-rno}
\costc
&=\tfrac{1}{\rno}\sum\nolimits_{\ind=1}^{\blx}\tfrac{\abp{\ebertsthreshold-\sigma_{\ind}^{2}}}{1 +(\rno-1)\frac{\sigma_{\ind}^2}{\ebertsthreshold}},
\\
\label{eq:eg:PGauss-ebert-spe}
\spe{\rate,\Wmn{[1,\blx]},\costc}
&=\tfrac{(1-\rno)\costc}{2\ebertsthreshold}
+\tfrac{1}{2}
\sum\nolimits_{\ind:\sigma_{\ind}^2\leq \ebertsthreshold}
\ln \tfrac{\rno\ebertsthreshold}{\ebertsthreshold-(1-\rno)\sigma_{\ind}^{2}}.
\end{align}
Note that one does not need to determine \(\lambda_{\rno}\)
or invoke \(\ebertsthreshold=\tfrac{\rno}{2\lgm_{\rno}}\)
in the above expressions.
One can first determine \(\ebertsthreshold\) using \eqref{eq:eg:PGauss-ebert-rate}, 
and then determine \(\rno\) using \eqref{eq:eg:PGauss-ebert-rno}, in order to determine
\(\spe{\rate,\Wmn{[1,\blx]},\costc}\), as noted by Ebert in \cite{ebert65} and \cite{ebert66}.
\end{example}

\subsection{Poisson Channels}\label{sec:examples-poisson}
Let \(\tlx \in \reals{+}\) and \(\mA,\mB\in\reals{\geq0}\) such that \(\mA\leq\mB\).
Then for the Poisson channel \(\Pcha{}:\fXS\to\pmea{\outA}\), 
the input set \(\fXS\) is the set of all measurable functions of the form 
\(\fX:(0,\tlx]\to[\mA,\mB]\), 
the output set \(\outS\) is the set of all nondecreasing, right-continuous, 
integer valued  functions on \((0,\tlx]\),
the \(\sigma\)-algebra of the output events \(\outA\) is the Borel \(\sigma\)-algebra 
for the topology generated by the Skorokhod metric on \(\outS\),
and \(\Pcha{}(\fX)\) is the Poisson point process with deterministic intensity function \(\fX\) 
for all \(\fX\in\fXS\).
With a slight abuse of notation we denote the Poisson process with constant intensity 
\(\gamma\) by \(\Pcha{}(\gamma)\).
The cost function \(\costf:\fXS\to\reals{\geq0}\) is
\begin{align}
\notag	
\costf(\fX)
&\DEF\tfrac{1}{\tlx}\int_{0}^{\tlx}  \fX(\tin) \dif{\tin}.
\end{align}
In \cite[{\S\ref*{A-sec:examples-poisson}}]{nakiboglu19A},
the (unconstrained) \renyi capacities and centers of various Poisson channels 
are determined. These 
expressions 
are equal to the corresponding Augustin capacities and centers
because \(\RCL{\rno}{\Wm\!}{0}=\GCL{\rno}{\Wm\!}{0}\) for any \(\Wm\) 
and \(\qma{\rno,\Wm}{0}=\qga{\rno,\Wm}{0}\) for any \(\Wm\) with finite 
\(\RCL{\rno}{\Wm\!}{0}\) by 
\cite[Thms. \ref*{C-thm:Lminimax} and \ref*{C-thm:Gminimax}]{nakiboglu19C}.
\begin{example}[The Poisson Channels With Given Average Intensity]\label{eg:poissonchannel-mean}
\cite[Example \ref*{A-eg:poissonchannel-mean}]{nakiboglu19A}
	considers \(\Pcha{\costc}:\fXS^{\costc}\to\pmea{\outA}\) 
	where \(\fXS^{\costc}=\{\fX\in\fXS:\costf(\fX)=\costc\}\)
	and \(\Pcha{\costc}(\fX)=\Pcha{}(\fX)\) for all \(\fX\in\fXS^{\costc}\).
	The \renyi  capacity and center of \(\Pcha{\costc}\)
	---hence the Augustin capacity and center of \(\Pcha{\costc}\)--- 
	are given in \cite[{(\ref*{A-eq:poissonchannel-mean-capacity}) and 
		(\ref*{A-eq:poissonchannel-mean-center})}]{nakiboglu19A} to be
	\begin{align}		
	\label{eq:poissonchannel-mean-capacity}
	\RC{\rno}{\Pcha{\costc}}
	&=\begin{cases}
	\tfrac{\rno}{\rno-1}(\pint_{\rno,\costc}-\costc)\tlx
	&\rno\neq 1
	\\
	\left(
	\tfrac{\costc-\mA}{\mB-\mA} \mB\ln \tfrac{\mB}{\costc}
	+\tfrac{\mB-\costc}{\mB-\mA} \mA\ln \tfrac{\mA}{\costc}
	\right)\tlx
	&\rno=1
	\end{cases},
	\\
	\label{eq:poissonchannel-mean-center}
	\qmn{\rno,\Pcha{\costc}}
	&=\Pcha{}(\pint_{\rno,\costc}),
	\\
	\label{eq:poissonchannel-mean-center-intensity}
	\pint_{\rno,\costc}
	&\DEF \left(\tfrac{\costc-\mA}{\mB-\mA} \mB^{\rno}+\tfrac{\mB-\costc}{\mB-\mA} \mA^{\rno}\right)^{\sfrac{1}{\rno}}.
	\end{align}
	Since the expression for \(\RC{\rno}{\Pcha{\costc}}\) is differentiable in \(\rno\);
	we obtain the following parametric expression for the SPE 
	using the derivative test
	\begin{align}
	\label{eq:eg:poissonchannel-mean-parametric-rate}
	\rate
	&=\left(\tfrac{\costc-\mA}{\mB-\mA}
	\mB^{\rno}\pint_{\rno,\costc}^{1-\rno}\ln\tfrac{\mB^{\rno}\pint_{\rno,\costc}^{1-\rno}}{\pint_{\rno,\costc}}+
	\tfrac{\mB-\costc}{\mB-\mA}
	\mA^{\rno}\pint_{\rno,\costc}^{1-\rno}\ln\tfrac{\mA^{\rno}\pint_{\rno,\costc}^{1-\rno}}{\pint_{\rno,\costc}}
	\right)\tlx,
	\\
	\label{eq:eg:poissonchannel-mean-parametric-spe}
	\spe{\rate,\Pcha{\costc}}
	&=\left(\costc-\pint_{\rno,\costc}
	+\tfrac{\costc-\mA}{\mB-\mA}\mB^{\rno}\pint_{\rno,\costc}^{1-\rno}
	\ln\tfrac{\mB^{\rno}\pint_{\rno,\costc}^{1-\rno}}{\mB}
	+\tfrac{\mB-\costc}{\mB-\mA}\mA^{\rno}\pint_{\rno,\costc}^{1-\rno}
	\ln\tfrac{\mA^{\rno}\pint_{\rno,\costc}^{1-\rno}}{\mA}
	\right)\tlx.
	\end{align}
	There is an alternative parametric characterization, which is considerably easier to remember
	in terms of the tilted channels. In order to derive that expression, first note that 
	\cite[(\ref*{A-eq:poissonRND}), (\ref*{A-eq:poissondivergence})]{nakiboglu19A}
	and the definition of the tilted channel given in \eqref{eq:def:tiltedchannel} 
	imply that
	\begin{align}
	\label{eq:eg:poissonchannel:tiltingformula}
	{\Pcha{}}_{\rno}^{\Pcha{}(\gX)}(\fX)
	&=\Pcha{}(\fX^{\rno}\gX^{(1-\rno)}). 
	\end{align}
	Then using \cite[(\ref*{A-eq:poissondivergence})]{nakiboglu19A}, 
	\eqref{eq:poissonchannel-mean-center}, \eqref{eq:poissonchannel-mean-center-intensity}, 
	\eqref{eq:eg:poissonchannel-mean-parametric-rate}, \eqref{eq:eg:poissonchannel-mean-parametric-spe},
	we get
	\begin{align}
	\label{eq:eg:poissonchannel-mean-parametric-alternative}
	\rate
	&=\RD{1}{{\Pcha{}}_{\rno}^{\qmn{\rno,\Pcha{\costc}}}(\fX_{opt})}{\qmn{\rno,\Pcha{\costc}}},
	\\
	\label{eq:eg:poissonchannel-mean-parametric-spe-alternative}
	\spe{\rate,\Pcha{\costc}}
	&=\RD{1}{{\Pcha{}}_{\rno}^{\qmn{\rno,\Pcha{\costc}}}(\fX_{opt})}{\Pcha{}(\fX_{opt})},
	\end{align}
	where \(\fX_{opt}\) is any \(\{\mA,\mB\}\) valued function in \(\fXS^{\costc}\), i.e.
	any function \(\fX_{opt}:(0,\tlx]\to\{\mA,\mB\}\) satisfying
	\(\costf(\fX_{opt})=\costc\).
\end{example}

\begin{example}[The Poisson Channels With Constrained Average Intensity]\label{eg:poissonchannel-constrained}
Let us first confirm that the constrained Augustin capacity \(\CRC{\rno}{\Pcha{}}{\costc}\)
and the constrained Augustin center \(\qmn{\rno,\Pcha{},\costc}\) are given 
by\footnote{Note that 
\(\CRC{\rno}{\Pcha{}}{\costc}\!=\!\RC{\rno}{\Pcha{\leq\costc}}\) and 
\(\qmn{\rno,\Pcha{},\costc}\!=\!\qmn{\rno,\Pcha{\leq\costc}}\)
for the Poisson Channel \(\Pcha{\leq\costc}\!:\!\{\fX\!\in\!\fXS\!:\!\costf(\fX)\!\leq\!\costc\}\!\to\!\pmea{\outA}\)
considered in \cite[Example \ref*{A-eg:poissonchannel-constrained}]{nakiboglu19A}.}
\begin{align}
\label{eq:poissonchannel-constrained-capacity}
\CRC{\rno}{\Pcha{}}{\costc}
&=\RC{\rno}{\Pcha{\costc \wedge \costc_{\rno}}},
\\
\label{eq:poissonchannel-constrained-center}
\qmn{\rno,\Pcha{},\costc}
&=\qmn{\rno,\Pcha{\costc \wedge \costc_{\rno}}},
\end{align}
where \(\RC{\rno}{\Pcha{\costc}}\) and \(\qmn{\rno,\Pcha{\costc}}\) are determined by
\eqref{eq:poissonchannel-mean-capacity},
\eqref{eq:poissonchannel-mean-center}, 
and
\eqref{eq:poissonchannel-mean-center-intensity},
\(\costc_{\rno}\) is a decreasing function of the order 
\(\rno\) defined as
\begin{align}
\label{eq:poissonchannel-bounded-optimal}
\costc_{\rno}
&\DEF\begin{cases}
\rno^{\frac{\rno}{1-\rno}}(\tfrac{\mB-\mA}{\mB^{\rno}-\mA^{\rno}})^{\frac{1}{1-\rno}}
+\tfrac{\mA\mB^{\rno}-\mB \mA^{\rno}}{\mB^{\rno}-\mA^{\rno}}
&
\rno\neq 1
\\
e^{-1}\mB^{\frac{\mB}{\mB-\mA}}\mA^{-\frac{\mA}{\mB-\mA}}
&
\rno=1
\end{cases}.
\end{align}
To establish \eqref{eq:poissonchannel-constrained-capacity} and \eqref{eq:poissonchannel-constrained-center},
first note that \(\RC{\rno}{\Pcha{\costc \wedge \costc_{\rno}}}\leq\CRC{\rno}{\Pcha{}}{\costc}\)
because \(\RC{\rno}{\Pcha{\costc \wedge \costc_{\rno}}}=\CRC{\rno}{\Pcha{}}{\pdis{\fXS^{\costc \wedge \costc_{\rno}}}}\)
and \(\pdis{\fXS^{\costc \wedge \costc_{\rno}}}\subset\cset(\costc)\),
where \(\cset(\costc)=\{\mP\in\pdis{\fXS}:\EXS{\mP}{\costf(\fX)}\leq\costc\}\).	
On the other hand, invoking first
\cite[(\ref*{A-eq:poissonchannel-constrained-A})]{nakiboglu19A},
and then 
\cite[(\ref*{A-eq:poissonchannel-mean-capacity-alt})]{nakiboglu19A},
we get
\begin{align}
\notag
\RD{\rno}{\Pcha{}(\fX)}{\Pcha{}(\pint_{\rno,\costc \wedge \costc_{\rno}})}
&\leq \tfrac{\mB-\costf(\fX)}{\mB-\mA}\RD{\rno}{\Pcha{}(\mA)}{\Pcha{}(\pint_{\rno,\costc \wedge \costc_{\rno}})}
+\tfrac{\costf(\fX)-\mA}{\mB-\mA} \RD{\rno}{\Pcha{}(\mB)}{\Pcha{}(\pint_{\rno,\costc \wedge \costc_{\rno}})}
\\
\notag
&=\RC{\rno}{\Pcha{\costc \wedge \costc_{\rno}}}+
\tfrac{\costc \wedge \costc_{\rno}-\costf(\fX)}{\mB-\mA}
\left[\RD{\rno}{\Pcha{}(\mA)}{\Pcha{}(\pint_{\rno,\costc \wedge \costc_{\rno}})}
-\RD{\rno}{\Pcha{}(\mB)}{\Pcha{}(\pint_{\rno,\costc \wedge \costc_{\rno}})}\right]
&
&\forall \fX\in\fXS.
\end{align}
Since \(\RD{\rno}{\Pcha{}(\mA)}{\Pcha{}(\pint_{\rno,\costc_{\rno}})}\!=\!
\RD{\rno}{\Pcha{}(\mB)}{\Pcha{}(\pint_{\rno,\costc_{\rno}})}\)
by \cite[(\ref*{A-eq:poissonchannel-constrained-B}) and (\ref*{A-eq:poissonchannel-constrained-C})]{nakiboglu19A},
we get
\begin{align}
\notag
\RD{\rno}{\Pcha{}(\fX)}{\Pcha{}(\pint_{\rno,\costc \wedge \costc_{\rno}})}
&\leq \RC{\rno}{\Pcha{\costc \wedge \costc_{\rno}}}+
\IND{\costc<\costc_{\rno}}\tfrac{\costc-\costf(\fX)}{\mB-\mA}
\left[\RD{\rno}{\Pcha{}(\mA)}{\Pcha{}(\pint_{\rno,\costc})}
-\RD{\rno}{\Pcha{}(\mB)}{\Pcha{}(\pint_{\rno,\costc})}\right]
&
&\forall \fX\in\fXS.
\end{align}
On the other hand \(\RD{\rno}{\Pcha{}(\mA)}{\Pcha{}(\pint_{\rno,\costc})}
\!\leq\!\RD{\rno}{\Pcha{}(\mB)}{\Pcha{}(\pint_{\rno,\costc})}\)
for all \(\costc\leq \costc_{\rno}\) by \cite[(\ref*{A-eq:poissonchannel-constrained-B})]{nakiboglu19A};
consequently, we have
\begin{align}
\notag
\CRD{\rno}{\Pcha{}(\fX)}{\Pcha{}(\pint_{\rno,\costc \wedge \costc_{\rno}})}{\mP}
&\leq \RC{\rno}{\Pcha{\costc \wedge \costc_{\rno}}}
&
&\forall \mP:\EXS{\mP}{\costf}\leq\costc.
\end{align}
Thus \eqref{eq:poissonchannel-constrained-capacity}
and
\eqref{eq:poissonchannel-constrained-center}
follow from 
\cite[Lemma \ref*{C-lem:capacityunion}]{nakiboglu19C}
because 
\(\Pcha{}(\pint_{\rno,\costc \wedge \costc_{\rno}})=\qmn{\rno,\Pcha{\costc \wedge \costc_{\rno}}}\)
by \eqref{eq:poissonchannel-mean-center}.

Since the expression for \(\CRC{\rno}{\Pcha{}}{\costc}\) given in 
\eqref{eq:poissonchannel-constrained-capacity} is differentiable in \(\rno\),
we can use the derivative test to determine optimal order for the SPE
defined in \eqref{eq:def:spherepackingexponent}. 
We obtain the following parametric form as a result
\begin{align}
\label{eq:eg:poissonchannel-bounded-parametric-rate}
\rate
&=\left(\tfrac{\costc\wedge\costc_{\rno}-\mA}{\mB-\mA}
\mB^{\rno}\pint_{\rno,\costc\wedge\costc_{\rno}}^{1-\rno}\ln\tfrac{\mB^{\rno}\pint_{\rno,\costc\wedge\costc_{\rno}}^{1-\rno}}{\pint_{\rno,\costc\wedge\costc_{\rno}}}+
\tfrac{\mB-\costc\wedge\costc_{\rno}}{\mB-\mA}
\mA^{\rno}\pint_{\rno,\costc\wedge\costc_{\rno}}^{1-\rno}\ln\tfrac{\mA^{\rno}\pint_{\rno,\costc\wedge\costc_{\rno}}^{1-\rno}}{\pint_{\rno,\costc\wedge\costc_{\rno}}}
\right)\tlx,
\\
\label{eq:eg:poissonchannel-bounded-parametric-spe}
\spe{\rate,\Pcha{},\costc}
&=\left(\costc\wedge\costc_{\rno}-\pint_{\rno,\costc\wedge\costc_{\rno}}
+\tfrac{\costc\wedge\costc_{\rno}-\mA}{\mB-\mA}\mB^{\rno}\pint_{\rno,\costc\wedge\costc_{\rno}}^{1-\rno}
\ln\tfrac{\mB^{\rno}\pint_{\rno,\costc\wedge\costc_{\rno}}^{1-\rno}}{\mB}
+\tfrac{\mB-\costc\wedge\costc_{\rno}}{\mB-\mA}\mA^{\rno}\pint_{\rno,\costc\wedge\costc_{\rno}}^{1-\rno}
\ln\tfrac{\mA^{\rno}\pint_{\rno,\costc\wedge\costc_{\rno}}^{1-\rno}}{\mA}
\right)\tlx.
\end{align}
Using  \cite[(\ref*{A-eq:poissondivergence})]{nakiboglu19A}, 
\eqref{eq:poissonchannel-mean-center},
\eqref{eq:eg:poissonchannel:tiltingformula},
\eqref{eq:poissonchannel-constrained-center},
\eqref{eq:eg:poissonchannel-bounded-parametric-rate},
\eqref{eq:eg:poissonchannel-bounded-parametric-spe},
we get the following parametric characterization
\begin{align}
\label{eq:eg:poissonchannel-bounded-parametric-alternative}
\rate
&=\RD{1}{{\Pcha{}}_{\rno}^{\qmn{\rno,\Pcha{},\costc}}(\fX_{\rno})}{\qmn{\rno,\Pcha{},\costc}},
\\
\label{eq:eg:poissonchannel-bounded-parametric-spe-alternative}
\spe{\rate,\Pcha{},\costc}
&=\RD{1}{{\Pcha{}}_{\rno}^{\qmn{\rno,\Pcha{},\costc}}(\fX_{\rno})}{\Pcha{}(\fX_{\rno})},
\end{align}
where \(\fX_{\rno}\) is any \(\{\mA,\mB\}\) valued function in \(\fXS^{\costc\wedge \costc_{\rno}}\), i.e.
any function \(\fX_{\rno}:(0,\tlx]\to\{\mA,\mB\}\) satisfying
\(\costf(\fX_{\rno})=\costc\wedge\costc_{\rno}\).
\end{example}
\section{Discussion}\label{sec:conclusion}
We have applied Augustin's method to derive SPBs for two families of memoryless channels. 
For the stationary memoryless channels with convex composition constraints,
the novel observation behind Augustin's method is:
\begin{align}
\label{eq:Augustinsprinciple}
\lim\nolimits_{\rnf\to\rno} \sup\nolimits_{\mP\in\cset} \CRD{\rno}{\Wm}{\qmn{\rnf,\Wm,\cset}}{\mP}
&=\CRC{\rno}{\Wm}{\cset}
&
&\forall \rno\in(0,1). 
\end{align}
Note that the results established for the convex composition constrained stationary memoryless channels also hold for the cost constrained stationary memoryless channels because any cost constraint on a stationary memoryless channel can be expressed as a convex composition constraint, as well. 
For the non-stationary cost constrained memoryless channels we have employed
\eqref{eq:Augustinsprinciple} together with the convex conjugation techniques.
Theorem \ref{thm:exponent-cost-ologn}  improves a similar result by Augustin,
i.e. \cite[Thm. 36.6]{augustin78},
in terms of the approximation error terms. 
The prefactor of Theorem \ref{thm:exponent-cost-ologn} is of the form \(e^{-\bigo{\ln\blx}}\), 
rather than \(e^{-\bigo{\sqrt{\blx}}}\) similar to \cite[Thm. 36.6]{augustin78}. 
Also, unlike  \cite[Thm. 36.6]{augustin78}, Theorem \ref{thm:exponent-cost-ologn}
does not assume the cost functions to be bounded
and thus holds for the Gaussian models considered in 
\cite{shannon59,ebert65,ebert66,richters67}, as well.

For classical-quantum channels, the SPB was established in \cite{dalai13A}.
Following this breakthrough, there has been a reviewed interest in the 
SPB for classical-quantum channels \cite{dalai17,dalaiW17,chengH18,chengHT19}.
Augustin's method, however, has not been applied to any quantum 
information theoretic model.
Successful applications of Augustin's method will allow us to get rid of 
the stationarity, and finite input set hypotheses, at the very least.

The Augustin's variant of Gallager's bound discussed in \S \ref{sec:innerbound}
is not widely known. 
Our main aim in \S \ref{sec:innerbound} was to present this approach in its simplest form. 
Thus while bounding \(\EX{\Pe}\), we were content with passing from 
\eqref{eq:boundontheensembleaverage} to \eqref{eq:lem:LIB-General}. 
Using more careful analysis and bounding the deviation of the order \(\rno\) A-L information 
random variable, i.e. \(\ln\tfrac{\fX_{\dinp}}{\hX_{\dinp}}\) for \(\mQ=\qmn{\rno,\mP}\),
one can obtain sharper bounds similar to the ones in \cite{altugW14D,scarlettMF14,honda15}.
The random coding bound for the classical-quantum channels has been 
established in \cite{burnashevH98}. 
It is suggested in \cite[p. 5606]{dalaiW17} that the reliance of \cite{burnashevH98} 
on the codes generated with i.i.d. symbols might make it hard to modify the proof 
to the constant composition, i.e. composition constrained, case.
We think the Augustin's variant of Gallager's bound might be helpful in overcoming this issue.

As a side note, let us point out that \S\ref{sec:innerbound} and \S \ref{sec:outerbound}
imply that SPE is the reliability function for certain fading channels,
i.e. for certain channels with state, 
provided that the list decoding is allowed, even in the non-stationary case.
In particular, both the fast fading channels with 
no state information (i.e. with statistical state information)
and the fast fading channels with state information only at the receiver are
cost constrained memoryless channels. 
Thus for the channels considered in \cite{richters67} and 
the ones considered in \cite[\S 4]{telatar99} the reliability function under list decoding
is equal to the SPE.
This is the case for the models with per antenna power constraints 
considered in \cite{vu11,caoOSS16,loyka17}, as well, 
because the channels considered in \cite{vu11,caoOSS16,loyka17}
are cost constrained memoryless channels albeit
with multiple constraints.
The determination of the SPE for these channels, however,
is a separate issue, as we have noted in \S \ref{sec:examples}.

The optimal prefactor of the SPB was known for specific channels since 
the early days of the information theory, see for example 
\cite{elias55B,dobrushin62B,shannon59}.
In recent years, there has been a reviewed interest in establishing such sharp SPBs 
under various symmetry hypothesis
\cite{altugW19,altugW11,altugW14A,lancho19,lanchoODKV19A,lanchoODKV19B}.
The parametric characterization of the SPE given in 
\eqref{eq:lem:spherepacking-cc:rate} and \eqref{eq:lem:spherepacking-cc:exponent}
is used to establish such bounds for constant composition codes in \cite{nakiboglu19-ISIT}.
A refined SPB can be established in every single one of the cases considered 
in \cite{elias55B,dobrushin62B,shannon59,altugW19,altugW11,altugW14A,lancho19,lanchoODKV19A,lanchoODKV19B},
using analogous parametric characterizations as demonstrated by \cite{nakiboglu19F}.
This is one of the reasons for us to present the parametric characterizations given in
\eqref{eq:eg:SGauss-parametric-rate-alt},
\eqref{eq:eg:SGauss-parametric-spe-alt},
\eqref{eq:eg:poissonchannel-mean-parametric-alternative},
\eqref{eq:eg:poissonchannel-mean-parametric-spe-alternative},
\eqref{eq:eg:poissonchannel-bounded-parametric-alternative},
and
\eqref{eq:eg:poissonchannel-bounded-parametric-spe-alternative}.
It turns out that one can strengthen the strong converses in terms of their prefactors 
using analogous parametric characterizations 
under appropriate symmetry hypothesis, as well, see \cite{chengN20A}.

\appendix
\numberwithin{equation}{subsection}
\def\thesubsection{\Alph{subsection}}
\def\thesubsectiondis{\Alph{subsection}.} 
\subsection{Blahut's Approach}\label{sec:blahut}
In \cite{blahut74}, Blahut derives a lower bound to the error probability 
of the channel codes without using constant composition arguments.
Blahut claims that the exponential decay rate of his bound is equal to the 
SPE, see \cite[Thm. 19]{blahut74}. 
Blahut claims the equality of the aforementioned exponent and the SPE 
in other publications too, 
see  \cite[Lemma 1]{blahut76} and \cite[Thm. 10.1.4]{blahut}.
We show in the following that for the \(Z\)-channel the exponent of the 
Blahut's bound is infinite for any rate less than the 
channel capacity.
Hence, \cite[Thm. 19]{blahut74}, \cite[Lemma 1]{blahut76}, and \cite[Thm. 10.1.4]{blahut} 
are all incorrect.
More importantly, we show that even the best bound that can be 
obtained using Blahut's method is strictly inferior to the SPB
in terms of its exponential decay rate with the block length.

For any \(\Wm:\inpS\to\pmea{\outA}\) and \(\rate\in [0,\RC{1}{\Wm}]\),
let \(\GX(\rate,\Wm,\mP,\mQ)\) be 
\begin{align}
\notag
\GX(\rate,\Wm,\mP,\mQ)
&\DEF \inf\nolimits_{\Vmn{}:\CRD{1}{\Vmn{}}{\mQ}{\mP}\leq \rate} \CRD{1}{\Vmn{}}{\Wm}{\mP}
&
&\forall \mP\in\pdis{\inpS}, \mQ\in\pmea{\outA}.
\end{align}
Since \(\CRD{1}{\Vmn{}}{(\mP\!\Vmn{})}{\mP}=\RMI{1}{\mP}{\Vmn{}\!~}\) where \((\mP\!\Vmn{})=\sum_{\dinp}\mP(\dinp)\!\Vm(\dinp)\), 
the alternative expression for \(\spe{\rate,\!\Wm\!}\) given in \eqref{eq:lem:haroutunianform} implies
\begin{align}
\label{eq:blahutspherepacking}
\spe{\rate,\Wm\!}
&=\sup\nolimits_{\mP\in\pdis{\inpS}} \inf\nolimits_{\mQ\in\pmea{\outA}} \GX(\rate,\Wm,\mP,\mQ)
&
&\forall \rate\in [0,\RC{1}{\Wm}].
\end{align}
Thus the max-min inequality implies
\begin{align}
\notag
\spe{\rate,\Wm\!}
&\leq  \inf\nolimits_{\mQ\in\pmea{\outA}} \sup\nolimits_{\mP\in\pdis{\inpS}} \GX(\rate,\Wm,\mP,\mQ)
&
&\forall \rate\in [0,\RC{1}{\Wm}].
\end{align}
The initial part of the proof of \cite[Thm. 19]{blahut74} establishes the following bound 
on the error probability of the codes on a stationary product channel with the component channel
\(\Wm\) whose input set \(\inpS\) and output set \(\outS\) are finite:
\begin{align}
\label{eq:blahutbound}
\Pem{\max,\blx}
&\geq \bigo{1}e^{-\smallo{\blx}-\blx\sup\nolimits_{\mP\in\pdis{\inpS}} \GX(\rate,\Wm,\mP,\mQ)} 
&
&\forall \mQ\in \pmea{\outA}.
\end{align}
The second half of the proof of \cite[Thm. 19]{blahut74} claims that
\(\sup\nolimits_{\mP\in\pdis{\inpS}}\GX(\rate,\Wm,\mP,\qmn{\rno_{\rate},\pmn{\rate}})\) is equal to 
\(\spe{\rate,\Wm\!}\) for some \((\rno_{\rate},\pmn{\rate})\) pair satisfying the following 
equalities\footnote{Blahut mentions only 
	the first equality explicitly.} 
\begin{align}
\notag
\spe{\rate,\Wm\!}
&=\spe{\rate,\Wm\!,\pmn{\rate}},
\\
\notag
&=\tfrac{1-\rno_{\rate}}{\rno_{\rate}}(\RMI{\rno_{\rate}}{\pmn{\rate}}{\Wm}-\rate).
\end{align}
When considered together with \eqref{eq:blahutspherepacking},
the second half of the proof of \cite[Thm. 19]{blahut74} asserts that 
\begin{align}
\label{eq:blahutmaxmin}
\sup\nolimits_{\mP\in\pdis{\inpS}}\GX(\rate,\Wm,\mP,\qmn{\rno_{\rate},\pmn{\rate}})
\mathop{=}^{?}\sup\nolimits_{\mP\in\pdis{\inpS}}\inf\nolimits_{\mQ\in\pmea{\outA}}  \GX(\rate,\Wm,\mP,\mQ).
\end{align}
\cite[Lemma 1]{blahut76} and \cite[Thm. 10.1.4]{blahut} 
imply the same equality when considered together with \eqref{eq:blahutspherepacking},
as well. 
In order to disprove \eqref{eq:blahutmaxmin},  we consider 
the \(Z\)-channel, which is 
a discrete channel with the input set \(\inpS=\{1,2\}\)
and the output set \(\outS=\{\mA,\mB\}\) such that 
\begin{align}
\notag
\Wm=
\left[
\begin{array}{cc}
1 & 0 \\ 
\varepsilon & 1-\varepsilon   
\end{array} 
\right].
\end{align}
We determine the order \(\rno\) Augustin capacity and the order \(\rno\) Augustin center using the identities
\(\RC{\rno}{\Wm}=\RD{\rno}{\Wm(1)}{\qmn{\rno,\Wm}}\) 
and
\(\RC{\rno}{\Wm}=\RD{\rno}{\Wm(2)}{\qmn{\rno,\Wm}}\): 
\begin{align}
\label{eq:erasurechannelcapacity}
\RC{\rno}{\Wm}
&=\ln \left(1+\left(\tfrac{1-\varepsilon^{\rno}}{(1-\varepsilon)^{\rno}}\right)^{\frac{1}{1-\rno}}\right),
\\
\label{eq:erasurechannelcenter}
\qmn{\rno,\Wm}(\mA)
&=\tfrac{(1-\varepsilon)^{\frac{\rno}{1-\rno}}}{(1-\varepsilon)^{\frac{\rno}{1-\rno}}+(1-\varepsilon^{\rno})^{\frac{1}{1-\rno}}}.
\end{align}
Then for \(\mQ=\qmn{\rno,\Wm}\) the tilted channel \(\Wma{\rno}{\mQ}\) defined in \eqref{eq:def:tiltedchannel} is 
\begin{align}
\notag
\Wma{\rno}{\qmn{\rno,\Wm}}=
\left[
\begin{array}{cc}
1 & 0 \\ 
\varepsilon^{\rno} & 1-\varepsilon^{\rno}   
\end{array} 
\right].
\end{align}
Furthermore, one can confirm by substitution that the order \(\rno\) Augustin center \(\qmn{\rno,\Wm}\) is the fixed point of 
the order \(\rno\) Augustin operator defined in \eqref{eq:def:Aoperator} for the prior 
\(\pmn{\rno,\Wm}\) satisfying
\begin{align}
\notag
\pmn{\rno,\Wm}(1)
&=\tfrac{(1-\varepsilon)^{\frac{\rno}{1-\rno}}-\varepsilon^{\rno}(1-\varepsilon^{\rno})^{\frac{\rno}{1-\rno}}}{(1-\varepsilon)^{\frac{\rno}{1-\rno}}+(1-\varepsilon^{\rno})^{\frac{1}{1-\rno}}}.
\end{align}
Thus the order \(\rno\) Augustin center \(\qmn{\rno,\Wm}\) is equal to the order \(\rno\) Augustin mean for the prior \(\pmn{\rno,\Wm}\),
i.e. \(\qmn{\rno,\Wm}\!=\!\qmn{\rno,\pmn{\rno,\Wm}}\), by \cite[{Lemma\ref*{C-lem:information}-(\ref*{C-information:zto},\ref*{C-information:oti})}]{nakiboglu19C}
and \(\RMI{\rno}{\pmn{\rno,\Wm}}{\Wm}=\CRD{\rno}{\Wm}{\qmn{\rno,\Wm}}{\pmn{\rno,\Wm}}\).
Consequently \(\RMI{\rno}{\pmn{\rno,\Wm}}{\Wm}=\RC{\rno}{\Wm}\), as well.

Note that \(\RC{\rno}{\Wm}\) given in \eqref{eq:erasurechannelcapacity}  
is a differentiable function of \(\rno\) such that
\begin{align}
\notag
\pder{}{\rno}\RC{\rno}{\Wm}
&=\tfrac{\qmn{\rno}(\mB)}{(1-\rno)^{2}}\left(\ln\tfrac{1-\varepsilon^{\rno}}{1-\varepsilon}
+\tfrac{\varepsilon^{\rno}}{1-\varepsilon^{\rno}}\ln\tfrac{\varepsilon^{\rno}}{\varepsilon} \right),
\\
\notag
&=\tfrac{1}{(1-\rno)^{2}}\CRD{1}{\Wma{\rno}{\qmn{\rno,\Wm}}
}{\Wm}{\pmn{\rno,\Wm}}.
\end{align}
Then using first the identity \(\RMI{\rno}{\pmn{\rno,\Wm}}{\Wm}=\RC{\rno}{\Wm}\)
and then \cite[{(\ref*{C-eq:lem:information:alternative:opt})}]{nakiboglu19C},
we get
\begin{align}
\notag
\pder{}{\rno}\tfrac{1-\rno}{\rno}(\RC{\rno}{\Wm}-\rate)
&=\tfrac{1}{\rno^2}\left(\rate-\RC{\rno}{\Wm}+(1-\rno)\rno \pder{}{\rno}\RC{\rno}{\Wm}\right)
\\
\notag
&=\tfrac{1}{\rno^2}\left(\rate-
\RMI{\rno}{\pmn{\rno,\Wm}}{\Wm}+\tfrac{\rno}{1-\rno}\CRD{1}{\Wma{\rno}{\qmn{\rno,\Wm}}
}{\Wm}{\pmn{\rno,\Wm}}\right)
\\
\notag
&=\tfrac{1}{\rno^{2}}\left(\rate-\RMI{1}{\pmn{\rno,\Wm}}{\Wma{\rno}{\qmn{\rno,\Wm}}
}\right).
\end{align}
One can confirm numerically that \(\RMI{1}{\pmn{\rno,\Wm}}{\Wma{\rno}{\qmn{\rno,\Wm}}
}\) is increasing function of \(\rno\)
for any \(\varepsilon\in(0,1)\). Thus we can express the rate and the corresponding SPE in the following 
parametric form
\begin{align}
\notag
\rate(\rno)
&=\RMI{1}{\pmn{\rno,\Wm}}{\Wma{\rno}{\qmn{\rno,\Wm}}},
\\
\notag
\spe{\rate(\rno),\Wm}
&=\CRD{1}{\Wma{\rno}{\qmn{\rno,\Wm}}
}{\Wm}{\pmn{\rno,\Wm}}.
\end{align}
Thus for \(\rate=\rate(\rnf)\) we have \((\rno_{\rate},\pmn{\rate})=(\rnf,\pmn{\rnf,\Wm})\)
and \(\qmn{\rno_{\rate},\pmn{\rate}}=\qmn{\rnf,\Wm}\). Then
\begin{align}
\notag
\sup\nolimits_{\mP\in\pdis{\inpS}} \GX(\rate,\Wm,\mP,\qmn{\rno_{\rate},\pmn{\rate}})
&\geq \left.\GX(\rate,\Wm,\mP,\qmn{\rno_{\rate},\Wm})\right\vert_{\mP:\mP(1)=1}.
\\
\notag
&=\inf\nolimits_{\mV:\RD{1}{\mV}{\qmn{\rno_{\rate},\Wm}}\leq \rate}  \RD{1}{\mV}{\Wm(1)} 
\end{align}
On the other hand, \(\mV(\mB)>0\)
for all \(\mV\)  satisfying \(\RD{1}{\mV}{\qmn{\rno_{\rate},\Wm}}\leq\rate\)
because \(\RD{1}{\mV}{\qmn{\rno_{\rate},\Wm}}=\ln \tfrac{1}{\qmn{\rno_{\rate},\Wm}(\mA)}=\RC{\rno_{\rate}}{\Wm}\) 
whenever \(\mV(\mB)=0\)
and \(\RC{\rno_{\rate}}{\Wm}>\rate\).
Furthermore, \(\RD{1}{\mV}{\Wm(1)}=\infty\) whenever \(\mV(\mB)>0\). Thus
\begin{align}
\notag
\inf\nolimits_{\mV:\RD{1}{\mV}{\qmn{\rno_{\rate},\Wm}}\leq \rate}  \RD{1}{\mV}{\Wm(1)} 
&=\infty.
\end{align}
Thus the bound established in the initial part of the proof of \cite[Thm. 19]{blahut74}, i.e. \eqref{eq:blahutbound}, 
is trivial for \(\mQ=\qmn{\rno_{\rate},\pmn{\rate}}\).
Furthermore, 
\cite[Thm. 19]{blahut74}, \cite[Lemma 1]{blahut76}, and \cite[Thm. 10.1.4]{blahut} 
are all incorrect
because \(\spe{\rate,\Wm\!}<\infty\) for all \(\rate\in [0,\RC{1}{\Wm}]\) for the
\(Z\)-channels.\footnote{Recently, Yang argued that Blahut's method can be used to derive the SPB
	if the minimax equality given in \cite[(3.63)]{yang15} holds. 
	Thus as a result of our analysis we can conclude that \cite[(3.63)]{yang15} does not holds in general. 
	This fact can be derived using the absence of the minimax equality for \(\GX(\rate,\Wm,\mP,\mQ)\) 
	without relying on the reasoning in \cite{yang15}, as well.}

Using other \(\mQ\)'s one can obtain non-trivial bounds from \eqref{eq:blahutbound};
those bounds, however, do not imply the SPB either. 
If \(\mQ(\mA)<e^{-\rate}\) then \(\sup\nolimits_{\mP\in\pdis{\inpS}} \GX(\rate,\Wm,\mP,\qmn{\rno_{\rate},\pmn{\rate}})=\infty\)
as a result of the analysis presented above.
We analyze the case \(\mQ(\mA)\geq e^{-\rate}\) in the following.
\begin{align}
\notag
\sup\nolimits_{\mP\in\pdis{\inpS}} \GX(\rate,\Wm,\mP,\mQ)
&\geq 
\inf\nolimits_{\Vmn{}:\CRD{1}{\Vmn{}}{\mQ}{\pmn{\rate}}\leq \rate} \CRD{1}{\Vmn{}}{\Wm}{\pmn{\rate}}
\\
\notag
&\geq\inf\nolimits_{\Vmn{}}\CRD{1}{\Vmn{}}{\Wm}{\pmn{\rate}}+\tfrac{1-\rno_{\rate}}{\rno_{\rate}}(\CRD{1}{\Vmn{}}{\mQ}{\pmn{\rate}}-\rate)
\\
\notag
&=\tfrac{1-\rno_{\rate}}{\rno_{\rate}}(\CRD{\rno_{\rate}}{\Wm}{\mQ}{\pmn{\rate}}-\rate)
&
&\mbox{by \cite[Thm. 30]{ervenH14},}
\\
\notag
&\geq \tfrac{1-\rno_{\rate}}{\rno_{\rate}}(\RMI{\rno_{\rate}}{\pmn{\rate}}{\Wm}+\RD{\rno_{\rate}}{\qmn{\rno_{\rate},\mP_{\rate}}}{\mQ}-\rate)
&
&\mbox{by \cite[{Lemma\ref*{C-lem:information}-(\ref*{C-information:zto})}]{nakiboglu19C},}
\\
\notag
&=\spe{\rate,\Wm\!}+\tfrac{1-\rno_{\rate}}{\rno_{\rate}}\RD{\rno_{\rate}}{\qmn{\rno_{\rate},\mP_{\rate}}}{\mQ}
\\
\notag
&\geq \spe{\rate,\Wm\!}+\tfrac{1-\rno_{\rate}}{2}\lon{\mQ-\qmn{\rno_{\rate},\mP_{\rate}}}^{2}
&
&\mbox{by \cite[Thm. 31]{ervenH14}}.
\end{align} 
On the other hand, \(\qmn{\rno_{\rate},\mP_{\rate}}(\mA)<e^{-\rate}\) because \(\qmn{\rno_{\rate},\mP_{\rate}}=\qmn{\rno_{\rate},\Wm}\) and \(\RC{\rno_{\rate}}{\Wm\!}>\rate\).
Thus 
\begin{align}
\notag
\inf\nolimits_{\mQ\in\pmea{\outA}} \sup\nolimits_{\mP\in\pdis{\inpS}} \GX(\rate,\Wm,\mP,\mQ)
&\geq \spe{\rate,\Wm\!}
+2(1-\rno_{\rate})(e^{-\rate}-\qmn{\rno_{\rate},\Wm}(\mA))^{2}
\\
\notag
&=\sup\nolimits_{\mP\in\pdis{\inpS}} \inf\nolimits_{\mQ\in\pmea{\outA}} \GX(\rate,\Wm,\mP,\mQ)
+2(1-\rno_{\rate})(e^{-\rate}-\qmn{\rno_{\rate},\Wm}(\mA))^{2}.
\end{align}
Hence, it is not possible to derive the SPB using Blahut's method, 
as it is presented in \cite{blahut74}. 
When the input set is finite, one can overcome this problem by employing 
composition based expurgations. 
But that approach had been presented by Haroutunian in \cite{haroutunian68}, 
before \cite{blahut74}.

\subsection{R-G Information Measures}\label{sec:RG-informationmeasures}
The order one R-G information measures are equal to 
the corresponding order one A-L information measures by definition. 
Thus our discussion will be confined to orders other than one.
\begin{definition}\label{def:Ginformation}
	For any \(\rno\in\reals{+}\setminus\{1\}\), channel \(\Wm:\inpS\to \pmea{\outA}\) with 
	a cost function \(\costf:\inpS\to \reals{\geq0}^{\ell}\), \(\mP\in \pdis{\inpS}\),  
	and \(\lgm \in \reals{\geq0}^{\ell}\) 
	\emph{the order \(\rno\) \renyi\!\!-Gallager (R-G) information 
		for the input distribution \(\mP\) 
		and the Lagrange multiplier \(\lgm\)} is
	\begin{align}
	\notag
	\GMIL{\rno}{\mP}{\Wm}{\lgm}
	&\DEF 
	\inf\nolimits_{\mQ\in \pmea{\outA}} \RD{\rno}{\mP \mtimes  \Wm e^{\frac{1-\rno}{\rno}\lgm\cdot\costf}}{\mP\otimes \mQ}.
	\end{align}
\end{definition}
If \(\lgm\) is a vector of zeros, then the R-G information is the \renyi information.
Similar to the \renyi information, the R-G information has a closed form expression, described
in terms of a mean achieving the infimum in its definition.
\begin{definition}\label{def:Gmean}
	For any \(\rno\in\reals{+}\), channel \(\Wm:\inpS\to \pmea{\outA}\) with 
	a cost function \(\costf:\inpS\to \reals{\geq0}^{\ell}\), \(\mP\in \pdis{\inpS}\),  
	and \(\lgm \in \reals{\geq0}^{\ell}\),
	\emph{the order \(\rno\) mean measure for the input distribution \(\mP\) and
		the Lagrange multiplier \(\lgm\)} is
	\begin{align}
	\label{eq:def:Gmeanmeasure}
	\der{\mma{\rno,\mP}{\lgm}}{\rfm}
	&\DEF\left[\sum\nolimits_{\dinp} \mP(\dinp) e^{(1-\rno)\lgm \cdot \costf(\dinp)}\left(\der{\Wm(\dinp)}{\rfm}\right)^{\rno}\right]^{\frac{1}{\rno}}.
	\end{align}
	\emph{The order \(\rno\) \renyi\!\!-Gallager (R-G) mean for the input distribution \(\mP\) 
		and the Lagrange multiplier \(\lgm\)} is 
	\begin{align}
	\notag
	\qga{\rno,\mP}{\lgm}
	&\DEF \tfrac{\mma{\rno,\mP}{\lgm}}{\lon{\mma{\rno,\mP}{\lgm}}}.
	\end{align}
\end{definition}
Both \(\mma{\rno,\mP}{\lgm}\) and \(\qga{\rno,\mP}{\lgm}\) depend on the Lagrange multiplier 
\(\lgm\) for \(\rno\!\in\!\reals{+}\!\setminus\!\{1\}\).
Furthermore, one can confirm by substitution that 
\begin{align}
\notag
\RD{\rno}{\mP \mtimes  \Wm e^{\frac{1-\rno}{\rno}\lgm\cdot\costf}}{\mP\otimes \mQ}
&=\GMIL{\rno}{\mP}{\Wm}{\lgm}
+\RD{\rno}{\qga{\rno,\mP}{\lgm}}{\mQ}
&
&\rno\in\reals{+}\setminus\{1\}.
\end{align}
Then as a result of \cite[Lemma \ref*{C-lem:divergence-pinsker}]{nakiboglu19C} we have
\begin{align}
\notag
\GMIL{\rno}{\mP}{\Wm}{\lgm}
&=\RD{\rno}{\mP \mtimes  \Wm e^{\frac{1-\rno}{\rno}\lgm\cdot\costf}}{\mP\otimes \qga{\rno,\Wm}{\lgm}}
&
&
\\
\label{eq:Ginformation-neq-alternative}
&=\tfrac{\rno}{\rno-1}\ln \lon{\mma{\rno,\mP}{\lgm}}.
&
&\rno\in\reals{+}\setminus\{1\}.
\end{align}
Using the definitions of the A-L information and the R-G information 
together with the Jensen's inequality and the concavity of the natural logarithm 
function we get
\begin{align}
\label{eq:RGAL-information-zto}
\RMIL{\rno}{\mP}{\Wm}{\lgm}
&\geq \GMIL{\rno}{\mP}{\Wm}{\lgm}
&
&\rno \in (0,1]
\\
\label{eq:RGAL-information-oti}
\RMIL{\rno}{\mP}{\Wm}{\lgm}
&\leq \GMIL{\rno}{\mP}{\Wm}{\lgm}
&
&\rno \in [1,\infty).
\end{align}
\begin{definition}
	For any \(\rno\in\reals{+}\), channel  \(\Wm:\inpS\to \pmea{\outA}\) with 
	a cost function \(\costf:\inpS\to \reals{\geq0}^{\ell}\), and \(\lgm \in \reals{\geq0}^{\ell}\), 
	\emph{the order \(\rno\) \renyi\!\!-Gallager (R-G) capacity for the Lagrange multiplier \(\lgm\)} is
	\begin{align}
	\label{eq:def:Gcapacity}
	\GCL{\rno}{\Wm}{\lgm}
	&\DEF \sup\nolimits_{\mP\in \pdis{\inpS}} \GMIL{\rno}{\mP}{\Wm}{\lgm}.
	\end{align}
\end{definition} 
Although inequalities in \eqref{eq:RGAL-information-zto} and \eqref{eq:RGAL-information-oti}
are strict for most input distributions, 
as a result of \cite[Thm. \ref*{C-thm:Gminimax}]{nakiboglu19C}, we have
\begin{align}
\label{eq:thm:Gminimaxradius}
\GCL{\rno}{\Wm\!}{\lgm}
&=\RRL{\rno}{\Wm\!}{\lgm}.
\end{align}
Thus \(\GCL{\rno}{\Wm\!}{\lgm}=\RCL{\rno}{\Wm\!}{\lgm}\) by \eqref{eq:thm:Lminimaxradius}.
This is the reason why in terms of determining the optimal performance either family 
can be used.
The following lemma is, in essence, a restatement of \cite[Thm. 8]{gallager65},
which is the result that popularized the use of R-G information measures 
in cost constrained problems, see \cite{ebert65,ebert66,richters67}.

\begin{lemma}\label{lem:LGIB}
	For any \(\ell,M,L\in \integers{+}\) s.t. \(L<M\),
	\(\Wm:\inpS\to\pmea{\outA}\),
	\(\costf:\inpS\to\reals{\geq0}^{\ell}\),
	\(\mP\in\pdis{\inpS}\), \(\cinpS\subset \inpS\), 
	and \(\rno\in [\tfrac{1}{1+L},1)\) 
	there exists an \((M,L)\) channel code with an encoding function of the form
	\(\enc:\mesS\to\cinpS\) such that
	\begin{align}
	\label{eq:lem:LGIB}
	\ln \Pe 
	&\leq  \tfrac{\rno-1}{\rno} \left[\GMIL{\rno}{\mP}{\Wm}{\lgm}
	+(\inf\nolimits_{\dinp \in \cinpS}  \lgm \cdot \costf(\dinp))
	-\ln \tfrac{(M-1)e}{L}
	\right]
	-\tfrac{\ln \mP(\cinpS)}{\rno}.
	\end{align}
\end{lemma}
\begin{proof}[Proof of Lemma \ref{lem:LGIB}]
We follow the proof of Lemma \ref{lem:LAVGIB} up to \eqref{eq:lem:LIB-General}.
As result of \eqref{eq:def:samplingdistribution} and \eqref{eq:lem:LIB-General}
we have
\begin{align}
\notag
\ln \EX{\Pe} 
&\leq 
\ln \EXS{\mQ}{
	\left(\sum\nolimits_{\dinp\in\cinpS} \tfrac{\mP(\dinp)}{\mP(\cinpS)}
	e^{(1-\rno)\lgm\cdot\costf(\dinp)}
	\left(\der{\Wm(\dinp)}{\mQ}\right)^{\rno}
	\right)^{\frac{1}{\rno}}}
+\tfrac{\rno-1}{\rno}\left[(\inf\nolimits_{\dinp\in\cinpS}\lgm\cdot\costf(\dinp))
-\ln \tfrac{(M-1)e}{L}\right]
\\
\notag
&\leq 
\ln \EXS{\mQ}{
		\left(\sum\nolimits_{\dinp\in\inpS} \mP(\dinp)
		e^{(1-\rno)\lgm\cdot\costf(\dinp)}
		\left(\der{\Wm(\dinp)}{\mQ}\right)^{\rno}
		\right)^{\frac{1}{\rno}}}
	+\tfrac{\rno-1}{\rno}\left[(\inf\nolimits_{\dinp\in\cinpS}\lgm\cdot\costf(\dinp))
	-\ln \tfrac{(M-1)e}{L}\right]
-\tfrac{\ln\mP(\cinpS)}{\rno}.
	\end{align}
	Since there exists a code with \(\Pe\) less than or equal to \(\EX{\Pe}\),
the existence of a code satisfying \eqref{eq:lem:LGIB}
with an encoding function of the form \(\enc:\mesS\to\cinpS\)
follows from \eqref{eq:def:Gmeanmeasure} and 
\eqref{eq:Ginformation-neq-alternative}.
\end{proof}
\subsection{Omitted Proofs}\label{sec:omitted-proofs}
\begin{proof}[Proof of Lemma \ref{lem:spherepacking}]
	\(\spe{\rate,\Wm\!,\cset}\) is convex in \(\rate\), because the pointwise supremum 
	of a family of convex  functions is convex and 
	\(\tfrac{1-\rno}{\rno}(\CRC{\rno}{\Wm\!}{\cset}-\rate)\) is convex in \(\rate\) 
	for any \(\rno\in(0,1)\).
	\(\spe{\rate,\Wm\!,\cset}\) is nonincreasing in \(\rate\) as a result of an analogous 
	argument.
	The continuity and finiteness claims are proved while establishing \eqref{eq:lem:spherepacking}. 
	
	Recall that \(\CRC{\rno}{\Wm\!}{\cset}\) is a nondecreasing function of the order \(\rno\) by 
	\cite[Lemma \ref*{C-lem:capacityO}-(\ref*{C-capacityO-ilsc})]{nakiboglu19C}.
	\begin{itemize} 
		\item If \(\CRC{0^{_{+}}\!}{\Wm\!}{\cset}=\infty\) then \(\RC{1/2}{\Wm\!}=\infty\) 
		and \(\spe{\rate,\Wm\!,\cset}=\infty\) for all \(\rate\in\reals{\geq0}\).
		On the other hand \(\rate<\CRC{0^{_{+}}\!}{\Wm\!}{\cset}\) for all 
		\(\rate\in\reals{\geq0}\). Hence \eqref{eq:lem:spherepacking} holds
		and claims about the continuity and the finiteness of \(\spe{\rate,\Wm\!,\cset}\) are void.
		\item If \(\CRC{0^{_{+}}\!}{\Wm\!}{\cset}<\infty\) and \(\CRC{0^{_{+}}\!}{\Wm\!}{\cset}=\CRC{1}{\Wm\!}{\cset}\) 
		then \(\spe{\rate,\Wm\!,\cset}=\infty\) for all \(\rate\in[0,\CRC{1}{\Wm\!}{\cset})\)
		and  \(\spe{\rate,\Wm\!,\cset}=0\) for all \(\rate\geq\CRC{1}{\Wm\!}{\cset}\).
		Hence \eqref{eq:lem:spherepacking}
		and claims about the continuity and the finiteness of \(\spe{\rate,\Wm\!,\cset}\) 
		hold.
		
		\item If \(\CRC{0^{_{+}}\!}{\Wm\!}{\cset}<\infty\) and 
		\(\CRC{0^{_{+}}\!}{\Wm\!}{\cset}\neq\CRC{1}{\Wm\!}{\cset}\) then
		\(\spe{\rate,\Wm\!,\cset}=\infty\) for all \(\rate\in[0,\CRC{0^{_{+}}\!}{\Wm\!}{\cset})\). 
		For \(\rate\geq\CRC{0^{_{+}}\!}{\Wm\!}{\cset}\), 
		the non-negativity of \(\tfrac{1-\rno}{\rno}(\CRC{\rno}{\Wm\!}{\cset}-\rate)\) 
		imply the restrictions given in \eqref{eq:lem:spherepacking}.
		
		As a result of \eqref{eq:lem:spherepacking}, \(\spe{\rate,\Wm\!,\cset}\) is finite for all \(\rate>\lim_{\rnf\downarrow 0}\CRC{\rnf}{\Wm\!}{\cset}\).
		Thus, \(\spe{\rate,\Wm\!,\cset}\) is continuous on \((\CRC{0^{_{+}}\!}{\Wm\!}{\cset},\infty)\) 
		by  \cite[Thm. 6.3.3]{dudley}.
		In order to extent the continuity to \([\CRC{0^{_{+}}\!}{\Wm\!}{\cset},\infty)\), note that 
		the function \(\tfrac{1-\rno}{\rno}(\CRC{\rno}{\Wm\!}{\cset}-\rate)\) is decreasing and continuous in \(\rate\)
		for each \(\rno\) in \((0,1)\). 
		Thus \(\spe{\rate,\Wm\!,\cset}\) is a nonincreasing  and lower semicontinuous function of \(\rate\). 
		Hence \(\spe{\rate,\Wm\!,\cset}\) is continuous from the right and hence
		at \(\rate=\CRC{0^{_{+}}\!}{\Wm\!}{\cset}\).
	\end{itemize}
\end{proof}

\begin{proof}[Proof of Corollary \ref{cor:cost}]
	\(\tfrac{1-\rno}{\rno}\left(\CRC{\rno}{\Wmn{[1,\blx]}}{\blx\costc}
	-\ln\tfrac{M}{L}\right)\geq \spe{\ln\tfrac{M}{L},\Wmn{[1,\blx]}\!,\blx\costc}-\tfrac{1}{\blx}\)
	for an \(\rno\in(\rnt,1)\)  by	Lemma \ref{lem:spherepacking}.
	There exists a \(\mP\!\in\!\pdis{\widetilde{\inpS}_{1}^{\blx}}\) 
	of the form 
	\(\mP\!=\!\bigotimes_{\tin=1}^{\blx}\pmn{\tin}\) satisfying
	\(\EXS{\mP}{\costf_{[1,\blx]}}\!\leq\!\blx(\widetilde{\costc}\!-\!\delta)\). 
	Let \(\widehat{\costc}\DEF\costc-\tfrac{3 e \varsigma}{\blx}\).
	There exists a \(\widetilde{\mP}\!\in\!\pdis{\widetilde{\inpS}_{1}^{\blx}}\) of the form 
	\(\widetilde{\mP}\!=\!\bigotimes_{\tin=1}^{\blx}\widetilde{\pmn{\tin}}\) satisfying
	both \(\RMI{\rno}{\widetilde{\mP}}{\widetilde{\Wm}_{[1,\blx]}}
	\!\geq\!\CRC{\rno}{\widetilde{\Wm}_{[1,\blx]}}{\blx\widehat{\costc}}\!-\epsilon\)
	and \(\EXS{\widetilde{\mP}}{\costf_{[1,\blx]}}\leq \blx\widehat{\costc}\) by
	\cite[Lemmas \ref*{C-lem:information:product} and \ref*{C-lem:CCcapacityproduct}]{nakiboglu19C}.
	Let \(\widehat{\Wm}_{[1,\blx]}\!:\!\widehat{\inpS}_{1}^{\blx}\!\to\!\pmea{\outA_{1}^{\blx}}\) be 
	a  product channel 
	satisfying \(\widehat{\Wm}_{[1,\blx]}(\dinp_{1}^{\blx})=\Wmn{[1,\blx]}(\dinp_{1}^{\blx})\) 
	for all \(\dinp_{1}^{\blx}\in\widehat{\inpS}_{1}^{\blx}\)
	and \(\widehat{\inpS}_{\tin}=\supp{\widetilde{\pmn{\tin}}}\cup\supp{\pmn{\tin}}\)
	for all \(\tin\in\{1,\ldots,\blx\}\).
	Then \(\CRC{\rno}{\widehat{\Wm}_{[1,\blx]}}{\blx\widehat{\costc}}\!\geq\! \CRC{\rno}{\widetilde{\Wm}_{[1,\blx]}}{\blx\widehat{\costc}}\!-\epsilon\)
	and \(\widetilde{\costc}\leq \widehat{\costc}\) 
	by the  construction.
	There exists a \(\widehat{\mP}\in\pdis{\widehat{\inpS}_{1}^{\blx}}\) satisfying 
	both \(\RMI{\rno}{\widehat{\mP}}{\widehat{\Wm}_{[1,\blx]}}=\CRC{\rno}{\widehat{\Wm}_{[1,\blx]}}{\blx\widehat{\costc}}\)
	and \(\EXS{\widehat{\mP}}{\costf_{[1,\blx]}}\leq\blx \widehat{\costc}\)
	by \cite[{Lemma \ref*{C-lem:capacityFLB}}]{nakiboglu19C}
	because \(\widehat{\inpS}_{1}^{\blx}\) is finite.
	Furthermore, we can assume that \(\widehat{\mP}\) is of the form 
	\(\widehat{\mP}=\bigotimes_{\tin=1}^{\blx}\widehat{\pmn{\tin}}\)
	without loss of generality by
	\cite[Lemma \ref*{C-lem:information:product}]{nakiboglu19C}
	because the cost function \(\costf_{[1,\blx]}\) is additive.
	
	Note that 
	\(\EXS{\widehat{\pmn{\tin}}}{\abs{\costf_{\tin}-\EXS{\widehat{\pmn{\tin}}}{\costf_{\tin}}}^{\knd}}^{\sfrac{1}{\knd}}
	\leq \varsigma\)
	for all \(\knd\in\reals{+}\) and \(\tin\in\{1,\ldots,\blx\}\) because \(\costf_{\tin}\) is 
	positive and less than \(\varsigma\) 
	with probability one under \(\widehat{\pmn{\tin}}\).
	Then using \cite[Lemma \ref*{B-lem:berryesseenN}]{nakiboglu19B} for \(\knd=\ln \blx\) we get
	\begin{align}
	\notag
	\widehat{\mP}(\abs{\costf_{[1,\blx]}(\dinp)-\EXS{\widehat{\mP}}{\costf_{[1,\blx]}}}<3 \varsigma e)
	&\geq \tfrac{1}{2\sqrt{\blx}}.
	\end{align}
	On the other hand, there exists a
	\(\widehat{\lgm}\!=\!\lgm_{\rno,\!\widehat{\Wm}_{[1,\blx]},\!\blx\widehat{\costc}}\) 
	satisfying
	\(\CRC{\rno}{\!\widehat{\Wm}_{[1,\blx]}\!}{\blx\widehat{\costc}}\!=\!
	\RCL{\rno}{\!\widehat{\Wm}_{[1,\blx]}\!}{\widehat{\lgm}}
	\!+\!\widehat{\lgm} \blx\widehat{\costc}\) 
	by \cite[Lemma \ref*{C-lem:Lcapacity}-(\ref*{C-Lcapacity:interior})]{nakiboglu19C}
	because 
	\(\blx\widehat{\costc}\) is in the interior of 
	the feasible cost constraints \(\widehat{\Wm}_{[1,\blx]}\) by construction.
	Furthermore,
	\(\RMIL{\rno}{\widehat{\mP}}{\widehat{\Wm}_{[1,\blx]}}{\widehat{\lgm}}\!=\!\RCL{\rno}{\widehat{\Wm}_{[1,\blx]}\!}{\widehat{\lgm}}\)
	by \cite[Lemma \ref*{C-lem:Lcapacity}-(\ref*{C-Lcapacity:optimal})]{nakiboglu19C}
	because  \(\RMI{\rno}{\widehat{\mP}}{\widehat{\Wm}_{[1,\blx]}}\!=\!\CRC{\rno}{\widehat{\Wm}_{[1,\blx]}}{\blx\widehat{\costc}}\) 
	and \(\EXS{\widehat{\mP}}{\costf_{[1,\blx]}}\!\leq\!\blx \widehat{\costc}\).
	As a result
	\(\RD{\rno}{\Wmn{[1,\blx]}(\dinp_{1}^{\blx})}{\qmn{\rno,\widehat{\mP}}}-\widehat{\lgm}\costf_{[1,\blx]}(\dinp_{1}^{\blx})
	=\RCL{\rno}{\widehat{\Wm}_{[1,\blx]}\!}{\widehat{\lgm}}\)
	for all \(\dinp_{1}^{\blx}\!\in\!\supp{\widehat{\mP}}\)
	and 
	\(\widehat{\lgm}\blx\widehat{\costc}=
	\widehat{\lgm}\EXS{\widehat{\mP}}{\costf_{[1,\blx]}}\).
	Applying Lemma \ref{lem:LAVGIB} for
	\(\cinpS\!=\!\{\dinp_{1}^{\blx}\!\in\!\supp{\widehat{\mP}}:\abs{\costf_{[1,\blx]}(\dinp_{1}^{\blx})-\EXS{\widehat{\mP}}{\costf_{[1,\blx]}}}<3 \varsigma e\}\)
	we can conclude that there exists an \((M,L)\) channel code satisfying
	\begin{align}
	\notag
	\ln \Pe
	&\leq \tfrac{\rno-1}{\rno}\left[\RCL{\rno}{\widehat{\Wm}_{[1,\blx]}}{\widehat{\lgm}}
	+\inf\nolimits_{\dinp_{1}^{\blx}\in \cinpS}\widehat{\lgm}\costf_{[1,\blx]}(\dinp_{1}^{\blx})
	-\ln \tfrac{(M-1)e}{L}\right]+\tfrac{\ln 4\blx}{2\rno}
	\\
	\notag
	&\leq \tfrac{\rno-1}{\rno}\left[\CRC{\rno}{\widehat{\Wm}_{[1,\blx]}}{\blx\widehat{\costc}}
	-3 \varsigma e\widehat{\lgm}
	-\ln \tfrac{(M-1)e}{L}\right]+\tfrac{\ln 4\blx}{2\rno}.
	\end{align}
	Since \(\CRC{\rno}{\widehat{\Wm}_{[1,\blx]}}{\blx\widehat{\costc}}
	+\widehat{\lgm}\blx (\costc-\widehat{\costc})
	\geq \CRC{\rno}{\widehat{\Wm}_{[1,\blx]}}{\blx\costc}\)
	by \cite[Lemma \ref*{C-lem:CCcapacity}-(\ref*{C-CCcapacity:interior})]{nakiboglu19C} 
	we get
	\begin{align}
	\notag
	\ln \Pe
	&\leq -\spe{\ln\tfrac{M}{L},\Wmn{[1,\blx]}\!,\costc}+\tfrac{1}{\blx}+\tfrac{1-\rno}{\rno}(
	6 \varsigma e\widehat{\lgm}+2\epsilon+1)
	+\tfrac{\ln 4\blx}{2\rno}.
	\end{align}
	Then \eqref{eq:cor:cost} holds 
	because \(\tfrac{1-\rno}{\rno}\CRC{\rno}{\Wmn{[1,\blx]}}{\blx\costc}\) is nonincreasing in \(\rno\)
	by \cite[{Lemma \ref*{C-lem:capacityO}-(\ref*{C-capacityO-decreasing})}]{nakiboglu19C},
	provided that 
	\(\widehat{\lgm}\leq 
	\tfrac{\CRC{\rno}{\Wmn{[1,\blx]}}{\blx\costc}}{\blx\delta}\).
	In order to see why such a bound holds 
	first note that \(\widehat{\costc}\leq \costc\) by definition and thus
	\(\CRC{\rno}{\widehat{\Wm}_{[1,\blx]}}{\blx\widehat{\costc}} 
	\leq\CRC{\rno}{\widehat{\Wm}_{[1,\blx]}}{\blx\costc}\).
	Furthermore,
	\begin{align}
	\notag
	\CRC{\rno}{\widehat{\Wm}_{[1,\blx]}}{\blx\widehat{\costc}} 
	&=\inf\nolimits_{\lgm\geq \widehat{\lgm}}
	\RCL{\rno}{\widehat{\Wm}_{[1,\blx]}}{\lgm}+\lgm \blx \widehat{\costc} 
	\\
	\notag
	&> \blx\delta \widehat{\lgm} 
	+\inf\nolimits_{\lgm\geq \widehat{\lgm}} 
	\RCL{\rno}{\widehat{\Wm}_{[1,\blx]}}{\lgm}+\blx\lgm (\widehat{\costc}-\delta)
	\\
	\notag
	&\geq \blx \delta \widehat{\lgm}+
	\CRC{\rno}{\widehat{\Wm}_{[1,\blx]}}{\blx(\widehat{\costc}-\delta)} 
	\\
	\notag
	&\geq \blx \delta \widehat{\lgm}.
	\end{align}
~\vspace{-0.9cm}\\	
\end{proof}

\begin{proof}[Proof of Lemma \ref{lem:avspherepacking}]
	\(\CRC{\rno}{\Wm\!}{\cset}\leq \CRCI{\rno}{\Wm\!}{\cset}{\epsilon}\) by
	\eqref{eq:def:avcapacity} and	
	\cite[Lemma \ref*{C-lem:capacityO}-(\ref*{C-capacityO-ilsc},\ref*{C-capacityO-decreasing})]{nakiboglu19C}.
	Then as a result of the expressions for  \(\spe{\rate,\Wm,\cset}\) given in 
	\eqref{eq:lem:spherepacking} and the definition of \(\spa{\epsilon}{\rate,\Wm,\cset}\)
	given in \eqref{eq:def:avspherepacking} we have
	\begin{align}
	\label{eq:avspherepacking-1}
	\spe{\rate,\Wm,\cset}&\leq\spa{\epsilon}{\rate,\Wm,\cset}
	&
	&\forall \rate\in \reals{\geq 0}. 
	\end{align} 
	Let us proceed with bounding \(\spa{\epsilon}{\rate,\Wm,\cset}-\spe{\rate,\Wm,\cset}\) 
	from above for \(\rate\in[\CRC{\rnf}{\Wm\!}{\cset},\infty)\).
	\begin{align}
	\notag
	\tfrac{1-\rno}{\rno} \left(\CRCI{\rno}{\Wm\!}{\cset}{\epsilon}-\rate\right)
	&=\tfrac{1}{\epsilon}\int_{\rno-\epsilon\rno}^{\rno+\epsilon(1-\rno)} 
	\left(\tfrac{1-\rno}{\rno}\vee \tfrac{1-\rnt}{\rnt}\right) \CRC{\rnt}{\Wm\!}{\cset} \dif{\rnt}- \tfrac{1-\rno}{\rno} \rate 
	\\
	\notag
	&=\tfrac{1}{\epsilon}\tfrac{1-\rno}{\rno} \int_{\rno}^{\rno+\epsilon(1-\rno)} 
	(\CRC{\rnt}{\Wm\!}{\cset}-\rate) \dif{\rnt}
	+
	\tfrac{1}{\epsilon}\int_{\rno-\epsilon\rno}^{\rno} 
	\tfrac{1-\rnt}{\rnt}(\CRC{\rnt}{\Wm\!}{\cset} -\rate)\dif{\rnt}
	+
	\tfrac{\rate}{\epsilon}\int_{\rno-\epsilon\rno}^{\rno} 
	\tfrac{\rno-\rnt}{\rnt\rno}\dif{\rnt}
	\\
	\label{eq:avspherepacking-2}
	&\leq \tfrac{1}{\epsilon}\tfrac{1-\rno}{\rno} \int_{\rno}^{\rno+\epsilon(1-\rno)} 
	(\CRC{\rnt}{\Wm\!}{\cset}-\rate) \dif{\rnt}
	+
	\tfrac{1}{\epsilon}\int_{\rno-\epsilon\rno}^{\rno} 
	\tfrac{1-\rnt}{\rnt}(\CRC{\rnt}{\Wm\!}{\cset} -\rate)\dif{\rnt}
	+\tfrac{\epsilon}{1-\epsilon}\rate.
	\end{align} 
	We bound \(\spa{\epsilon}{\rate,\Wm,\cset}\) by bounding the expression in \eqref{eq:avspherepacking-2} separately on two intervals for \(\rno\). 
	
	In order to bound the expression in \eqref{eq:avspherepacking-2}  for \(\rno\in[\rnf,1)\), 
	we use the fact that \(\sup_{\rnt\in(0,1)}\tfrac{1-\rnt}{\rnt}(\CRC{\rnt}{\Wm\!}{\cset} -\rate)=\spe{\rate,\Wm,\cset}\).
	\begin{align}
	\notag
	\tfrac{1-\rno}{\rno} \left(\CRCI{\rno}{\Wm\!}{\cset}{\epsilon}-\rate\right)
	&\leq \tfrac{1}{\epsilon}\tfrac{1-\rno}{\rno} \int_{\rno}^{\rno+\epsilon(1-\rno)} 
	(\CRC{\rnt}{\Wm\!}{\cset}-\rate) \dif{\rnt}
	+
	\tfrac{1}{\epsilon}\int_{\rno-\epsilon\rno}^{\rno} 
	\tfrac{1-\rnt}{\rnt}(\CRC{\rnt}{\Wm\!}{\cset} -\rate)\dif{\rnt}
	+\tfrac{\epsilon}{1-\epsilon}\rate
	\\
	\notag
	&\leq \tfrac{1}{\epsilon}\tfrac{1-\rno}{\rno} \int_{\rno}^{\rno+\epsilon(1-\rno)} 
	\tfrac{\rnt}{1-\rnt}\spe{\rate,\Wm,\cset}  \dif{\rnt}
	+
	\tfrac{1}{\epsilon}\int_{\rno-\epsilon\rno}^{\rno} 
	\spe{\rate,\Wm,\cset}\dif{\rnt}
	+\tfrac{\epsilon}{1-\epsilon}\rate
	\\
	\notag
	&\leq \spe{\rate,\Wm,\cset}+
	\tfrac{\epsilon}{1-\epsilon}
	\tfrac{1-\rno}{\rno}\spe{\rate,\Wm,\cset}
	+\tfrac{\epsilon}{1-\epsilon} \rate.
	\end{align} 
	Since \((1-\rno)\spe{\rate,\Wm,\cset}+\rno \rate\leq (\rate\vee \spe{\rate,\Wm,\cset})\) for all \(\rno\in (0,1)\),
	we have
	\begin{align}
	\label{eq:avspherepacking-3new}
	\tfrac{1-\rno}{\rno} \left(\CRCI{\rno}{\Wm\!}{\cset}{\epsilon}-\rate\right)
	&\leq \spe{\rate,\Wm,\cset}+
	\tfrac{\epsilon}{1-\epsilon}
	\tfrac{\rate\vee \spe{\rate,\Wm,\cset}}{\rnf}
	&
	\rno
	&\in [\rnf,1).
	\end{align}

In order to bound the expression in \eqref{eq:avspherepacking-2}  for \(\rno\!\in\!(0,\rnf]\), 
recall that \(\CRC{\rno}{\Wm\!}{\cset}\) is nondecreasing in \(\rno\) by 
\cite[Lemma \ref*{C-lem:capacityO}-(\ref*{C-capacityO-ilsc})]{nakiboglu19C}.
Thus for any \(\rate\geq\CRC{\rnf}{\Wm\!}{\cset}\) we have  
\begin{align}
\notag
\tfrac{1-\rno}{\rno} \left(\CRCI{\rno}{\Wm\!}{\cset}{\epsilon}-\rate\right)
&\leq \tfrac{1}{\epsilon}\tfrac{1-\rno}{\rno} \int_{\rnf}^{\rno+\epsilon(1-\rno)} 
\tfrac{\rnt}{1-\rnt}\spe{\rate,\Wm,\cset}  \dif{\rnt}
\IND{\rno\in[\frac{\rnf-\epsilon}{1-\epsilon},\rnf]}
+
\tfrac{\epsilon}{1-\epsilon}\rate
\\
\notag
&\leq \tfrac{1}{\epsilon}\tfrac{\rno+\epsilon(1-\rno)}{\rno(1-\epsilon)}
\int_{\rnf}^{\rno+\epsilon(1-\rno)} \spe{\rate,\Wm,\cset}\dif{\rnt} 
\IND{\rno\in[\frac{\rnf-\epsilon}{1-\epsilon},\rnf]}
+
\tfrac{\epsilon}{1-\epsilon}\rate
\\
\notag
&= \left[
\tfrac{\rno(1-\epsilon) +2\epsilon-\rnf}{\epsilon}
-\tfrac{\rnf-\epsilon}{\rno(1-\epsilon)}
\right]
\spe{\rate,\Wm,\cset} 
\IND{\rno\in[\frac{\rnf-\epsilon}{1-\epsilon},\rnf]}
+
\tfrac{\epsilon}{1-\epsilon}\rate
\\
\label{eq:avspherepacking-4new}
&\leq (1-\rnf)\spe{\rate,\Wm,\cset}+
\tfrac{\epsilon}{1-\epsilon} \tfrac{(1-\rnf)\spe{\rate,\Wm,\cset}+\rnf\rate}{\rnf}
&
\rno
&\in (0,\rnf].
\end{align} 
\eqref{eq:lem:avspherepacking} follows from 
\eqref{eq:avspherepacking-1}, \eqref{eq:avspherepacking-3new}, 
and \eqref{eq:avspherepacking-4new}.
	
	On the other hand \(\CRC{\rno}{\Wm\!}{\cset}\) is nondecreasing and \(\tfrac{1-\rno}{\rno}\CRC{\rno}{\Wm\!}{\cset}\) 
	is nonincreasing in \(\rno\) by 
	\cite[Lemma \ref*{C-lem:capacityO}-(\ref*{C-capacityO-ilsc},\ref*{C-capacityO-decreasing})]{nakiboglu19C}.
	Then as a result of the expression for \(\spe{\rate,\Wm,\cset}\) given in 
	\eqref{eq:lem:spherepacking}, 
	we have \(\spe{\rate,\Wm,\cset}\leq \tfrac{1-\rnf}{\rnf}\rate\) for all \(\rate\in[\CRC{\rnf}{\Wm\!}{\cset},\infty)\).
	Hence,
	\begin{align}
	\label{eq:avspherepacking-6}
	\rate\vee \spe{\rate,\Wm,\cset}
	&\leq \sfrac{\rate}{\rnf}
	&
	&\forall \rate\in[\CRC{\rnf}{\Wm\!}{\cset},\infty).
	\end{align}
	\eqref{eq:lem:avspherepackingB} follows from \eqref{eq:lem:avspherepacking} and \eqref{eq:avspherepacking-6}.
	\end{proof}

\begin{proof}[Proof of Theorem \ref{thm:exponent-convex}]
	We prove Theorem \ref{thm:exponent-convex} using Lemmas \ref{lem:avspherepacking} and \ref{lem:spb-convex}.
	We are free to choose different values for \(\knd\) and \(\epsilon\)  for different 
	values of \(\blx\), provided that the  hypotheses of Lemmas \ref{lem:avspherepacking} and \ref{lem:spb-convex} 
	are satisfied. 
	
As a result of Assumption \ref{assumption:convex-ologn} 
there exists a \(K\in [1,\infty)\) and an \(\blx_{0}\in \integers{+}\) such that 
	\begin{align}
\notag
	\max\nolimits_{\tin:\tin\leq\blx} \RC{\sfrac{1}{2}}{\Uma{\cset}{(\tin)}} &\leq K \ln (\blx)
	&
	&\forall\blx\geq \blx_{0}.
	\end{align}
	Let \(\knd_{\blx}\) be \(\knd_{\blx}=K \ln(1+\blx)\). Then 
	\begin{align}
	\label{eq:exponent-convex-2}
	\gamma_{\blx}
	&\leq 40 (K+1)\ln(1+\blx)
	&
	&\forall\blx\geq \blx_{0}.
	\end{align}
	\(\CRC{\rno}{\Wm\!}{\cset}\) is nondecreasing in \(\rno\) by  
	\cite[Lemma \ref*{C-lem:capacityO}-(\ref*{C-capacityO-ilsc})]{nakiboglu19C}
	and \(\tfrac{1-\rno}{\rno}\CRC{\rno}{\Wm\!}{\cset} \) is nonincreasing in \(\rno\) on \((0,1)\) by 
	\cite[Lemma \ref*{C-lem:capacityO}-(\ref*{C-capacityO-decreasing})]{nakiboglu19C}.
	Thus, we can bound \(\CRCI{\rno}{\Wm\!}{\cset}{\epsilon}\) using its definition given 
	in \eqref{eq:def:avcapacity}:
	\begin{align}
	\notag
	\CRCI{\rno}{\Wm\!}{\cset}{\epsilon}
&=\tfrac{1}{\epsilon}\int_{\rno-\epsilon\rno}^{\rno+\epsilon(1-\rno)} 
\left[ 1\vee \left(\tfrac{\rno}{1-\rno}\tfrac{1-\rnt}{\rnt}\right)\right] \CRC{\rnt}{\Wm\!}{\cset}  \dif{\rnt}
\\
\notag
&\leq \tfrac{1}{\epsilon}\int_{\rno-\epsilon\rno}^{\rno+\epsilon(1-\rno)} 
\left[ 1\vee \left(\tfrac{\rno}{1-\rno}\tfrac{1-\rnt}{\rnt}\right)\right] 
\left[ 1\vee \left(\tfrac{1-\rno}{\rno}\tfrac{\rnt}{1-\rnt}\right)\right]
\CRC{\rno}{\Wm\!}{\cset}  \dif{\rnt}
\\
\notag
&=\tfrac{\CRC{\rno}{\Wm\!}{\cset}}{\epsilon}\int_{\rno-\epsilon\rno}^{\rno+\epsilon(1-\rno)} 
\left[ \left(\tfrac{\rno}{1-\rno}\tfrac{1-\rnt}{\rnt}\right)\vee \left(\tfrac{1-\rno}{\rno}\tfrac{\rnt}{1-\rnt}\right)\right]
\dif{\rnt}
\\
\notag
&\leq \left(1+\tfrac{\epsilon}{1-\epsilon}\tfrac{\rno^{2}+(1-\rno)^{2}}{\rno(1-\rno)}\right)
	\CRC{\rno}{\Wm\!}{\cset}.
	\end{align}
Then for \(\epsilon_{\blx}=\sfrac{1}{\blx}\), 	\eqref{eq:exponent-convex-2} imply that
	\begin{align}
	\label{eq:exponent-convex-3}
	\blx \CRCI{\rnf}{\Wm\!}{\cset}{\epsilon}+\tfrac{\gamma_{\blx}}{1-\rnf}
	+\ln\tfrac{8 e^{3}\blx^{1.5} }{\epsilon}
	&\leq \blx \CRC{\rnf}{\Wm\!}{\cset}+\tfrac{\blx}{\blx-1}\tfrac{\CRC{\rnf}{\Wm\!}{\cset}}{\rnf(1-\rnf)}
	+\tfrac{40(K+1)\ln(1+\blx)}{1-\rnf}+3\ln (2 e \blx)
		&
	&\forall\blx\geq \blx_{0}.
	\end{align}
Thus the hypothesis of Theorem \ref{thm:exponent-convex} implies 
the hypothesis of Lemma \ref{lem:spb-convex} for all \(\blx\) large enough. 
Consequently, for all \(\blx\) large enough Lemma \ref{lem:spb-convex} and
\eqref{eq:exponent-convex-2} implies that
	\begin{align}
	\label{eq:exponent-convex-4}
	\Pe
	&\geq  \left(\tfrac{(1+\blx)^{-80(K+1)}}{8 e^{3}\blx^{2.5}}\right)^{\sfrac{1}{\rnf}}  
	e^{-\blx \spa{\epsilon_{\blx}}{\frac{1}{\blx}\ln\frac{M_{\blx}}{L_{\blx}},\Wm\!,\cset}}.
	\end{align}
	On the other hand,  Lemma \ref{lem:avspherepacking}, \eqref{eq:thm:exponent-convex-hypothesis},
	and the monotonicity of \(\CRC{\rno}{\Wm\!}{\cset}\) in \(\rno\) imply that for \(\blx\) 
	large enough
	\begin{align}
	\label{eq:exponent-convex-5}
	\spa{\epsilon_{\blx}}{\tfrac{1}{\blx}\ln\tfrac{M_{\blx}}{L_{\blx}},\Wm\!,\cset}
	&\leq\spe{\tfrac{1}{\blx}\ln\tfrac{M_{\blx}}{L_{\blx}},\Wm\!,\cset}
	+\tfrac{\CRC{\rno_{1}}{\Wm\!}{\cset}}{(\blx-1)\rnf^{2}}.
	\end{align}
	\eqref{eq:thm:exponent-convex} follows from \eqref{eq:exponent-convex-4} and \eqref{eq:exponent-convex-5}.
\end{proof}

\begin{proof}[Proof of Theorem \ref{thm:exponent-cost-ologn}]
	We prove Theorem \ref{thm:exponent-cost-ologn} using Lemmas \ref{lem:avspherepacking} and \ref{lem:spb-cost-ologn}.
	We are free to choose different values for \(\knd\) and \(\epsilon\)  for different 
	values of \(\blx\), provided that the  hypotheses of Lemmas \ref{lem:avspherepacking} and \ref{lem:spb-cost-ologn} 
	are satisfied. 
	
As a result of Assumption \ref{assumption:cost-ologn}
there exists a \(K\in [1,\infty)\) and an \(\blx_{0}\in \integers{+}\) such that 
	\begin{align}
	\label{eq:exponent-cost-ologn-1}
	\max\nolimits_{\tin:\tin\leq\blx} \CRC{\sfrac{1}{2}}{\Wmn{\tin}}{\blx \costc} &\leq K \ln \blx
	&
	&\forall\blx\geq \blx_{0}.
	\end{align}
	Let \(\knd_{\blx}\) be \(\knd_{\blx}=K \ln(1+\blx)\). Then
	\begin{align}
	\label{eq:exponent-cost-ologn-2}
	\gamma_{\blx}
	&\leq 40 (K+1)  \ln(1+\blx)
	&
	&\forall\blx\geq \blx_{0}.
	\end{align}
	\(\CRC{\rno}{\Wm\!}{\costc}\) is nondecreasing in \(\rno\) by  
	\cite[Lemma \ref*{C-lem:capacityO}-(\ref*{C-capacityO-ilsc})]{nakiboglu19C}
	and \(\tfrac{1-\rno}{\rno}\CRC{\rno}{\Wm\!}{\costc}\) is nonincreasing in \(\rno\) on \((0,1)\) by 
	\cite[Lemma \ref*{C-lem:capacityO}-(\ref*{C-capacityO-decreasing})]{nakiboglu19C}.
	Thus, we can bound \(\CRCI{\rno}{\Wm\!}{\costc}{\epsilon}\) using its definition given 
	in \eqref{eq:def:avcapacity}:
	\begin{align}
	\notag
	\CRCI{\rno}{\Wmn{[1,\blx]}}{\blx \costc}{\epsilon}
	&\leq \left(1+\tfrac{\epsilon}{1-\epsilon}\tfrac{\rno^{2}+(1-\rno)^{2}}{\rno(1-\rno)}\right)
	\CRC{\rno}{\Wmn{[1,\blx]}}{\blx \costc}.
	\end{align}
	Then for \(\epsilon_{\blx}=\tfrac{1}{\blx}\),
	\eqref{eq:exponent-cost-ologn-1}, \eqref{eq:exponent-cost-ologn-2}, and 
	\cite[Lemmas \ref*{C-lem:capacityO}-(\ref*{C-capacityO-ilsc},\ref*{C-capacityO-decreasing}), 
	\ref*{C-lem:CCcapacity}-(\ref*{C-CCcapacity:function}), \ref*{C-lem:CCcapacityproduct}]{nakiboglu19C}
	imply that
	\begin{align}
	\label{eq:exponent-cost-ologn-3}
\CRCI{\rno}{\Wmn{[1,\blx]}}{\blx \costc}{\epsilon}
+\tfrac{\gamma_{\blx}}{1-\rnf}
+\ln\tfrac{8e^{3}\blx^{1.5} }{\epsilon}
&\!\leq\!
\CRC{\rno}{\Wmn{[1,\blx]}}{\blx \costc}
\!+\!\tfrac{\blx}{\blx-1}\tfrac{K\ln (\blx)}{\rnf(1-\rnf)}
\!+\!\tfrac{40 (K+1)\ln(1+\blx)}{1-\rnf}
\!+\!3\ln (2 e \blx)
&
&\forall\blx\geq \blx_{0}.
\end{align}
Thus the hypothesis of Theorem \ref{thm:exponent-cost-ologn} implies 
the hypothesis of Lemma \ref{lem:spb-cost-ologn} for all \(\blx\) large enough. 
Consequently, for all \(\blx\) large enough Lemma \ref{lem:spb-cost-ologn} and
\eqref{eq:exponent-cost-ologn-2} implies that
	\begin{align}
	\label{eq:exponent-cost-ologn-4}
	\Pe
	&\geq  \left(\tfrac{(1+\blx)^{-80(K+1)}}{8 e^{3}\blx^{2.5}}\right)^{\sfrac{1}{\rnf}}  
	e^{-\spa{\epsilon_{\blx}}{\ln \frac{M_{\blx}}{L_{\blx}},\Wmn{[1,\blx]},\blx \costc}}.
	\end{align}
	On the other hand,  Lemma \ref{lem:avspherepacking}, \eqref{eq:thm:exponent-cost-ologn-hypothesis},
	and the monotonicity of \(\CRC{\rno}{\Wm\!}{\costc}\) in \(\rno\) imply that for \(\blx\) large enough
	\begin{align}
	\label{eq:exponent-cost-ologn-5}
	\spa{\epsilon_{\blx}}{\ln \tfrac{M_{\blx}}{L_{\blx}},\Wmn{[1,\blx]},\blx \costc}
	&\leq
	\spe{\ln \tfrac{M_{\blx}}{L_{\blx}},\Wmn{[1,\blx]},\blx \costc}
	+\tfrac{\CRC{\rno_{1}}{\Wmn{[1,\blx]}}{\blx \costc}}{(\blx-1)\rnf^{2}}.
	\end{align}
Note that \(\CRC{\rno}{\Wmn{[1,\blx]}}{\blx \costc}\) is nondecreasing in \(\rno\) by  
\cite[Lemma \ref*{C-lem:capacityO}-(\ref*{C-capacityO-ilsc})]{nakiboglu19C}
and \(\tfrac{1-\rno}{\rno}\CRC{\rno}{\Wmn{[1,\blx]}}{\blx \costc}\) is nonincreasing in \(\rno\) on \((0,1)\) by 
\cite[Lemma \ref*{C-lem:capacityO}-(\ref*{C-capacityO-decreasing})]{nakiboglu19C}.
Thus
\begin{align}
\label{eq:exponent-cost-ologn-6}
\CRC{\rno_{1}}{\Wmn{[1,\blx]}}{\blx \costc}
\leq (\tfrac{\rno_{1}}{1-\rno_{1}} \vee 1)\CRC{1/2}{\Wmn{[1,\blx]}}{\blx \costc}.
\end{align}
Since \(\CRC{\rno}{\Wm\!}{\costc}\) is non-decreasing in \(\costc\) by 
\cite[Lemmas \ref*{C-lem:CCcapacity}-(\ref*{C-CCcapacity:function})]{nakiboglu19C},
\eqref{eq:exponent-cost-ologn-1} and \cite[Lemma \ref*{C-lem:CCcapacityproduct}]{nakiboglu19C}
imply for all  \(\blx\) large enough
\begin{align}
\label{eq:exponent-cost-ologn-7}
\CRC{1/2}{\Wmn{[1,\blx]}}{\blx \costc}
\leq  K \blx  \ln \blx.
\end{align}
	For \(\blx\)  large enough 
\eqref{eq:exponent-cost-ologn-5}, \eqref{eq:exponent-cost-ologn-6}, and \eqref{eq:exponent-cost-ologn-7} imply
\begin{align}
\label{eq:exponent-cost-ologn-8}
\spa{\epsilon_{\blx}}{\ln \tfrac{M_{\blx}}{L_{\blx}},\Wmn{[1,\blx]},\blx \costc}
&\leq
\spe{\ln \tfrac{M_{\blx}}{L_{\blx}},\Wmn{[1,\blx]},\blx \costc}
+\tfrac{2}{\rnf^{2}} (\tfrac{\rno_{1}}{1-\rno_{1}} \vee 1) K \ln \blx.
\end{align}
\eqref{eq:thm:exponent-cost-ologn} follows from \eqref{eq:exponent-cost-ologn-4} and \eqref{eq:exponent-cost-ologn-8}.
\end{proof}
\section*{Acknowledgment}
The author would like to thank Fatma Nakibo\u{g}lu and Mehmet Nakibo\u{g}lu for 
their hospitality; this work would not have been possible without it.  
The author would like to thank  Marco Dalai for informing him about Fano's implicit 
assertion of  the fixed point property in  \cite{fano},
Gonzalo Vazquez-Vilar for pointing out Poltyrev's paper \cite{poltyrev82} on 
the random coding bound and for his suggestions on the manuscript,
Wei Yang for pointing out \cite[(3.63)]{yang15} and its relation to Blahut's approach,
and the reviewer his suggestions on the manuscript.
\bibliographystyle{unsrt} 
\newcommand{\noopsort}[1]{} \newcommand{\printfirst}[2]{#1}
  \newcommand{\singleletter}[1]{#1} \newcommand{\switchargs}[2]{#2#1}

\end{document}